\documentclass[a4paper,10pt]{article}
\usepackage[utf8]{inputenc}
\usepackage[english]{babel}
\usepackage{csquotes}
\usepackage{indentfirst}
\usepackage{vmargin}
\usepackage{comment}
\usepackage{tensor}
\usepackage{enumitem}
\usepackage{amsmath,amssymb,amsthm,mathrsfs,amsfonts,dsfont}
\usepackage{mathtools}
\usepackage{xcolor}
\usepackage{graphicx}
\usepackage[small, skip=10pt]{caption}
\usepackage{tikz,pgfplots}
\pgfplotsset{compat=1.18}
\usepgfplotslibrary{fillbetween}
\usetikzlibrary{calc,arrows,arrows.meta,cd,shapes.geometric,positioning,through,decorations.markings,decorations.text}
\tikzset{>=latex}
\usepackage[mode=buildnew]{standalone}
\usepackage[unicode]{hyperref}
\usepackage[style=phys, 
    citestyle=numeric-comp,
    biblabel=brackets,
    eprint=true,
    sorting=none]{biblatex}
\DefineBibliographyStrings{english}{bibliography = {References}}
\newcommand{\vperp}{W}

\newcommand{\bb}{\mathbb}
\newcommand{\GL}{\text{GL}}
\newcommand{\gl}{\mathfrak{gl}}
\newcommand{\gm}{\mathfrak{g}}
\theoremstyle{plain}
\newtheorem{lemma}{Lemma}
\newtheorem{proposition}[lemma]{Proposition}
\newtheorem{theorem}[lemma]{Theorem}

\theoremstyle{definition}

\newtheorem*{remark}{Remark}
\newcommand{\SO}{\mathop{\rm SO}}
\newcommand{\RR}{\mathbb{R}}
\renewcommand{\d}{\partial}
\newcommand{\ISO}{\operatorname{ISO}}

\newcommand{\g}{\mathfrak{g}}
\newcommand{\ZZ}{\mathbb{Z}}
\newcommand{\ann}{\operatorname{ann}}
\newcommand{\End}{\operatorname{End}}
\newcommand{\Hom}{\operatorname{Hom}}
\newcommand{\coker}{\operatorname{coker}}
\newcommand{\ad}{\operatorname{ad}}
\newcommand{\im}{\operatorname{im}}
\newcommand{\gr}{\operatorname{gr}}
\newcommand{\into}{\hookrightarrow}
\newcommand{\onto}{\twoheadrightarrow}
\newcommand{\isom}{\xrightarrow{\cong}}
\newcommand{\so}{\mathfrak{so}}
\newcommand{\V}{\mathbb{V}}
\newcommand{\rank}{\operatorname{rank}}
\newcommand{\cK}{\mathscr{K}}
\newcommand{\cT}{\mathscr{T}}
\newcommand{\Ort}{\operatorname{O}}
\newcommand{\NN}{\mathbb{N}}
\newcommand{\EE}{\mathbb{E}}
\newcommand{\FF}{\mathbb{F}}
\newcommand{\id}{\operatorname{id}}
\newcommand{\pr}{\operatorname{pr}}
\newcommand{\Ad}{\operatorname{Ad}}
\newcommand{\tr}{\operatorname{tr}}

\newcommand{\eX}{\mathcal{X}}
\newcommand{\eL}{\mathcal{L}}
\definecolor{darkergreen}{rgb}{0.0, 0.5, 0.0}
\newcommand{\JR}[1]{\textcolor{black}{#1}}
\newcommand{\IR}[1]{\textcolor{black}{#1}}

\allowdisplaybreaks[0]
\numberwithin{equation}{section}
\hypersetup{
        linkbordercolor={red}
}

\hypersetup{
unicode=false,          
pdftoolbar=true,        
pdfmenubar=true,        
pdffitwindow=false,     
pdfstartview={FitH},    
pdftitle={Generalized Non-Lorentzian Geometries},    
pdfauthor={Eric Bergshoeff, José Figueroa-O'Farrill, Kevin van Helden, Jan Rosseel, Iisakki Rotko, Tonnis ter Veldhuis},     
pdfsubject={},   
pdfnewwindow=true,      
colorlinks=false,       
linkcolor=red,          
citecolor=red,        
filecolor=red,      
urlcolor=red           
linkbordercolor={red},
citebordercolor={red},
urlbordercolor={red}
}

\title{\bf $p\;$-brane Galilean and Carrollian Geometries and Gravities}
\author{}
\date{}

\begin{document}

\begin{flushright}
\small
\today\\
\normalsize
\end{flushright}
{\let\newpage\relax\maketitle}
\maketitle
\def\equationautorefname~#1\null{eq.~(#1)\null}
\def\tableautorefname~#1\null{tab.~(#1)\null}

\vspace{0.8cm}

\begin{center}
\renewcommand{\thefootnote}{\alph{footnote}}
{\sl\large E.~A.~Bergshoeff$^{~1}$}\footnote{Email: {\tt e.a.bergshoeff[at]rug.nl}},
{\sl\large J.~M.~Figueroa-O'Farrill$^{~2}$}\footnote{Email: {\tt j.m.figueroa[at]ed.ac.uk}},
{\sl\large K.~S.~van~Helden$^{~1,3}$}\footnote{Email: {\tt k.s.van.helden[at]rug.nl}},\\[.1truecm]
{\sl\large J.~Rosseel$^{~4}$}\footnote{Email: {\tt rosseelj[at]gmail.com}},
{\sl\large I.~J.~Rotko$^{~1}$}\footnote{Email: {\tt i.j.rotko[at]student.rug.nl}},
{\sl\large T.~ter Veldhuis$^{~1}$}\footnote{Email: {\tt terveldhuis@macalester.edu}}\footnote{On leave of absence from Macalester College, Saint Paul (USA).}
\setcounter{footnote}{0}
\renewcommand{\thefootnote}{\arabic{footnote}}

\vspace{0.5cm}
${}^1${\it Van Swinderen Institute, University of Groningen\\
Nijenborgh 4, 9747 AG Groningen, The Netherlands}\\
\vskip .2truecm
${}^2${\it  School of Mathematics and Maxwell Institute,\\
The University of Edinburgh, Edinburgh EH9 3FD, Scotland, UK}\\
\vskip .2truecm
${}^3${\it Bernoulli Institute, University of Groningen\\
Nijenborgh 9, 9747 AG Groningen, The Netherlands}\\
\vskip .2truecm
${}^4${\it Division of Theoretical Physics, Rudjer Bo\v{s}kovi\'c Institute,\\
Bijeni\v{c}ka 54, 10000 Zagreb, Croatia}\\

\vspace{1.8cm}

{\bf Abstract}
\end{center}
\begin{quotation}
{\small \noindent
  We study $D$-dimensional $p$-brane Galilean geometries via the
  intrinsic torsion of their adapted connections.  These
  non-Lorentzian geometries are examples of $G$-structures whose
  characteristic tensors consist of two degenerate ``metrics'' of ranks $(p+1)$ and
  $(D-p-1)$.  We carry out the analysis in two different ways.  In one way,
  inspired by Cartan geometry, we analyse in detail the space of intrinsic
  torsions (technically, the cokernel of a Spencer differential) as a
  representation of $G$, exhibiting for generic $(p,D)$ five classes of
  such geometries, which we then proceed to interpret geometrically.
  We show how to re-interpret this classification in terms of
  ($D-p-2$)-brane Carrollian geometries.  The same result is recovered
  by methods inspired by similar results in the physics literature:
  namely by studying how far an adapted connection can be determined
  by the characteristic tensors and by studying which components of
  the torsion tensor do not depend on the connection.  As an
  application, we derive a gravity theory with underlying $p$-brane
  Galilean geometry as a non-relativistic limit of Einstein--Hilbert
  gravity and discuss how it gives a gravitational realisation of some
  of the intrinsic torsion constraints found in this paper. Our
  results also have implications for gravity theories with an
  underlying ($D-p-2$)-brane Carrollian geometry.}
\end{quotation}

\newpage

\tableofcontents

\section{Introduction}
\label{sec:introduction}

During the last hundred years or so, our description of spacetime,
born out of Einstein's theories of Special and General Relativity, has
been based on Lorentzian geometry, a variant of Riemannian
geometry in which the metric tensor is no longer positive-definite.
This seemingly insignificant change has far-reaching consequences and
makes Lorentzian geometry in many ways richer than Riemannian
geometry, cf. the richer variety of submanifolds that a Lorentzian
manifold can admit.  Nevertheless, for all their differences,
Riemannian and Lorentzian geometry share one important feature: the
so-called Fundamental Theorem of Riemannian Geometry, which states
that there exists a unique metric-compatible torsion-free affine
connection: the Levi-Civita connection.  Properly understood, this
theorem contains two separate statements. The first, which follows
paying close attention to the usual proof of the theorem, is that a
metric-compatible affine connection is uniquely determined by its
torsion.  The second says that any metric-compatible connection,
possibly with torsion, can be \emph{uniquely} modified (by the addition of a
\emph{unique} contorsion tensor) so that the resulting connection is
still metric-compatible but now also torsion-free.

There are two ways we can understand this result. This dichotomy is
reflected in the two seemingly different approaches followed in this
paper to the study of non-Lorentzian $p$-brane geometries. The
approach we normally teach undergraduate students in a differential
geometry course is by explicit construction of the Levi-Civita
connection -- we actually solve for the Levi-Civita connection in terms 
of the metric tensor. This can be done using notation with or without indices.
The version with indices results in the famous expression for the
Christoffel symbols of the metric $g$:
\begin{equation}
  \label{eq:Christoffel}
  \Gamma_{\lambda\rho}{}^\mu = \tfrac12 g^{\mu\nu} \left( \d_\lambda
    g_{\nu\rho} + \d_\rho g_{\lambda\nu} - \d_\nu g_{\lambda\rho} \right),
\end{equation}
where as usual $g_{\mu\nu}$ and $g^{\mu\nu}$ are the components of the
metric tensor and its inverse relative to local coordinates $x^\mu$.
The Christoffel symbols enter in the definition of an affine
connection via $\nabla_\mu \d_\nu = \Gamma_{\mu\nu}{}^\rho \d_\rho$.
The version without indices results in the equally famous Koszul
formula and says that
\begin{equation}
  \label{eq:Koszul}
  2 g(\nabla_X Y, Z) = X g(Y,Z) + Y g(X,Z) - Z g(X,Y) + g(Z,[X,Y]) -
  g(Y,[X,Z]) - g(X,[Y,Z]),
\end{equation}
where $X,Y,Z$ are any three vector fields on the manifold. Both
expressions have their virtues and which one one prefers is largely a question of
aesthetics.

There is, however, a different approach to prove the Fundamental Theorem
of Riemannian Geometry, which does not require deriving the explicit
expression for the connection and which is normally not taught to
undergraduates. This approach is based on Cartan geometry, where we work locally relative to a chosen local (pseudo)orthonormal frame $e_a$ with canonically dual
coframe $\theta^a$. The coframe defines a fibrewise isomorphism of
every tangent space with $\RR^D$, with $D$ the dimension of the
manifold. The metric tensor has a local expression
$g = \eta_{ab}\theta^a\theta^b$ where $\eta_{ab} = g(e_a,e_b)$ is
constant and defines an inner product on $\RR^D$. Let $G$ denote (the
identity component of) the subgroup of $\GL(D,\RR)$ which preserves
the inner product $\eta_{ab}$: in the Riemannian case $G$ will be the rotation group $\SO(D)$, in the Lorentzian case $G$ is the proper orthochronous Lorentz group
$\SO(D-1,1)_0$, et cetera. A metric-compatible
connection is locally given by a one-form $\omega$ with values in the
Lie algebra $\g$ of $G$. The torsion tensor of the connection is
described in this language by the torsion $2$-form $\Theta$ defined
via Cartan's first structure equation:
\begin{equation}
  \label{eq:fse}
  \Theta^a = d\theta^a + \omega^a{}_b \wedge \theta^b,
\end{equation}
where the condition that $\omega$ is $\g$-valued simply says that
$\omega_{ab} := \eta_{ac}\omega^c{}_b = - \omega_{ba}$.  Let $\omega'$
be the connection one-form of another metric-compatible connection.
Then $\kappa:= \omega' - \omega$ is a one-form with values in $\g$,
but unlike $\omega$ or $\omega'$, which transform as gauge fields under
a change of local (pseudo)orthonormal frame, $\kappa$ transforms
tensorially.  It then follows from the structure
equation~\eqref{eq:fse} that the torsion $2$-forms of $\omega$ and
$\omega'$ are related via
\begin{equation}
  \Theta'{}^a = \Theta^a + \kappa^a{}_b \wedge \theta^b.
\end{equation}
The linear map sending $\omega' - \omega$ to $\Theta' - \Theta$ is an
instance of a so-called Spencer differential, and it sends $\kappa$,
which is a local $1$-form with values in $\g$ to $\kappa^a{}_b \wedge
\theta^b$ which is a local $2$-form with values in the so-called fake
tangent bundle.  This is induced from a linear-algebraic map
\begin{equation}
  \label{eq:spencer}
  \d : \Hom(\RR^D,\g) \to \Hom(\wedge^2\RR^D,\RR^D)
\end{equation}
where
\begin{equation}
  \label{eq:specer-explicit}
  (\d\kappa)(v\wedge w) = \kappa(v)w - \kappa(w)v.
\end{equation}
The kernel of $\d$ consists of those $\kappa$ which obey $\kappa^a{}_b
\wedge \theta^b = 0$; that is, those $\kappa$ which do not change the
torsion $2$-form.  The image of $\d$ are the possible changes of the
torsion due to changes in the connection.  In this language, the
Fundamental Theorem of Riemannian Geometry simply says that $\d$ is an
isomorphism and, as intimated earlier, it consists of two
statements\footnote{In this case, the two vector spaces in
  Equation~\eqref{eq:spencer} have the same dimension, so injectivity
  implies surjectivity (and viceversa) by the Rank Theorem.}: $\d$ is
surjective, so that given any connection $\omega$ we can find some
$\kappa$ so that $\Theta' = 0$ (existence of a torsion-free
metric-compatible connection) and $\d$ is injective, so that this
$\kappa$ is unique.  In other words, there exists a unique
metric-compatible torsion-free connection.

Comparing this latter proof with the standard one, we see that in some
sense the Cartan-theoretic proof boils down to linear algebra, with
the heavy lifting from the linear algebra to the differential
geometry being done by Cartan's theory of moving frames, which is the
traditional name for what is now called the theory of $G$-structures,
itself a special case of the notion of a Cartan geometry.

It is by now well-understood that Lorentzian geometry stops being an
adequate description of spacetimes in many limiting
scenarios. For instance, this is the case for the so-called non- or
ultra-relativistic limits, where the local speed of light becomes
infinite or zero, respectively.  In such limits, the metric
degenerates and the resulting geometries, while no longer being
Lorentzian, are nevertheless still instances of $G$-structures, only
with a different choice of subgroup $G$ of $\GL(D,\RR)$.  (In the case
of geometries obtained via limits from Lorentzian geometry, the new
$G$ is a contraction of the Lorentz group.)  Perhaps the
simplest examples of these non-Lorentzian geometries are Galilean,
which is the relevant geometry for the non-relativistic limit, and
Carrollian, which is the relevant geometry for the ultra-relativistic
limit.  The relevant subgroups $G \subset \GL(D,\RR)$ in these two
cases are abstractly isomorphic to the Euclidean group $\ISO(D-1)$ and
related by transposition inside $\GL(D,\RR)$, but crucially they are
not conjugate in $\GL(D,\RR)$ and hence describe truly different
geometries.\footnote{This serves to emphasise the fact that the $G$ in
  a $G$-structure is not an abstract group but a group $G$ together
  with a faithful representation in $\RR^D$ or, equivalently, a
  subgroup of $\GL(D,\RR)$, analogously to the fact that the holonomy
  group of a connection should more properly be referred to as the
  holonomy representation.}  Non-Lorentzian geometries are defined by
characteristic tensor fields other than a metric.  These tensor fields
are constructed out of $G$-invariant tensors by the use of the moving
frames and coframes in precise analogy to how the metric tensor in
(pseudo)Riemannian geometry is constructed as
$g = \eta_{ab}\theta^a\theta^b$ with $\eta_{ab}$ an invariant tensor
of $G$.  (In fact, that's how $G$ was defined.)

As an example, let us consider Galilean geometry. In the same way that
a Lorentzian manifold is said to look locally like Minkowski
spacetime, a Galilean manifold $M$ looks locally like Galilei
spacetime.  The characteristic tensors replacing the Lorentzian metric
are now a clock one-form $\tau \in \Omega^1(M)$, whose square may be
interpreted as a very degenerate rank-$1$ ``metric'' on $M$, and a
ruler, which is a positive-definite metric on the distribution
annihilated by $\tau$.  The ruler so described is not actually a
tensor field, being a section of a quotient of a tensor bundle, so
instead we replace it by a ``spatial co-metric'' $h$, a
positive-semidefinite metric on one-forms which has rank $D-1$ (so it
is degenerate) and obeys $h(\tau,\alpha)=0$ for all one-forms
$\alpha$.  The triple $(M,\tau,h)$ is sometimes called a weak
Galilean structure, reserving the unadorned name Galilean structure
for a quartet $(M,\tau,h,\nabla)$ with $\nabla$ an affine
connection which is adapted to the $G$-structure and hence compatible
with the characteristic tensors in that $\nabla\tau = 0$ and
$\nabla h = 0$. But therein lies one crucial difference between
Lorentzian and non-Lorentzian (e.g., Galilean) geometries: there is no
non-Lorentzian analogue of the Fundamental Theorem of Riemannian
Geometry.

The way this fact manifests itself in our two approaches to the proof
of the Fundamental Theorem of Riemannian Geometry is as follows.  In
the first approach, one simply cannot solve for the compatible
connection in terms of $\tau$ and $h$.  There is always an
ambiguity and hence a choice to be made.  In the second approach this
follows from the linear algebraic fact that the Spencer differential
is no longer an isomorphism: it may have nontrivial kernel and/or may
fail to be surjective.  In the case of Galilean and Carrollian
$G$-structures, the two vector spaces in Equation~\eqref{eq:spencer}
have the same dimension, so having a nontrivial kernel implies that it
is not surjective and vice versa.  In the first approach, the lack of
surjectivity manifests itself in components of the torsion which do
not actually depend on the connection one-form.   Such components
signal the existence of nontrivial intrinsic torsion, which is more
properly defined as the cokernel of the Spencer differential: the
quotient of the space of torsions by the image of the Spencer
differential.

Since the Spencer differential is $G$-equivariant, both its kernel and
its cokernel are representations of $G$.  They are typically not
irreducible, but neither are they typically fully reducible.  Every
subrepresentation of $G$ in the cokernel of the Spencer
differential defines an intrinsic torsion class.  This is not something
we are familiar with from Lorentzian geometry, since one consequence
of the Fundamental Theorem of Riemannian Geometry is that (pseudo)Riemannian
$G$-structures have zero intrinsic torsion.  In contrast, as shown in
\cite{Figueroa-OFarrill:2020gpr} (but going back to the work of
Künzle~\cite{MR334831}), in generic dimension there are three intrinsic
torsion classes of Galilean geometries, typically called torsion-free,
twistless torsional and torsional \cite{Christensen:2013lma}.  Another
side of the same coin is that there is not a unique compatible affine
connection to a given Galilean structure, hence we cannot solve for
the connection just from a knowledge of the clock and ruler: the
ambiguity is an arbitrary $2$-form.

The study of intrinsic torsion for Galilean and Carrollian
$G$-structures are by now well understood (see, e.g.,
\cite{Figueroa-OFarrill:2020gpr}) and the Galilean case will be
briefly reviewed in our two formalisms below.  The aim of this paper
is to initiate the systematic study of intrinsic torsion in more
general types of non-Lorentzian geometries, collectively called
$p$-brane non-Lorentzian geometries, to be defined below, which can
often be obtained from limits of Lorentzian geometry in which the
local speed of light is sent to infinity or zero in some of the
spatial dimensions, while keeping it finite and nonzero in the
remaining dimensions. They form the natural spacetimes in which non-
or ultra-relativistic $p$-dimensional objects propagate.  For instance,
the target spacetimes of non-relativistic string theory
\cite{Gomis:2000bd,Danielsson:2000gi} exhibit String Newton--Cartan
geometry \cite{Andringa:2012uz,Bidussi:2021ujm,Bergshoeff:2022fzb}, an
example of $p$-brane non-Lorentzian geometry with $p=1$.  Here, we
will confine ourselves to $p$-brane Galilean and Carrollian geometry
and mostly focus on the $p$-brane Galilean case; although we
will explain how to read off the results for $p$-brane Carrollian
geometry from those of ($D-p-2$)-brane Galilean structures (cf. the
formal duality map of \cite{Barducci:2018wuj,Bergshoeff:2020xhv}).

In this paper, we will study $p$-brane Galilean geometry in both a
Cartan-geometric way in the language of $G$-structures, but also in a more
traditional treatment that is similar to discussions of the $p=0$ case
that have previously appeared in the physics literature.  The group
$G$ defining the $G$-structure is isomorphic to
\begin{equation}
  G \cong (\Ort(1,p) \times \Ort(D-p-1)) \ltimes \mathbb{R}^{(p+1)(D-p-1)}\,,
\end{equation}
but as mentioned above it is not the isomorphism class of the group
that determines the $G$-structure, but how it sits inside $\GL(D,\RR)$.
There are many ways to embed $(\Ort(1,p) \times \Ort(D-p-1)) \ltimes
\mathbb{R}^{(p+1)(D-p-1)}$ inside $\GL(D,\RR)$, but each way requires
a choice of splitting of $\RR^D$ into a ($p+1$)-dimensional
``longitudinal'' space and a ($D-p-1$)-dimensional ``transverse''
space.  One of the things that the Cartan-geometric approach will teach
us is that whereas the transverse space is an invariant notion
independent of the frame, the longitudinal space is frame-dependent
and hence not intrinsic to the geometry.\footnote{This is, of course,
physically desirable since the longitudinal space corresponds to the
tangent space to the brane worldvolume and if this were forced to
coincide with a given subspace of the tangent space, it would not
allow for interesting brane dynamics.}  The first aim in this paper is
the determination of the intrinsic torsion classes of a $p$-brane
Galilean geometry and their geometric characterisation in terms of the
characteristic tensors of the $G$-structure.  We will do this in two
ways, which we could describe as with and without choices.  We
illustrate the two approaches in Section~\ref{sec:Galilean}, where we
revisit the classification of Galilean structures in terms of
intrinsic torsion.

This paper is organised as follows.  In Section~\ref{sec:Galilean},
which we intend to serve as part of a Rosetta Stone between our two
approaches, we briefly review the classification of particle Galilean
geometries in the two languages used in the rest of the paper.  We
start with the standard description of such geometries in the physics
literature and then we translate into the Cartan-geometric language of
$G$-structures.  Section~\ref{sec:without} contains the classification
of $p$-brane Galilean structures in the language of Cartan geometry
and at the start of that section we have a more detailed description
of its contents.  Section~\ref{sec:dictionary}, which we intend to
serve as the second part of our Rosetta Stone, contains a dictionary
between the mathematical Section~\ref{sec:without} and the more
physical Section~\ref{sec:TNC2}.  We continue with
Section~\ref{sec:galgrav} where we investigate, using a physics language,
which $p$-brane Galilean and Carrollian geometries can be obtained by
taking a special limit of the Einstein-Hilbert action of general
relativity with its underlying Lorentzian geometry. The intrinsic
torsion classes play an important role in this discussion. They arise
in different ways depending on whether one takes a limit of the
Einstein-Hilbert action in a first-order or second-order
formulation. In the first case we find that the limit leads to
constraints on the intrinsic torsion classes specifying which Galilean
and Carrollian geometries can be obtained in this way. In the second
case, one finds an additional invariant, called
`electric' Galilei/Carroll gravity in the physics literature, that
does not have a first-order formulation and therefore does not arise in
the first case. To compare the remaining invariant that is common to
the two limits we need to apply a so-called Hubbard--Stratonovich
transformation in the second-order formulation introducing a Lagrange
multiplier.  Finally, Section~\ref{sec:conclusions} contains some
conclusions and points to future work.


\section{Galilean Geometry}
\label{sec:Galilean}

In this section we will review the classification of Galilean
geometries via their intrinsic torsion in two different ways,
illustrating in a simple (but nontrivial) example the two approaches
followed in this paper.

Roughly speaking, a $D$-dimensional Galilean manifold $M$ looks
locally like $D$-dimensional Galilei spacetime, which is a
$D$-dimensional affine space, with a distinguished notion of clock and
ruler.  The clock defines a fibration onto the affine line (``time'')
whose fibres are affine hypersurfaces of simultaneity, copies of
($D-1$)-dimensional affine space on which the ruler agrees with the
Euclidean distance.

More precisely, then, a Galilean manifold is a triple
$(M,\tau,\gamma)$, where $\tau \in \Omega^1(M)$ is a nowhere-vanishing
clock one-form defining a corank-$1$ subbundle
$E = \ker\tau \subset TM$, and $\gamma$ (the ruler) is a Riemannian
metric on $E$. This means that $\gamma$ is a positive-definite section
of $\odot^2 E^*$. It is alas a linear-algebraic fact that the dual of
a subspace is not a subspace of the dual, but rather a quotient, and
so too for vector bundles. In particular, $E^* \cong T^*M/\ann E$,
where $\ann E \subset T^*M$ is the annihilator of $E$, which in this
case is the line bundle spanned by $\tau$. Therefore $\gamma$ is
\emph{not} a tensor field on $M$, not being a section of any tensor
bundle\footnote{By a tensor bundle on a manifold $M$ we mean a vector
  subbundle of $\underbrace{TM \otimes \cdots \otimes TM}_r \otimes
  \underbrace{T^*M \otimes \cdots T^*M}_s$ for some non-negative
  integers $r$ and $s$.}, but rather of a quotient of a tensor bundle.
Since we are used to working with tensor fields, the way to remedy
this is to observe that a metric on $E$ defines a metric on $E^*$,
which is then a section $h$, say, of $\odot^2 E \subset \odot^2 TM$
and hence a bona fide tensor field on $M$. This so-called
\emph{spatial cometric} $h$ has corank $1$ and is such that
$h(\tau,-)=0$. We may therefore rephrase a Galilean manifold as a
triple $(M,\tau,h)$ with the above properties.  Relative to a local
chart with coordinates $x^\mu$, we have local expressions:
$\tau = \tau_\mu dx^\mu$ with $\tau_\mu = \tau(\d_\mu)$ and
$h = h^{\mu\nu} \d_\mu \otimes \d_\nu$, with
$h^{\mu\nu} = h(dx^\mu,dx^\nu)$, where $h^{\mu\nu}\tau_\nu = 0$.

There is a Galilean analogue of the notion of ``local orthonormal
frames'' in Lorentzian geometry.  It is simpler to start by defining
a distinguished family of coframes
$(\tau, e^a)$, for $a =1,\dots,D-1$, which are one-forms on some chart
$U \subset M$.  Whereas $\tau$ extends to a global one-form, the $e^a$
are only locally defined, but they satisfy $h(e^a, e^b) = h^{\mu\nu}
e^a_\mu e^b_\nu = \delta^{ab}$, with $\delta^{ab}$ the Kronecker delta.  It is possible
to cover $M$ by charts on which we have such distinguished coframes
and on the overlaps between two such charts, the coframes are related
via local $G$-transformations, where $G$ is the subgroup of
$\GL(D,\RR)$ which preserves $\tau$ and 
$h$-orthonormality:
\begin{equation}
  \tau \mapsto \tau \qquad\text{and}\qquad e^a \mapsto e^b R_b{}^a +
  v^a \tau,
\end{equation}
with $R \in \Ort(D-1)$ and $v \in \RR^{D-1}$.  This group $G$ is
abstractly isomorphic to the Euclidean group $\Ort(D-1) \ltimes
\RR^{D-1}$, but there are many such subgroups of $\GL(D,\RR)$ and it
is important to describe it not up to isomorphism but as an actual
subgroup of $\GL(D,\RR)$.  As a matrix group, $G$ is given by $D\times
D$ matrices of the form
\begin{equation}
  \label{eq:G-as-matrix-group}
  \begin{pmatrix}
    1 & \boldsymbol{0}^T \\ \boldsymbol{v} & R
  \end{pmatrix} \qquad\text{where $\boldsymbol{v} \in \RR^{D-1}$ and
    $R \in \Ort(D-1)$.}
\end{equation}
Infinitesimally, the above transformations are
\begin{equation}
  \delta \tau = 0 \qquad\text{and}\qquad \delta e^a = - \lambda^a{}_b
  e^b + \lambda^a \tau.
\end{equation}
The index $a$ will be freely raised and lowered with a Kronecker delta
$\delta^{ab}$ or $\delta_{ab}$ in what follows. The parameters
$\lambda^a{}_b$ obey $\lambda_{ab} := \delta_{ac}\lambda^c{}_b = -
\lambda_{ba}$ and are those of an infinitesimal spatial rotation,
while $\lambda^a$ are those of an infinitesimal Galilean boost.

It is now imperative to rewrite the above transformation law in terms
of the components of the coframe relative to a local chart:
\begin{align} \label{eq:galvb}
  \delta \tau_\mu &= 0 \,, \qquad & \delta e_\mu{}^a &= -\lambda^a{}_b\, e_\mu{}^b + \lambda^a\, \tau_\mu \,,
\end{align}
since, following time-honoured tradition, we are going to use the
\emph{same} letters for the canonically dual frame, namely $(\tau^\mu,
e_a{}^\mu)$.  Notice that were we to leave out the local coordinate
indices, we would not be able to tell whether $\tau$ was a one-form or
a vector field.  Notice also that whereas the one-form $\tau$ is a
global one-form, the vector field $\tau$ is only locally defined.  The
dual frame fields $\tau^\mu$ and $e_a{}^\mu$ are defined by the
equations
\begin{align}
  \tau_\mu \tau^\mu = 1 \,, \qquad \tau_\mu e_a{}^\mu = 0 \,, \qquad e_\mu{}^a \tau^\mu = 0 \,, \qquad                                                                                                & e_\mu{}^a e_b{}^\mu = \delta^a_b \,, \qquad \tau_\mu \tau^\nu + e_\mu{}^a e_a{}^\nu = \delta_\mu^\nu \,,
\end{align}
which are just the canonical dual relations expanded out.

In order to define parallel transport in Galilean geometry, we must
introduce affine connections.  In the standard theoretical physics
procedure one introduces \emph{two} connections: an affine connection
with coefficients $\Gamma^\rho_{\mu\nu}$ as well as spin connection
with coefficients $\omega_\mu{}^{ab} = -\omega_\mu{}^{ba}$ and
$\omega_\mu{}^a$ for spatial rotations and Galilean boosts,
respectively.  The spin connection transforms under infinitesimal
local $G$-transformations according to the following rules:
\begin{align}
  \delta \omega_\mu{}^{ab} = \partial_\mu \lambda^{ab} - 2 \lambda^{[a|}{}_c \, \omega_\mu{}^{c|b]}\,, \qquad \quad
  \delta \omega_\mu{}^a = \partial_\mu \lambda^a - \lambda^{a}{}_{b}\, \omega_\mu{}^b + \lambda^b\, \omega_\mu{}^a{}_b \,.
\end{align}
The affine connection and spin connections are of course not independent of each
other but are related via the following ``Vielbein postulate'':
\begin{align} \label{eq:galvbpost}
   \partial_\mu \tau_\nu - \Gamma^\rho_{\mu\nu} \tau_\rho = 0 \,, \qquad \qquad
   \partial_\mu e_\nu{}^a + \omega_\mu{}^a{}_b\, e_\nu{}^b - \omega_\mu{}^a\, \tau_\nu - \Gamma^\rho_{\mu\nu} e_\rho{}^a = 0 \,.
\end{align}
This postulate ensures that the connection is compatible with the
rank-1 temporal metric $\tau_{\mu\nu} = \tau_\mu \tau_\nu$ and the
rank-$(D-1)$ spatial co-metric $h^{\mu\nu} = e_a{}^\mu e_b{}^\nu
\delta^{ab}$ of Galilean geometry.

Anti-symmetrization of \eqref{eq:galvbpost} leads to Cartan's first
structure equation of Galilean geometry:\begin{subequations} \label{eq:fststreqs}
\begin{alignat}{2}
  &  2 \partial_{[\mu} \tau_{\nu]} = T_{\mu\nu}{}^0 \qquad \qquad & & \text{with } \ T_{\mu\nu}{}^0 \equiv 2 \Gamma^\rho_{[\mu\nu]} \tau_\rho  \,, \label{eq:fststreqs1} \\
  &  2 \partial_{[\mu} e_{\nu]}{}^a + 2 \omega_{[\mu|}{}^a{}_b\, e_{|\nu]}{}^b - 2 \omega_{[\mu}{}^a\, \tau_{\nu]} = T_{\mu\nu}{}^a \qquad \qquad & & \text{with } \ T_{\mu\nu}{}^a \equiv  2 \Gamma^\rho_{[\mu\nu]} e_\rho{}^a  \label{eq:fststreqs2} \,,
\end{alignat}
\end{subequations}
where $T_{\mu\nu}{}^0$ and $T_{\mu\nu}{}^a$ are the temporal and
spatial components of the torsion tensor $T_{\mu\nu}{}^\rho \equiv 2
\Gamma_{[\mu\nu]}^\rho$.

Equation \eqref{eq:fststreqs2} can be viewed as a system of linear
equations for the spin connection components $\omega_\mu{}^{ab}$ and
$\omega_\mu{}^a$. Solving it leads to expressions for some of these
components in terms of derivatives of the coframe fields, the frame
fields and the spatial torsion tensor components $T_{\mu\nu}{}^a$.
Note, however, that \eqref{eq:fststreqs2} does not suffice to express
all components of $\omega_\mu{}^{ab}$ and $\omega_\mu{}^a$ in this
way. In particular, one finds that the $D(D-1)/2$ components
$\tau^\mu \omega_{\mu}{}^{ab}$ and $\tau^\mu \omega_\mu{}^a$ are not
determined by \eqref{eq:fststreqs2}. Unlike what happens in Lorentzian
geometry, one thus sees that metric-compatible connections in Galilean
geometry are not uniquely determined in terms of the metric structure
and torsion tensor. Rather, there is a family of metric-compatible
connections with given torsion that is parametrised by the
undetermined spin connection components $\tau^\mu \omega_{\mu}{}^{ab}$
and $\tau^\mu \omega_\mu{}^a$.

Another difference with Lorentzian geometry concerns the different
role that the temporal and spatial torsion tensor fields
$T_{\mu\nu}{}^0$ and $T_{\mu\nu}{}^a$ play in Galilean
geometry. Whereas setting components of $T_{\mu\nu}{}^a$ equal to zero
merely amounts to a choice of connection, setting components of
$T_{\mu\nu}{}^0$ equal to zero also entails extra constraints on the
geometry of $M$. Indeed, there are two possible ways in which
components of $T_{\mu\nu}{}^0$ can be set to zero in a $G$-covariant
manner:
\begin{align} \label{eq:intTcond}
  e_a{}^\mu e_b{}^\nu\, T_{\mu\nu}{}^0 = 0  \qquad\qquad\text{or}\qquad\qquad T_{\mu\nu}{}^0 = 0 \, .
\end{align}
From \eqref{eq:fststreqs1}, the first possibility implies that
\begin{align} \label{eq:tdt}
  e_a{}^\mu e_b{}^\nu \partial_{[\mu} \tau_{\nu]} = 0 \qquad
  \Leftrightarrow \qquad \tau_{[\mu} \partial_\nu \tau_{\rho]} = 0
  \qquad \Leftrightarrow \qquad \partial_{[\mu} \tau_{\nu]} =
  \alpha_{[\mu} \tau_{\nu]}\, ,
\end{align}
for some (global) one-form $\alpha_\mu$.  It then follows from the
Frobenius Integrability Theorem that $M$ is foliated by spatial
hypersurfaces, whose tangent spaces are spanned by the $D-1$ vector
fields $e_a{}^\mu$.  The second possibility $T_{\mu\nu}{}^0 = 0$ implies that
$\tau_\mu$ is closed, a constraint stronger than
\eqref{eq:tdt}.  In that case, if $\tau$ were to be exact (e.g., if
$M$ were simply connected), so that $\tau_\mu = \partial_\mu T$, the
time function $T$ would define an absolute time on $M$.  If $\tau$ is
closed but not exact, then the time function only exists locally.

Torsion tensor components like those of $T_{\mu\nu}{}^0$ in Galilean
geometry, whose vanishing leads to geometric constraints, are
indicative of intrinsic torsion.  From the above example one sees that
intrinsic torsion can be informally defined as those torsion tensor
components for which the corresponding structure
equations (Equation~\eqref{eq:fststreqs1} in Galilean geometry) do not
contain any spin connection components.

Let us now translate the above into the language of Cartan geometry
and $G$-structures.  First, some notation: we introduce the shorthand
$\V = \RR^D$ for convenience.  Notice that $\V$ is not to be thought
of as an abstract vector space, but is just a notation for $\RR^D$.  Let
$G\subset \GL(\V) = \GL(D,\RR)$ be a Lie subgroup.  A $G$-structure on
a $D$-dimensional manifold $M$ is a subset of
distinguished frames, which are related not by general linear
transformations, but only by those in $G$.  A more precise definition
is that it is a principal $G$-subbundle $P$, say, of the frame bundle
of $M$. Every $G$-structure comes with a soldering form $\theta$,
which is a one-form on $P$ with values in $\V$.  Its value at a frame
$u \in P$ at $p \in M$ on a vector $X \in T_uP$ tangent to $P$ at $u$
is just the coordinate vector of the projection of $X$ to $T_pM$
relative to the frame $u$.  The components of $\theta$ relative to the
canonical basis of $\V$ define local coframes\footnote{We choose to
  index them starting from $0$ by analogy with the familiar Lorentzian
  case.} $\theta^0,\theta^1,\dots,\theta^{D-1}$, where
$\theta^0 =\tau$ and canonical dual local frame
$e_0,e_1,\dots,e_{D-1}$. The soldering form is instrumental in the
dictionary between the linear algebra of representations of $G$ and
the tensor bundles over $M$. Every representation of $G$ gives rise to
an associated vector bundle over $M$. Since $G\subset \GL(\V)$, it acts
naturally on $\V$ and hence there is a vector bundle over $M$
associated to $\V$. This bundle is called the fake tangent bundle and the
soldering form, which defines a one-form on $M$ with values in the
fake tangent bundle, defines an isomorphism from the real to the fake
tangent bundle.

Any $G$-equivariant linear map between representations of $G$ induces
a bundle map from the corresponding associated vector bundles. For
example, a subrepresentation $W \subset \V$ defines a subbundle $E
\subset TM$.  Indeed, many of the natural maps between vector bundles
in the differential geometry of manifolds with a $G$-structure are
induced from linear-algebraic maps between representations of $G$.
Similarly, $G$-invariant tensors on $\V$ give rise to tensor fields on
$M$ which are said to be characteristic to the $G$-structure. For
example, in Lorentzian geometry, $G$ is defined as the orthogonal
subgroup of $\GL(\V)$ preserving some inner product. That inner
product gives rise to a metric on $M$. In Galilean geometry, $G$ is
such that it leaves invariant a covector in $\V^*$ and a symmetric
tensor in $\odot^2\V$, giving rise to the clock one-form and the
spatial cometric on $M$.

In this language connections adapted to the $G$-structure arise as
follows. One picks an Ehresmann connection on the $G$-structure, which
we recall is a principal $G$-bundle $P$ over $M$. It can be described
in many equivalent ways: as a $G$-invariant horizontal distribution in
$TP$, as a $\g$-valued one-form on $P$ or, locally on $M$, as a
$\g$-valued one-form (i.e., a gauge field). An Ehresmann connection in
turn gives a Koszul connection on any associated vector bundle to the
$G$-structure, and therefore in particular on the fake tangent bundle. Since the
soldering form sets up a bundle isomorphism between the fake and real
tangent bundles, we may transport the Koszul connection on
the fake tangent bundle to an affine connection on the tangent bundle.
This affine connection is said to be adapted to the $G$-structure and
is compatible with the characteristic tensor fields. We may
contrast this with the approach outlined above, which introduces both
an affine and a spin connections subject to a Vielbein postulate. The
Vielbein postulate encodes the fact that the affine connection is the
one obtained from the Koszul connection on the fake tangent bundle via
the soldering form. Indeed, the Vielbein postulate simply says that
the soldering form (as a one-form on $M$ with values in the fake
tangent bundle) is parallel, where we use the affine connection on the
one-form indices and the spin connection on the frame indices. It is
then an easy consequence of the fact that the characteristic tensor
fields come from $G$-invariant tensors, that any adapted affine
connection is compatible with the characteristic tensor fields.

It is often convenient to work on the total space $P$ of the
$G$-structure, since the objects of interest are globally defined. Let
$\omega$ be the $\g$-valued one-form of the Ehresmann connection and
let $\theta$ be the $\V$-valued soldering form.  Cartan's first
structure equation defines the $\V$-valued torsion $2$-form $\Theta$
by
\begin{equation}
  \Theta = d\theta + \omega \wedge \theta,
\end{equation}
where the second term in the RHS involves the action of $\omega$, which
is $\g$-valued, on $\theta$ which is $\V$-valued.  $\Theta$ is a
globally defined $2$-form on $P$ with values in $\V$.  It can also be
understood as a $2$-form on $M$ with values in the fake tangent bundle
or, via the soldering form, also as a $2$-form on $M$ with values in
the tangent bundle, which is the usual description of the torsion
tensor of an affine connection.  Indeed, for any two vector fields
$X,Y \in \eX(M)$,
\begin{align*}
\Theta^a(X,Y) &= d\theta^a(X,Y) + \omega^a{}_b(X)\theta^b(Y) - \omega^a{}_b(Y)\theta^b(X)\\
              &= X \theta^a(Y) - Y\theta^a(X) - \theta^a([X,Y]) + \omega^a{}_b(X)\theta^b(Y) - \omega^a{}_b(Y)\theta^b(X)\\
              &= \theta^a(\nabla_X Y) - \theta^q(\nabla_Y X) - \theta^a([X,Y]) &\tag{by the vielbein postulate}\\
              &= \theta^a(\nabla_X Y - \nabla_Y X - [X,Y])\\
              &= \theta^a(T^\nabla(X,Y)),
\end{align*}
where we have used the vielbein postulate in the form
\begin{equation}
  X \theta^a(Y) + \omega^a{}_b(X) \theta^b(Y) - \theta^a(\nabla_X Y) = 0.
\end{equation}

Now suppose that $\omega'$ is the connection one-form of another
Ehresmann connection.  Then $\omega' = \omega + \kappa$, where
$\kappa$ is a one-form with values in $\g$, descending to $M$ as a
one-form with values in the vector bundle $\Ad P$ associated to the
adjoint representation $\g$.  The torsion $2$-form $\Theta'$ of
$\omega'$ is related to that of $\omega$ by
\begin{equation}
  \Theta'= d\theta + \omega' \wedge \theta = \Theta + \kappa \wedge
  \theta.
\end{equation}
The passage from $\kappa$ to $\kappa\wedge \theta$ can be described on
many levels: it is a bundle map between vector bundles on $M$ and a
tensorial map between the space of sections.  These sections are
$G$-equivariant functions on $P$ with values in the corresponding
representations of $G$ and the tensorial map is simply composition
with a linear map between the corresponding representations. This
linear map $\d : \g \otimes \V^* \to \V \otimes \wedge^2\V^*$ is an
instance of a Spencer differential.  We also use the same notation for
the associated bundle map and map between sections.  In particular we can
write $\Theta'= \Theta + \d\kappa$, where for any two vector fields
$X,Y \in \eX(M)$,
\begin{equation}
  \d\kappa(X,Y) = \kappa(X) Y - \kappa(Y) X.
\end{equation}
Any linear map has a kernel and cokernel and $\d$ is no exception.
But there is more: since $\d$ is $G$-equivariant, its kernel and
cokernel are $G$-submodules, and we get a four-term exact sequence of
$G$-modules and $G$-equivariant linear maps:
\begin{equation}
  \label{eq:spencer-4-term-exact-sequence}
  \begin{tikzcd}
    0 \arrow[r] & \ker\d \arrow[r] & \g \otimes \V^* \arrow[r,"\d"] & \V
    \otimes \wedge^2\V^* \arrow[r] & \coker\d \arrow[r] & 0,
  \end{tikzcd}
\end{equation}
where $\coker\d = (\V \otimes \wedge^2\V^*)/\im \d$.  This sequence is
illustrated in Figure~\ref{fig:orlandogram-for-spencer}.

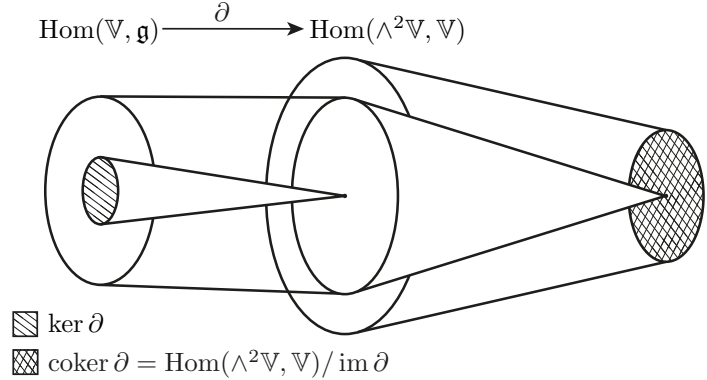
\begin{figure}[ht!]
  \begin{centering}
    \graphicspath{{diagrams_that_should_not_bear_my_name/}}
    \def\svgwidth{.62\linewidth}
\begingroup%
  \makeatletter%
  \providecommand\color[2][]{%
    \errmessage{(Inkscape) Color is used for the text in Inkscape, but the package 'color.sty' is not loaded}%
    \renewcommand\color[2][]{}%
  }%
  \providecommand\transparent[1]{%
    \errmessage{(Inkscape) Transparency is used (non-zero) for the text in Inkscape, but the package 'transparent.sty' is not loaded}%
    \renewcommand\transparent[1]{}%
  }%
  \providecommand\rotatebox[2]{#2}%
  \newcommand*\fsize{\dimexpr\f@size pt\relax}%
  \newcommand*\lineheight[1]{\fontsize{\fsize}{#1\fsize}\selectfont}%
  \ifx\svgwidth\undefined%
    \setlength{\unitlength}{820.53900146bp}%
    \ifx\svgscale\undefined%
      \relax%
    \else%
      \setlength{\unitlength}{\unitlength * \real{\svgscale}}%
    \fi%
  \else%
    \setlength{\unitlength}{\svgwidth}%
  \fi%
  \global\let\svgwidth\undefined%
  \global\let\svgscale\undefined%
  \makeatother%
  \begin{picture}(1,0.54122349)%
    \lineheight{1}%
    \setlength\tabcolsep{0pt}%
    \put(0,0){\includegraphics[width=\unitlength,page=1]{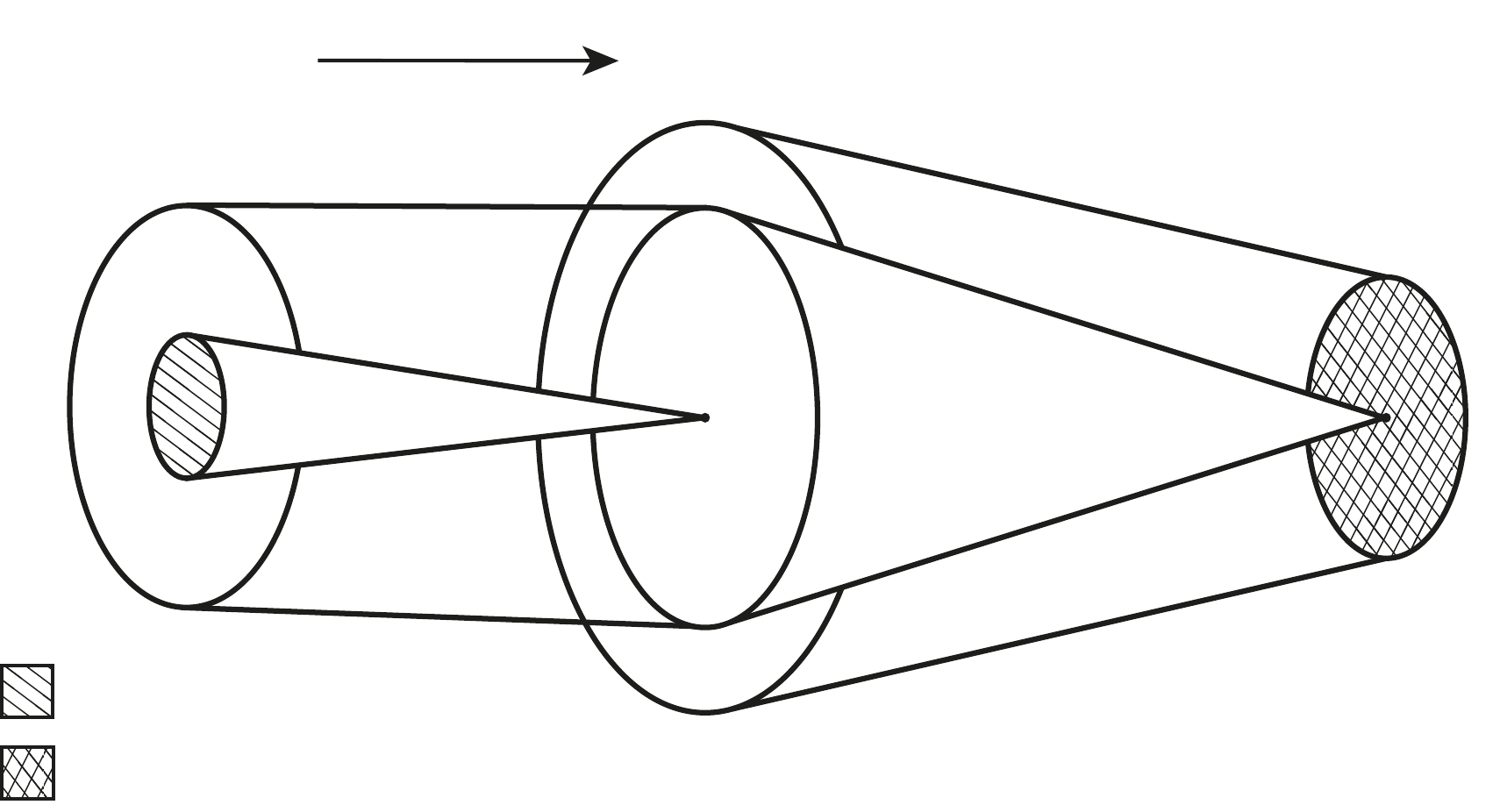}}%
    \put(0.2851361,0.51133971){\color[rgb]{0.11372549,0.11372549,0.10588235}\makebox(0,0)[lt]{\lineheight{1.25}\smash{\begin{tabular}[t]{l}$\partial$\end{tabular}}}}%
    \put(0.05051521,0.01257771){\color[rgb]{0.11372549,0.11372549,0.10588235}\makebox(0,0)[lt]{\lineheight{1.25}\smash{\begin{tabular}[t]{l}$\coker\partial = \Hom(\wedge^2 \mathbb{V}, \mathbb{V}) / \operatorname{im}\partial$\end{tabular}}}}%
    \put(0.05036199,0.06916378){\color[rgb]{0,0,0}\makebox(0,0)[lt]{\lineheight{1.25}\smash{\begin{tabular}[t]{l}$\operatorname{ker}\partial$\end{tabular}}}}%
    \put(0.03760362,0.48817555){\color[rgb]{0,0,0}\makebox(0,0)[lt]{\lineheight{1.25}\smash{\begin{tabular}[t]{l}$\Hom(\mathbb{V},\gm)$\end{tabular}}}}%
    \put(0.42035472,0.48750405){\color[rgb]{0,0,0}\makebox(0,0)[lt]{\lineheight{1.25}\smash{\begin{tabular}[t]{l}$\Hom(\wedge^2 \mathbb{V}, \mathbb{V})$\end{tabular}}}}%
  \end{picture}%
\endgroup%

    \caption{This figure illustrates the different spaces in the
      $4$-term exact sequence~\eqref{eq:spencer-4-term-exact-sequence}
      defined by the Spencer differential $\partial$.  The larger of
      the two left ellipses stands for the space
      $\Hom (\V,\mathfrak{g}) \cong \g \otimes \V^*$ of differences of
      adapted connections, whereas the smaller shaded ellipse is the
      subspace of such differences which do not alter the torsion: the kernel of $\d$.  The
      larger of the two middle ellipses stands for the vector space
      $\Hom(\wedge^2\V,\V) \cong \V \otimes \wedge^2\V^*$ of all
      torsion tensor components, whereas the smaller one stands for
      the image of $\d$.  The rightmost ellipse stands for the
      quotient vector space
      $\coker \d = \Hom (\Lambda^2 \V, \V)/\im \partial$ of
      equivalence classes of torsion tensors, where two torsion
      tensors are equivalent if their difference is due to a change in
      the connection.  Therefore these equivalence classes are, by
      definition, independent of the connection and hence intrinsic to
      the $G$-structure.}
  \label{fig:orlandogram-for-spencer}
  \end{centering}
\end{figure}

The bundles associated to the representations in the
sequence~\eqref{eq:spencer-4-term-exact-sequence} have the following
geometric interpretation:
\begin{itemize}
\item $\g \otimes \V^*$ corresponds to the differences between adapted
  connections;
\item $\V \otimes \wedge^2\V$ corresponds to the space of torsions;
\item $\ker\d \subset \g\otimes \V^*$ are the changes in the
  connection which do not alter the torsion; and
\item $\coker\d$ is the space of intrinsic torsions, which is the main
  object of study in this paper.
\end{itemize}
The reason we say that $\coker\d$ consists of intrinsic torsions is
that the difference in the torsions of any two adapted connections
lives in the image of the Spencer differential, hence the projection
of the torsion to $\coker\d$ is independent on the adapted connection.
It is therefore an intrinsic property of the $G$-structure.  This
means that we can calculate the intrinsic torsion departing from any
adapted connection.

For Galilean $G$-structures, as calculated in
\cite{MR334831,Figueroa-OFarrill:2020gpr}, one has isomorphisms
$\ker\d \cong \coker\d \cong \wedge^2\V^*$ as representations of $G$.
On the one hand, the isomorphism $\ker\d \cong \wedge^2 \V^*$ implies
that the torsion does not uniquely determine an adapted connection:
one may modify the connection by a section of the associated vector
bundle to $\wedge^2\V^*$ (i.e., by a $2$-form) without altering the
torsion.  On the other hand, the isomorphism $\coker\d \cong
\wedge^2\V^*$ says that the intrisinc torsion of a Galilean structure
is captured by a $2$-form, which one calculates to be the composition
of the torsion $T^\nabla$ of the adapted affine connection $\nabla$
with the clock one-form $T^\nabla  \mapsto \tau \circ T^\nabla$.
Since $T^\nabla$ is a $2$-form with values in vector fields, its
composition with the one-form $\tau$ is simply a $2$-form.  Using that
$\tau$ is parallel with respect to $\nabla$, one calculates that $\tau
\circ T^\nabla = d\tau$, which, as expected, is independent of the
connection.

The story does not end there. As a representation of $G$, the vector
space $\wedge^2\V^*$ is not irreducible; although it is
indecomposable.  In generic dimension (here, $\dim \V \neq 2$)
it has a unique proper subrepresentation, whose associated vector
bundle is the subbundle of the bundle of $2$-forms whose sections are
characterised by the property that they are in the kernel of wedging
with $\tau$.  Therefore we have three $G$-subrepresentations of
$\wedge^2\V^*$ and hence three intrinsic torsion classes:
\begin{itemize}
\item vanishing intrinsic torsion: $d\tau = 0$;
\item ``twistless'' intrinsic torsion: $d\tau \neq 0$ but $\tau \wedge
  d\tau = 0$; and
\item generic intrinsic torsion: $\tau \wedge d\tau \neq 0$.
\end{itemize}
This classification agrees with the torsionless, twistless torsional
and torsional Newton--Cartan geometries introduced in
\cite{Christensen:2013lma}.


\section{\texorpdfstring{$p\,$-brane Galilean geometries à la Cartan}{p-brane Galilean geometries à la Cartan}}
\label{sec:without}

In the previous section we have discussed the classification of
Galilean geometries in terms of the intrinsic torsion of the
associated $G$-structure.  In this paper we follow the philosophy that
such geometries describe the Galilean structure on a manifold on which
a point particle propagates. If instead we are interested in
describing the propagation of extended objects (strings or, more
generally, $p$-branes), then the relevant structure is the so-called
$p$-brane Galilean structure.

Intuitively, a $p$-brane propagating in a $D$-dimensional manifold $M$
is defined by an embedding $\Sigma \to M$, where the image of $\Sigma$
in $M$ is the ($p+1$)-dimensional worldvolume of the $p$-brane.  Let
us consider $\Sigma$ to be Lorentzian, for the sake of exposition, so
that it has a (pseudo-)orthonormal coframe
$\vartheta^0,\dots,\vartheta^p$ where the metric is given by
$-(\vartheta^0)^2 + (\vartheta^1)^2 + \dots + (\vartheta^p)^2$.
This suggests that in $M$ we should have distinguished coframes
$(\theta^0,\dots,\theta^{D-1})$ with
$\eta = -(\theta^0)^2 + (\theta^1)^2 + \dots + (\theta^p)^2$ a global
section of $\odot^2 T^*M$.  The common kernel of
$\theta^0,\dots,\theta^p$ defines a subbundle $E \subset TM$.  Dually,
we say that $\theta^0,\dots,\theta^p$ span the annihilator
$\ann E \subset T^*M$.  We see that $\eta$ is then a section of
$\odot^2\ann E \subset \odot^2T^*M$, which is nondegenerate in that it
defines a Lorentzian metric on $TM/E$.  By analogy with the particle
case, we take the ``ruler'' now to be a Riemannian metric $\gamma$ on
$E$, which as before we invert to a cometric $h$ on $E^*$ which is a
tensor field on $M$, being a section of $\odot^2 E\subset \odot^2 TM$.
Therefore we take a $p$-brane Galilean structure on $M$ to be a pair
of tensors $\eta$ and $h$, where $\eta$ defines a metric on $TM/E$,
with $E \subset TM$ a corank-($p+1$) subbundle of $TM$, and $h$, a
section of $\odot^2E \subset \odot^2 TM$, defines a Riemannian
cometric: a Riemannian metric on $E^*$.  This gives rise to
distinguished coframes $(\theta^0,\dots,\theta^{D-1})$ where the
$\theta^0,\dots,\theta^p$ are such that
$\eta = - (\theta^0)^2 + \dots + (\theta^p)^2$ and
$h(\theta^i,\theta^j)=\delta^{ij}$ for $i,j=p+1,\dots,D-1$.  Such
distinguished coframes are related on the overlaps by local
$G$-transformations where $G \subset \GL(D,\RR)$ is the subgroup which
preserves $\eta$ and $h$.  The canonical dual frames to these
distinguished coframes define a $G$-structure on $M$ and it is one of
the aims of this paper to explain the classification of such
$G$-structures in terms of their intrinsic torsion classes.

Given a frame $(e_0,e_1,\dots,e_{D-1})$ in the $G$-structure, the
tangent spaces to $M$ break up into two subspaces: a longitudinal
subspace spanned by $e_0,\dots,e_p$ and a transverse subspace spanned
by $e_{p+1},\dots,e_{D-1}$.  This latter subspace is the fibre to the
subbundle $E$ and hence it is well-defined and independent of the
frame. However, this is not so for the longitudinal subspace, which is
not preserved by $G$.  Indeed, whereas $\theta^0,\dots,\theta^p$ do
transform into each other (analogous to how $\theta^0$ is invariant in
the particle case -- i.e., $p=0$), the remaining
$\theta^{p+1},\dots,\theta^{D-1}$ do not.  Dually, whereas the span of
the $e_{p+1},\dots,e_{D-1}$ is preserved, that of the
$e_0,\dots,e_{p}$ is not.  What this says is that the defining
representation of $G$ is not irreducible (the transverse subspace is
invariant), but it is indecomposable in that the transverse subspace
has no invariant complement.  This somewhat complicates the treatment
and suggests rephrasing the discussion in terms of filtered
$G$-representations to be useful.

This section is organised as follows. In Section~\ref{sec:filtrations}
we briefly recap the basic language of filtered modules and their
associated graded modules. In Section~\ref{sec:group-interest} we
identify the group $G$ of interest. The group $G$ is isomorphic to the
semidirect product \begin{equation*}
  \Hom(\RR^{p+1},\RR^{D-p-1}) \rtimes (\Ort(p,1) \times \Ort(D-p-1)),
\end{equation*}
but this isomorphism is not canonical.  What we do have is an abelian
extension of $H:= \Ort(p,1) \times \Ort(D-p-1)$ by the abelian normal
subgroup $B:=\Hom(\RR^{p+1},\RR^{D-p-1})$, with $(g,h) \in H$ acting
on $b \in B$ by $b \mapsto h \circ b \circ g^{-1}$.  This extension
splits, but not canonically.\footnote{This is analogous to the
  Poincaré group being described as the semidirect product of the
  Lorentz group and the translations, suggesting misleadingly that
  there is a preferred Lorentz subgroup of the Poincaré group: in
  fact, there is one for every point in Minkowski spacetime and they
  are all conjugate (under translations) subgroups of the Poincaré
  group.  The translations, being a normal subgroup, are invariant
  under conjugation and hence they are unambiguously defined.} In
Section~\ref{sec:g-structures-adapted-connections} we briefly recall
the notion of the intrinsic torsion of a $G$-structure and set up the
problem we wish to solve: in the first instance, the determination of
the $G$-submodules of the cokernel of the Spencer differential, and
then the geometric characterisation of each of the intrinsic torsion
classes.  In this section we work only with $G$-invariant
objects: $G$-modules and $G$-equivariant maps. This comes at a price:
since the group $G$ does not act fully reducibly in the vector spaces
of interest, we have $G$-modules and $G$-submodules without
$G$-invariant complements. In other words, we have filtered
$G$-modules, with Section~\ref{sec:some-filtered-g} displaying the relevant
$G$-invariant filtrations of the vector spaces. In
Section~\ref{sec:kern-spenc-diff} we determine the kernel of the
Spencer differential and its structure as a $G$-module. We do this
while maintaining manifest $G$-symmetry via the trick of passing from
a filtered module to its associated graded module. This results in
non-isomorphic $G$-modules with isomorphic underlying vector spaces,
hence equivalent for the purposes of counting dimension. In
Section~\ref{sec:intr-tors-class} we establish an explicit
$G$-equivariant isomorphism between the kernel and cokernel of the
Spencer differential. This allows us to transport the $G$-submodules
of the kernel and arrive at the $G$-submodules of the cokernel: that
is, the intrinsic torsion classes, culminating in
Theorem~\ref{thm:intrinsic-torsion-classes}. In
Section~\ref{sec:geom-interpr} we describe each of the torsion classes
algebraically as submodules of the space of torsions and we derive
alternative geometric characterisations of each of the classes in
terms of the subbundle $E$ and the characteristic tensor fields of the
$G$-structure. The main result is summarised in
Theorem~\ref{thm:summary}.

Throughout this section we have used the language of $p$-brane
Galilean structures, which assumes that $\eta$ is Lorentzian and
$h$ Riemannian, but nothing in the formalism uses the signatures.
By taking $\eta$ Riemannian and $h$ Lorentzian, we can describe
($D-p-2$)-brane Carrollian structures.  This means that the
classification of standard (particle) Carrollian $G$-structures
follows formally from that of the ($D-2$)-brane Galilean structures,
where the worldvolume of the brane is a domain wall.  In this section
we are taking $p$ to be generic, and this leaves some special cases to
be further discussed, chiefly stringy Galilean and Carrollian
structures, which will be the subject of a later work.

\subsection{Brief recap of filtrations}
\label{sec:filtrations}

In this section we will be making use of some basic language of
filtered representations of a Lie group $G$, or, with some abuse of
language, filtered $G$-modules.

Suppose that $V$ is a $G$-module which admits a $G$-submodule
$W \subset V$; that is, a vector subspace which is preserved by $G$,
so that $g \cdot v \in W$ for all $v \in W$ and $g \in G$. This gives
rise to a canonical short exact sequence of $G$-modules:
\begin{equation}
  \begin{tikzcd}
    0 \arrow[r] & W \arrow[r] & V \arrow[r] & V/W \arrow[r] & 0.
  \end{tikzcd}
\end{equation}
In an ideal world, one would expect that there should be a
complementary $G$-submodule $W' \subset V$ such that
$V = W \oplus W'$.  If so, we would say that the above sequence splits
(in the category of $G$-modules): the splitting being the $G$-module
isomorphism $V/W \to W'\subset V$.  Alas, the real world is far from
ideal and it happens very often (and we will see plenty of examples
below) that no such complementary submodule exists.  Of course, we
always have a complementary vector subspace, but it will not be 
preserved by $G$.  This situation is a paradigmatic example of a
filtered $G$-module.

More generally, a (finite, descending) filtration $V^\bullet$ of a
$G$-module $V$ is sequence of $G$-submodules
\begin{equation}
  V = V^0 \supset V^{-1} \supset V^{-2} \supset \cdots \supset  V^{-N} = 0,
\end{equation}
for some $N \in \NN$.  We often find it convenient to extend the
filtration infinitely in both directions, declaring that $V^ i = V$
for all $i \geq 0$ and $V^i = 0$ for all $i \leq -N$.  We say that
$v\in V^i$ has \emph{filtration degree} $\leq i$.  It follows that
every quotient $V_i := V^i/V^{i-1}$ is naturally a $G$-module.  These
are the graded pieces of the associated graded module to the filtered
module $V$:
\begin{equation}
  \gr V = \bigoplus_{i\leq 0} V_i = V_0 \oplus V_{-1} \oplus \cdots \oplus V_{-N+1},
\end{equation}
where $V_{-N+1} = V^{-N+1}$, since $V^{-N} = 0$.  Again, we sometimes
find it convenient to extend the grading to a $\ZZ$-grading, declaring
that $V_i=0$ for $i>0$ and $i\leq -N$.  If $v \in V^i$, we will let
$[v]$ denote its image in $V_i$.  We say that $[v]$ has degree $i$.
Counting dimensions, we see that
\begin{align*}
  \dim \gr V &= \dim V_0 + \dim V_{-1} + \cdots + \dim V_{-N+1}\\
             &= \left( \dim V^0 - \dim V^{-1}\right) + \left(\dim V^{-1} - \dim V^{-2}  \right) + \cdots + \left( V^{-N+1}- \dim V^{-N}\right) \\
             &= \dim V^0 - \dim V^{-N}\\
             &= \dim V.
\end{align*}
Despite the fact that $\gr V$ and $V$ are $G$-modules of the same
dimension, they need not be isomorphic as $G$-modules.  Nevertheless,
it is often convenient to pass to $\gr V$ in order to count
dimension.  We shall see some examples of this procedure below.

A filtration $V^\bullet$ induces a filtration of the dual module $V^*$:
\begin{equation}
  0 = \ann V^0 \subset \ann V^{-1} \subset \ann V^{-2} \subset \cdots
  \subset \ann V^{-N+1} \subset \ann V^{-N} = V^*,
\end{equation}
where $\ann$ denotes the annihilator of a subspace.  The associated
graded $G$-module $\gr V^*$ is now positively graded:
\begin{equation}
  \gr V^* = \bigoplus_{i\geq 1} \left(V^*\right)_i,
\end{equation}
where
\begin{equation}
  \left( V^* \right)_i = \ann V^{-i}/\ann V^{-i+1} \cong \left(
    V^{-i+1}/V^{-i} \right)^* = V_{-i+1}^*.
\end{equation}
This implies that
\begin{equation}
  \gr V^* \cong \bigoplus_{i\geq 1} V_{-i+1}^* = \bigoplus_{i\leq 0} V_i^*,
\end{equation}
where we have re-indexed the direct sum.

Let $V^\bullet$ and $W^\bullet$ be two filtered $G$-modules and let
$\phi : V \to W$ be a $G$-equivariant linear map.  Assume further that
$\phi$ is compatible with the filtration, in that $\phi : V^i \to W^i$
for all $i$.  Then $\phi$ induces a linear map
$\gr \phi : \gr V \to \gr W$ by declaring for all $v \in V^i$
\begin{equation}
  \gr\phi([v]) = [\phi(v)]
\end{equation}
or, equivalently, $\gr\phi(v + V^{i-1}) = \phi(v) + W^{i-1}$.  Since
$\phi$ is $G$-equivariant, its kernel is a $G$-submodule of $V$.
Moreover, it's a filtered $G$-submodule, inheriting the filtration from
that of $V$ via
\begin{equation}
  (\ker\phi)^i = V^i \cap \ker \phi.
\end{equation}
Its associated graded $G$-module is $\gr\ker\phi = \bigoplus_i
(\ker\phi)_i$, where $(\ker\phi)_i$ consists of those $[v]$ where $v
\in (\ker\phi)^i$.  How does this compare with the kernel of $\gr
\phi$?  The kernel of $\gr\phi$ is a graded $G$-submodule of $\gr V$,
where $\ker\gr\phi = \bigoplus_i (\ker\gr\phi)_i$ with
$(\ker\gr\phi)_i \subset V_i$ consisting of those $[v] \in V_i$ such
that $\phi(v) \in W_{i-1}$.  Therefore we see that $\gr\ker\phi
\subseteq \ker\gr\phi$ as a graded $G$-submodule.  A natural question
is when this inclusion becomes an equality.

The linear map $\phi$ typically has pieces of different filtration
degrees defined as follows.  We have vector space isomorphisms $j_V: V
\to \gr V$ and $j_W: W \to \gr W$, so that given $\phi: V \to W$ we
obtain a map $\Phi:= j_W \circ \phi \circ j_V^{-1} : \gr V \to \gr W$
defined by the commutativity of the following square:
\begin{equation}
  \begin{tikzcd}
    V \arrow[d,"j_V"] \arrow[r,"\phi"] & W \arrow[d,"j_W"] \\
    \gr V \arrow[r,"\Phi"] & \gr W.
  \end{tikzcd}
\end{equation}
This map $\Phi$ between graded $G$-modules breaks up into pieces
$\Phi = \phi_0 + \phi_{-1} + \phi_{-2} + \cdots$, where
$\phi_j : V_i \to W_{i+j}$ for all $i$ and only non-positive degrees
appear because $\phi$ respects the filtration.  It is clear that
$\phi_0 = \gr\phi$.  Notice also that $\ker \Phi = \gr\ker\phi$ and
hence, as we saw above, $\ker \Phi \subseteq \ker \phi_0$.  If
$\Phi = \phi_0$, so that $\phi_j = 0$ for all $j<0$, then clearly
$\ker \Phi = \ker \phi_0$ and hence $\gr\ker\phi = \ker\gr\phi$.  We
summarise this discussion as follows:

\begin{lemma}
  \label{lem:gr-ker=ker-gr}
  If $\Phi = \phi_0$, so that $\phi_j = 0$ for all $j<0$, then
  $\ker\gr\phi = \gr \ker \phi$.
\end{lemma}

\subsection{\texorpdfstring{The Lie group $G$ of interest}{The Lie group G of interest}}
\label{sec:group-interest}

In this section we will define the Lie group $G$ of the $p$-brane
Galilean $G$-structure and will identify it as a subgroup of
$\GL(D,\RR)$.

Let $\V = \RR^D$ and let $W \subset \V$ be a nonzero subspace.  We
will let $\dim W =  D-p-1$.  Let $\ann W \subset \V^*$ denote the
annihilator of $W$.  It therefore has dimension $\dim \ann W = 
p+1$.  Notice that $W^* \not\subset \V^*$, but rather $W^* \cong
\V^*/\ann W$.

We let $\eta \in \odot^2 \ann W \subset \odot^2 \V^*$ and $h \in
\odot^2 W \subset \odot^2 \V$ be non-degenerate, in the sense that we
have short exact sequences
\begin{equation}
  \label{eq:eta}
  \begin{tikzcd}
    0 \arrow[r] & W \arrow[r] & \V \arrow[r,"\eta^\flat"] & \ann W
    \arrow[r] & 0,
  \end{tikzcd}
\end{equation}
and
\begin{equation}
  \label{eq:delta}
  \begin{tikzcd}
    0 \arrow[r] & \ann W \arrow[r] & \V^* \arrow[r,"h^\sharp"] &
    W \arrow[r] & 0,
  \end{tikzcd}
\end{equation}
where $\eta^\flat : \V \to \ann W$ and $h^\sharp : \V^* \to W$
are defined by
\begin{equation}
  \label{eq:eta-and-delta}
  \eta(v,v') = \left<\eta^\flat(v),v'\right>  \qquad\text{and}\qquad
  h(\alpha,\alpha') = \left<\alpha, h^\sharp(\alpha')\right>,
\end{equation}
for all $v,v' \in \V$ and $\alpha,\alpha' \in \V^*$, where
$\left<-,-\right> : \V^* \times \V \to \RR$ denotes the dual pairing.
Notice that $\eta^\flat \circ h^\sharp = 0$ and $h^\sharp
\circ \eta^\flat = 0$, but in fact, more is true: $\im h^\sharp =
\ker \eta^\flat$ and $\im \eta^\flat = \ker h^\sharp$, so that we
have an \emph{exact pair}
\begin{equation}\label{eq:exact-pair-delta-eta}
  \begin{tikzcd}
    \V \arrow[r, shift left, "\eta^\flat"] & \V^* \arrow[l,shift
    left, "h^\sharp"].
  \end{tikzcd}
\end{equation}

It follows that $\eta$ and $h$ induce symmetric bilinear inner
products $\overline\eta$ on $\V/W$ and $\overline h$ on $W^* \cong
\V^*/\ann W$, respectively.  Of course, $\overline h$ in turn
defines an inner product on $W$, which we will denote $\gamma$: explicitly,
\begin{equation}
  \label{eq:gamma}
  \gamma(h^\sharp \alpha ,w) := \left<\alpha, w\right>.
\end{equation}
Since $\im h^\sharp = W$, this allows us to calculate
$\gamma(w',w)$ simply by choosing $\alpha$ with
$h^\sharp(\alpha)=w'$.  Such $\alpha$ is not unique, but the
ambiguity lies in $\ann W$ and hence $\gamma$ is
well-defined.\footnote{An alternative definition of $\gamma$ is
  $\gamma(w,w') = h(\alpha,\alpha')$ where $\alpha,\alpha' \in
  \V^*$ are such that $h^\sharp\alpha = w$ and
  $h^\sharp\alpha' = w'$, which again is well-defined despite the
  fact that $\alpha$ and $\alpha'$ are only defined up to $\ann W$.}

Let $G \subset \GL(\V) = \GL(D,\RR)$ be the subgroup preserving $\eta$
and $h$. In other words, $g \in \GL(\V)$ belongs to $G$ if and
only if
\begin{equation}
  \label{eq:G-def}
  \eta(g \cdot v, g \cdot v') = \eta(v,v') \qquad\text{and}\qquad
  h(g \cdot \alpha, g \cdot \alpha') = h(\alpha,\alpha')
\end{equation}
for all $v, v'\in \V$ and $\alpha,\alpha'\in \V^*$.  Here $g \cdot $
means the natural action of $g \in \GL(\V)$ in $\V$ or $\V^*$, which
are related by
\begin{equation}
  \label{eq:dual-action-G}
  \left<g \cdot \alpha, v\right> = \left<\alpha, g^{-1} \cdot v\right>.
\end{equation}

\begin{lemma}
  \label{lem:G}
  Let $g \in \GL(\V)$.  Then $g \in G$ if and only if
  \begin{equation}
    \label{eq:G}
    g \cdot \eta^\flat(v) = \eta^\flat(g \cdot v) \qquad\text{and}\qquad
    g \cdot h^\sharp(\alpha) = h^\sharp(g \cdot \alpha),
  \end{equation}
  for all $v \in \V$ and $\alpha \in \V^*$. 
\end{lemma}

\begin{proof}
  This follows from the fact that the equations in
  \eqref{eq:G-def} and \eqref{eq:G} coincide.  Let us first discuss
  the first equation.  Let $g \in \GL(\V)$ and let $v, v'  \in \V$.
  Then from the first equation in \eqref{eq:eta-and-delta},
  \begin{align*}
    \eta(g \cdot v, g \cdot v') &= \left<\eta^\flat(g\cdot v), g \cdot v'\right>\\
                                &= \left<g^{-1} \cdot \eta^\flat(g\cdot v), v'\right>
  \end{align*}
  and hence $g$ satisfies the first equation in \eqref{eq:G-def} if
  and only if for all $v, v' \in \V$,
  \begin{equation*}
    \left<g^{-1} \cdot \eta^\flat(g\cdot v), v'\right> = \left<\eta^\flat(v), v'\right>.
  \end{equation*}
  Since this holds for all $v'$, nondegeneracy of the dual pairing
  says that we may abstract $v'$ and hence arrive at
  \begin{equation*}
    g^{-1}\cdot \eta^\flat(g \cdot v) = \eta^\flat(v) \qquad\text{for
      all $v \in \V$},
  \end{equation*}
  which is equivalent to $\eta^\flat(g \cdot v) = g \cdot
  \eta^\flat(v)$ for all $v \in \V$, which is the first equation in
  \eqref{eq:G}.

  To prove the equivalence between the second equations, we again let
  $g \in \GL(\V)$ and now $\alpha,\alpha'\in \V^*$.  From the second
  equation in \eqref{eq:eta-and-delta}, we have that
  \begin{align*}
    h(g\cdot\alpha',g \cdot\alpha) &= \left<g \cdot \alpha', h^\sharp(g \cdot \alpha)\right>\\
                                   &= \left<\alpha', g^{-1}\cdot h^\sharp(g \cdot \alpha)\right>
  \end{align*}
  and hence $g$ satisfies the second equation in \eqref{eq:G-def} if
  and only if for all $\alpha,\alpha'\in \V^*$,
  \begin{equation*}
    \left<\alpha', h^\sharp(\alpha)\right> = \left<\alpha', g^{-1}\cdot
      h^\sharp(g \cdot \alpha)\right>.
  \end{equation*}
  By non-degeneracy of the dual pairing we may abstract $\alpha'$,
  arriving at
  \begin{equation*}
    h^\sharp(\alpha) = g^{-1}\cdot h^\sharp(g \cdot \alpha)
    \qquad\text{for all $\alpha \in \V^*$},
  \end{equation*}
  which is equivalent to $h^\sharp(g\cdot \alpha) = g \cdot
  h^\sharp(\alpha)$ for all $\alpha \in \V^*$, which is the second
  equation in \eqref{eq:G}.
\end{proof}

Notice that $W \subset \V$ and $\ann W \subset \V^*$ are
$G$-submodules, since they are kernels of $G$-equivariant linear maps.
In other words, we have the following short exact sequences of
$G$-modules:
\begin{equation}
  \label{eq:V-as-G-mod}
  \begin{tikzcd}
    0 \arrow[r] & W \arrow[r] & \V \arrow[r] & \V/W \arrow[r] & 0
  \end{tikzcd}
\end{equation}
and, dually,
\begin{equation}
  \label{eq:Vdual-as-G-mod}
  \begin{tikzcd}
    0 \arrow[r] & \ann W \arrow[r] & \V^* \arrow[r] & \V^*/\ann W \arrow[r] & 0.
  \end{tikzcd}
\end{equation}
Neither of these sequences splits (as $G$-modules), so that there is
no $G$-submodule of $\V$ (resp.~$\V^*$) complementary to $W$
(resp.~$\ann W$). The exact sequences \eqref{eq:eta} and \eqref{eq:delta}, as well as the
exact pair \eqref{eq:exact-pair-delta-eta}, are also exact sequences
of $G$-modules with $G$-equivariant linear maps.  Indeed, these
sequences exhibit $\V$ and $\V^*$ as filtered $G$-modules.

\begin{lemma}
  $G$ preserves the inner product $\gamma$ on $W$ defined by
  Equation~\eqref{eq:gamma}.
\end{lemma}

\begin{proof}
  Let $w = h^\sharp(\alpha)$ and $w' = h^\sharp(\alpha')$
  for some $\alpha,\alpha' \in \V^*$ and let $g \in G$.  Then
  \begin{align*}
    \gamma(g\cdot w, g \cdot w') &= \gamma(g \cdot h^\sharp(\alpha), g \cdot h^\sharp(\alpha'))\\
                                 &= \gamma(h^\sharp(g \cdot \alpha), h^\sharp(g \cdot \alpha')) &\tag{by the second equation in \eqref{eq:G}}\\
                                 &= h(g \cdot \alpha, g \cdot\alpha') &\tag{by definition of $\gamma$}\\
                                 &= h(\alpha,\alpha') &\tag{since $h$ is $G$-invariant}\\
                                 &= \gamma(h^\sharp(\alpha), h^\sharp(\alpha')) &\tag{by definition of $\gamma$}\\
                                 &= \gamma(w,w').
  \end{align*}
\end{proof}

Let $\g$ denote the Lie algebra of $G$.  Since $G \subset \GL(\V)$, we
have that $\g$ is a Lie subalgebra of $\gl(\V)$.  This means that we
may identify every $X \in \g$ with the corresponding endomorphism $X
\in \End\V$.  The dual representation on $\V^*$ is such that $X
\mapsto -X^t \in \End\V^*$, where $X^t \in \End\V^*$ is the
transpose map, defined by
\begin{equation}
  \label{eq:transpose}
  \left<X^t \cdot \alpha, v\right> = \left<\alpha, X \cdot v\right>,
\end{equation}
for all $\alpha \in \V^*$ and $v \in \V$.

The following lemma follows from Lemma~\ref{lem:G} by taking $g =
\exp(t X)$ in Equation~\eqref{eq:G} and differentiating with respect
to $t$ at $t=0$.

\begin{lemma}
  \label{lem:lie-algebra-g}
  The Lie algebra $\g$ of $G$ consists of those endomorphisms $X \in
  \End\V$ such that
  \begin{equation}
    \label{eq:Lie-algebra-G}
    \eta^\flat \circ X + X^t \circ \eta^\flat = 0 
    \qquad\text{and}\qquad
    X \circ h^\sharp + h^\sharp \circ X^t = 0.
  \end{equation}
\end{lemma}

The $G$-submodule $W \subset \V$ defines a $G$-invariant filtration of
$\V$ and a corresponding dual filtration of $\V^*$:
\begin{equation}
  \label{eq:filtrations-V}
    0 \subset W \subset \V
  \qquad\text{and}\qquad
    0 \subset \ann W \subset \V^*
\end{equation}
and these in turn give rise to the following $G$-submodules of
$\End\V \cong \V \otimes \V^*$:
\begin{equation}
  \label{eq:submodules-EndV}
  \begin{tikzcd}
    & & W \otimes \V^* \arrow[rd] & \\
    0 \arrow[r] & W \otimes \ann W \arrow[ru] \arrow[rd] & & \V \otimes \V^*,\\
    & & \V \otimes \ann W \arrow[ru] & \\
  \end{tikzcd}
\end{equation}
with arrows depicting inclusions.  This results in a
filtration\footnote{We use the notation $V_1 + V_2$ for the vector
  space sum of subspaces $V_1$ and $V_2$ which need not be direct,
  since $V_1 \cap V_2$ need not be $0$.  Indeed, in the filtration of
  $\End\V$ given by Equation~\eqref{eq:filtration-EndV}, the
  intersection of $W \otimes \V^*$ and $\V \otimes \ann W$ is
  precisely $W \otimes \ann W$.} of $\V\otimes \V^*$:
\begin{equation}
  \label{eq:filtration-EndV}
  0 \subset W \otimes \ann W \subset W \otimes \V^* + \V \otimes \ann W  \subset \V \otimes \V^*.
\end{equation}
By intersecting with $\g$ we get $G$-submodules of $\g \subset \End\V$
\begin{equation}
  \label{eq:submodules-g}
  \begin{tikzcd}
    & & \g \cap (W \otimes \V^*) \arrow[rd] & \\
    0 \arrow[r] & \g \cap (W \otimes \ann W) \arrow[ru] \arrow[rd] & & \g,\\
    & & \g \cap (\V \otimes \ann W) \arrow[ru] & \\
  \end{tikzcd}
\end{equation}
which in turn gives rise to a filtration of $\g$:
\begin{equation}
  \label{eq:filtration-g}
  0 \subset \g \cap (W \otimes \ann W) \subset \g \cap (W \otimes \V^*) + \g \cap (\V \otimes \ann W) \subset \g. 
\end{equation}

Let us now start to understand the structure of $\g$.

\begin{lemma}
  \label{lem:im-X-in-W-in-ker-X}
  Let $X \in \End\V$.
  \begin{enumerate}[label=(\alph*)]
  \item The condition $\im X \subset W$ implies that $\eta^\flat \circ
    X + X^t \circ \eta^\flat = 0$.
  \item The condition $W \subset \ker X$ implies that $X \circ
    h^\sharp + h^\sharp \circ X^t = 0$.
  \end{enumerate}
\end{lemma}

\begin{proof}
  \begin{enumerate}[label=(\alph*)]
  \item Since $W = \ker\eta^\flat$, if $\im X \subset W$, then
    $\eta^\flat \circ X = 0$.  But this means that for all $v, v' \in
    \V$,
    \begin{align*}
      0 &= \eta(v,Xv')\\
        &= \left<\eta^\flat v, X v'\right>\\
        &= \left<X^t \eta^\flat v, v'\right>
    \end{align*}
    so that $X^t \circ \eta^\flat = 0$.
  \item Since $W = \im h^\sharp$, the fact that $W \subset
    \ker X$ says that $X \circ h^\sharp = 0$.  But this means
    that for all $\alpha,\alpha' \in \V^*$,
    \begin{align*}
      0 &= \left<\alpha', X h^\sharp \alpha\right>\\
        &= \left<X^t \alpha', h^\sharp \alpha\right>\\
        &= h(X^t\alpha', \alpha)\\
        &= h(\alpha, X^t\alpha')\\
        &= \left<\alpha, h^\sharp X^t \alpha'\right>
    \end{align*}
    and hence $h^\sharp \circ X^t = 0$.
  \end{enumerate}
\end{proof}

The endomorphisms in $W \otimes \ann W$ -- i.e., those whose image
lies in $W$ and whose kernel contains $W$ -- belong to $\g$  and can
in fact be interpreted as the $p$-brane analogue of Galilean boosts.

\begin{lemma}\label{lem:boosts}
  $W \otimes \ann W \subset \g$.
\end{lemma}

\begin{proof}
  Any endomorphism $X \in W \otimes \ann W$ is such that $W \subset
  \ker X$ and $\im X \subset W$, hence by
  Lemma~\ref{lem:im-X-in-W-in-ker-X} and
  Lemma~\ref{lem:lie-algebra-g}, $X \in \g$.
\end{proof}

The spaces in the above filtration~\eqref{eq:filtration-g} of $\g$ can
be characterised as follows:
\begin{equation}\label{eq:filtered-space-chars}
  \begin{split}
    W \otimes \ann W &= \left\{X \in \End\V ~ \middle |~ \im X \subset
    W \quad\text{and}\quad W \subset \ker X\right\}\\
  \g \cap (W \otimes \V^*)  &= \left\{ X \in \End\V ~\middle |~ \im X \subset W \quad\text{and}\quad X \circ h^\sharp +
    h^\sharp \circ X^t = 0\right\}\\
  \g \cap (\V \otimes \ann W) &= \left\{ X \in \End\V ~\middle |~ W
    \subset \ker X \quad\text{and}\quad X^t \circ \eta^\flat +
    \eta^\flat \circ X = 0\right\}.
  \end{split}
\end{equation}

The fact that $W\subset \V$ is a $G$-submodule can be rephrased as
saying that every $X \in \g$ is a sum $X = Y + Z$ of an endomorphism
$Y$ with $W \subset \ker Y$ and an endomorphism $Z$ with $\im Z
\subset W$.  A consequence of that observation is that whereas $(W
\otimes \V^*) +  (\V \otimes \ann W) \subsetneq \V \otimes \V^*$, it
is nevertheless true that
\begin{equation}
 \g \cap (W \otimes \V^*) + \g \cap (\V \otimes \ann W) = \g.
\end{equation}
It follows that $\g$ is a filtered Lie algebra:
\begin{equation}\label{eq:filtered-g}
    0 \subset W \otimes \ann W \subset \g,
\end{equation}
where $W \otimes \ann W = \g \cap (W \otimes \V^*) \cap (\V \otimes
\ann W)$.

Recall that the filtration of $\V$ induces a filtration of $\V^*$, by
declaring that the dual pairing should have degree zero.  This in turn
induces a filtration of $\End\V \cong \V \otimes \V^*$ and hence of
$\g \subset \End\V$.  One way to rephrase all this is that the
filtration of the Lie algebra $\g$ is compatible with the filtrations
of the modules $\V$ and $\V^*$.  Indeed, letting $\g^0= \g$, $\g^{-1}
= W \otimes \ann W$ and $\g^{-2}=0$ and similarly $\V^0=\V$, $\V^{-1}=W$ and
$\V^{-2}=0$, then it is clear that $\g^i \times
\V^j \to \V^{i+j}$ under the action of $\g$.\footnote{For
  this to hold for every $i,j$, it is necessary to extend the
filtrations in both directions, by declaring $\g^i = \g$ for all
$i\geq 0$ and $\g^i = 0$ for all $i \leq -2$ and similarly for $\V$.
We similarly extend the filtration of $\V^*$ for what follows.}  Similarly, if we let $(\V^*)^1 = \V^*$,
$(\V^*)^0 = \ann W$ and $(\V^*)^{-1} = 0$, then $\g^i \times (\V^*)^j
\to (\V^*)^{i+j}$.  Of course, the adjoint module of $\g$ itself is
filtered: $[\g^i,\g^j] \subset \g^{i+j}$, since $\g$ is a filtered Lie
algebra.

\begin{lemma}
  \label{lem:common-ideal}
  The subspaces $\g \cap (W \otimes \V^*)$ and $\g \cap (\V \otimes
  \ann W)$ are Lie subalgebras of $\g$ containing $W \otimes \ann W$
  as an abelian ideal.
\end{lemma}

\begin{proof}
  This follows easily from the characterisations
  \eqref{eq:filtered-space-chars} of the subspaces $\g\cap(W \otimes
  \V^*)$ and $\g \cap (\V \otimes \ann W)$.  For example, let $X,Y \in
  \g \cap (W \otimes \V^*)$.  Then the image of $[X,Y] = X \circ Y - Y
  \circ X$ is certainly contained in $W$ since this is true for $X$
  and $Y$ separately.  The condition $X \circ h^\sharp +
  h^\sharp \circ X^t = 0$ says that $h$ is invariant, which
  is preserved under commutators of endomorphisms:
  \begin{align*}
    [X,Y]\circ h^\sharp &= X \circ Y \circ h^\sharp - Y  \circ X \circ h^\sharp\\
                             &= - X \circ h^\sharp \circ Y^t + Y \circ h^\sharp \circ X^t\\
                             &= h^\sharp \circ X^t \circ Y^t - h^\sharp \circ Y^t \circ X^t\\
                             &= h^\sharp \circ [X^t, Y^t]\\
                             &= h^\sharp \circ [Y,X]^t\\
                             &= - h^\sharp \circ [X,Y]^t.
  \end{align*}
  The proof that $\g \cap (\V \otimes \ann W)$ is a Lie subalgebra is
  similar, mutatis mutandis. Finally, let us show that both Lie
  algebras have $W \otimes \ann W$ as an abelian ideal. That
  $W \otimes \ann W$ is abelian is clear since if $X,Y
  \in W \otimes \ann W$, then $X \circ Y = 0$.  To show that it is an
  ideal of $\g \cap (\V \otimes \ann W)$, let $X \in \g \cap (\V
  \otimes \ann W)$ and $Y \in W \otimes \ann W$.  Then $X \circ Y = 0$
  and $Y \circ X$ has image in $W$ because $Y$ does and annihilates
  $W$ since $X$ does. Therefore $[X,Y] = - Y \circ X \in W \otimes
  \ann W$.  The proof for $\g  \cap (W \otimes \V^*)$ is similar,
  and results in $[X,Y] = X \circ Y$ for $X \in \g \cap (W \otimes
  \V^*)$ and $Y \in W \otimes \ann W$, which is clearly in $W \otimes
  \ann W$.
\end{proof}

We now identify the quotient Lie algebras of
$\g \cap (W \otimes \V^*)$ and $\g \cap (\V \otimes \ann W)$ by their
ideal $W \otimes \ann W$. As discussed above, $\eta$ induces a
non-degenerate inner product $\overline\eta$ on the quotient vector
space $\V/W$. Let $\so(\V/W,\overline\eta)$ denote the
$\overline\eta$-skewsymmetric endomorphisms of $\V/W$. To be more
concrete, let us denote by $v \mapsto \overline v$ the canonical
surjection $\V \to \V/W$. Then
$\overline\eta (\overline v, \overline v') = \eta(v,v')$, which is
well defined since $W = \ker \eta^\flat$. Let $X \in \End \V$ preserve
$W$, so that $X w \in W$ for all $w \in W$. Then $X$ induces an
endomorphism $\overline X \in \End(\V/W)$ by
$\overline X\, \overline v = \overline{X v}$, which is well-defined
precisely because $X$ preserves $W$. Moreover, all endomorphisms of
$\V/W$ are of the form $\overline X$ for some $X \in \End\V$
preserving $W$. Finally, $\so(\V/W,\overline\eta)$ consists of those
endomorphisms $\overline X$ of $\V/W$ such that
$\overline\eta(\overline X \, \overline v, \overline v') = -
\overline\eta(\overline v, \overline X \, \overline v')$ for all
$\overline v, \overline v' \in \V/W$. Similarly, $\so(W,\gamma)$ are
those endomorphisms $Y$ of $W$ such that
$\gamma(Y w, w') = - \gamma(w,Y w')$for all $w,w' \in W$.

\begin{proposition}
  \label{prop:so-isos}
  There are Lie algebra isomorphisms
  \begin{equation}
    \label{eq:so-isos}
    \frac{\g \cap (W \otimes \V^*)}{W \otimes \ann W}  \cong
    \so(W,\gamma) \qquad\text{and}\qquad
    \frac{\g \cap (\V \otimes \ann W)}{W \otimes \ann W}  \cong \so(\V/W,\overline\eta).
  \end{equation}
\end{proposition}

\begin{proof}
  Let $X : \V \to W$ and let $X_| \in \End W$ denote its restriction
  to $W$.  This defines a surjective map $\Hom(\V,W) \to \End W$ with
  kernel $W \otimes \ann W$.  This map is actually a Lie algebra
  homomorphism.  Indeed, if we let $X,Y \in \Hom(\V,W)$, then
  \begin{align*}
    [X_|,Y_|] &= X_| \circ Y_| - Y_| \circ X_| \\
              &= X \circ Y_| - Y \circ X_| & \tag{since $\im X \subset W$ and $\im Y \subset W$}\\
              &= (X \circ Y)_| - (Y \circ X)_|\\
              &= [X,Y]_|.
  \end{align*}
  If in addition $X \in \g$, so that $X \in \g \cap (W
  \otimes \V^*)$ then $X_| \in \so(W,\gamma)$. Indeed, let $w =
  h^\sharp \alpha$ and $w'= h^\sharp \alpha'$, then
  \begin{equation*}
    X w = X  h^\sharp \alpha = (X \circ
    h^\sharp)(\alpha) = - (h^\sharp \circ X^t)(\alpha) = -
    h^\sharp(X^t\alpha)
  \end{equation*}
  and hence
  \begin{align*}
    \gamma(X w,w') &= - \gamma(h^\sharp(X^t\alpha), h^\sharp(\alpha'))\\
                         &= - \left<X^t\alpha, h^\sharp\alpha'\right> & \tag{by definition of $\gamma$}\\
                         &= - \left<\alpha, X h^\sharp\alpha'\right>& \tag{by definition of $X^t$}\\
                         &= \left<\alpha, h^\sharp (X^t \alpha')\right> &\tag{since $X \in \g$}\\
                         &= \gamma(h^\sharp\alpha, h^\sharp(X^t \alpha')) & \tag{by definition of $\gamma$}\\
                         &= -\gamma(w, X w').
  \end{align*}
  Furthermore, we claim that the linear map $\g \cap (W \otimes \V^*) \to
  \so(W,\gamma)$ is surjective.  Indeed, let $X_| \in \so(W,\gamma)$
  for some $X \in W \otimes \V^*$ where $X$ is defined up to $W \otimes
  \ann W$.  We want to show that $X \in \g \cap (W \otimes \V^*)$,
  which, by the characterisation \eqref{eq:filtered-space-chars} and
  since $\im X \subset W$, is tantamount to showing that $X \circ
  h^\sharp + h^\sharp \circ X^t = 0$.  The first observation
  is that if it holds for some $X \in W \otimes \V^*$, it holds for $X +
  B$ for any $B \in W \otimes \ann W$, since $B \circ h^\sharp =
  0$ (since $\im h^\sharp  = W$) and $h^\sharp \circ B = 0$
  (since $\ker h^\sharp = \ann W$).  So it is enough to show it
  for any such $X$.  Since the restriction to $W$ lies in
  $\so(W,\gamma)$, we have that for all $w,w' \in W$,
  \begin{equation*}
    \gamma(w,Xw') + \gamma(w',Xw)= 0.
  \end{equation*}
  Since $h^\sharp : \V^* \to W$ is surjective, this is equivalent
  to
  \begin{equation}
    \label{eq:X-in-so}
    \gamma(h^\sharp \alpha, X h^\sharp \beta) +
    \gamma(h^\sharp \beta, X h^\sharp \alpha)=0
  \end{equation}
  for all $\alpha,\beta \in \V^*$.  Therefore,
  \begin{align*}
    0 &= \gamma(h^\sharp \alpha, X h^\sharp \beta) +
    \gamma(h^\sharp \beta, X h^\sharp \alpha) &\tag{by Equation~\eqref{eq:X-in-so}}\\
      &= \left<\alpha, Xh^\sharp \beta\right> + \left<\beta,X h^\sharp\alpha\right> &\tag{by definition of $\gamma$}\\
      &= \left<\alpha, Xh^\sharp \beta\right> + \left<X^t\beta, h^\sharp\alpha\right> &\tag{by definition of $X^t$}\\
      &= \left<\alpha, Xh^\sharp \beta\right> + \left<\alpha, h^\sharp X^t\beta\right> & \tag{since $h$ is symmetric}\\
      &= \left<\alpha, (X \circ h^\sharp + h^\sharp \circ X^t)\beta\right>.
  \end{align*}
  Since $\alpha,\beta \in \V^*$ are arbitrary, this shows that $X
  \circ h^\sharp + h^\sharp \circ X^t = 0$ and hence $X \in
  \g \cap (W \otimes \V^*)$.  In summary, we have a surjective Lie
  algebra homomorphism $\g \cap (W \otimes \V^*) \to \so(W,\gamma)$
  with kernel the ideal $W \otimes \ann W$, which proves the first
  isomorphism.

  To prove the second isomorphism, we observe that we have an
  isomorphism $\V \otimes \ann W \cong \Hom(\V/W,\V)$.  Composing with
  the projection $\V \to \V/W$ we obtain a surjective map $\V \otimes
  \ann W \to \End(\V/W)$, sending $X \mapsto \overline X$ with kernel
  $W \otimes \ann W$.   This map is again a Lie algebra homomorphism.
  Indeed, if $X,Y \in \V \otimes \ann W$, then
  \begin{align*}
    [\overline X, \overline Y] \overline v &= \overline X (\overline Y  (\overline v))  - \overline Y (\overline X  (\overline v))\\
                                           &= \overline X (\overline {Y v})  - \overline Y (\overline {Xv}) &\tag{by definition of $\overline X,\overline Y$}\\
                                           &= \overline {X Y v}  - \overline {Y X v} &\tag{by definition of $\overline X,\overline Y$}\\
                                           &= \overline {[X,Y]v}\\
                                           &= \overline {[X,Y]} \overline v. &\tag{by definition of $\overline{[X,Y]}$}
  \end{align*}
  Since this holds for every $\overline v$, we may abstract it and hence $[\overline X, \overline Y] = \overline{[X,Y]}$.
  
  If now $X \in \g \cap (\V \otimes \ann W)$, then for all $v, v' \in \V$,
    \begin{align*}
    \overline\eta(\overline X \overline v, \overline{v'}) &= \overline\eta(\overline{X  v}, \overline{v'})\\
                                                                &= \eta(X  v, v')\\
                                                                &= - \eta(v, X v') &\tag{since $X \in \g$}\\
                                                                &= - \overline\eta(\overline v, \overline{X v'})\\
                                                                &= - \overline\eta(\overline v, \overline X \cdot \overline{v'}),
  \end{align*}
  so that $\overline X \in \so(\V/W,\overline\eta)$. Furthermore, we
  claim that this map $\g \cap (\V \otimes \ann W) \to \so(\V/W,\overline\eta)$ is
  surjective.  Every $\overline X \in \End(\V/W)$ comes
  from some $X \in \V \otimes\ann W$, which is only defined up to the
  addition of some $B \in W \otimes \ann W$.  We want to show that if
  $\overline X \in \so(\V/W,\overline\eta)$, then $X \in \g \cap (\V
  \otimes \ann W)$.  From the characterisation
  \eqref{eq:filtered-space-chars} and since for $X \in \V \otimes \ann
  W$, it follows that $\ker X \supset W$, we need only show that $X^t
  \circ \eta^\flat + \eta^\flat \circ X = 0$.  As before, we observe
  that if this holds for some $X \in \V \otimes \ann W$, it holds for
  $X + B$ for any $B \in W \otimes \ann W$.  This is because
  $B^t\circ\eta^\flat = 0$ (since $\im \eta^\flat = \ann W \subset
  \ker B^t$) and $\eta^\flat \circ B = 0$ (since $\im B \subset W =
  \ker\eta^\flat$).  Now let $u,v \in \V$ be arbitrary and consider
  \begin{align*}
    \left<(X^t \circ \eta^\flat + \eta^\flat \circ X) u, v\right> &= \left<X^t \eta^\flat u, v\right> + \left<\eta^\flat X u, v\right>\\
          &= \left<\eta^\flat u, X v\right> + \left<\eta^\flat X u, v\right> &\tag{by definition of $X^t$}\\
          &= \eta(u, X v) + \eta(Xu, v) &\tag{by definition of $\eta$}\\
          &= \overline\eta(\overline u, \overline{X v}) + \overline\eta(\overline{Xu}, \overline v) &\tag{by definition of $\overline\eta$}\\
          &= \overline\eta(\overline u, \overline X \overline v) + \overline\eta(\overline X \overline u, \overline v) &\tag{by definition of $\overline X$}\\
          &= 0. \tag{since $\overline X \in \so(\V/W,\overline \eta)$}
  \end{align*}
  Therefore we have established a surjective Lie algebra homomorphism
  $\g \cap (\V \otimes \ann W) \to \so(\V/W,\overline\eta)$ with
  kernel the ideal $W \otimes \ann W$, which proves the second
  isomorphism.
\end{proof}

Finally, we can characterise the Lie algebra $\g$.

\begin{proposition}
  \label{prop:g-as-lie-algebra}
  The Lie algebra $\g$ is isomorphic to an abelian extension
  \begin{equation}
    \label{eq:Lie-algebra-g-iso}
    \begin{tikzcd}
      0 \arrow[r] & \Hom(\V/W,W) \arrow[r] & \g \arrow[r] &
      \so(W,\gamma)\oplus \so(\V/W,\overline\eta)\arrow[r] & 0,
    \end{tikzcd}
  \end{equation}
  where the action of $(X_|,\overline Y) \in \so(W,\gamma)\oplus \so(\V/W,\overline\eta)$
  on $Z \in \Hom(\V/W,W)$ is given by
  \begin{equation}
    [(X_|,\overline Y),Z] = X \circ Z - Z \circ Y.
  \end{equation}
\end{proposition}

\begin{proof}
  Recall that $\g$ admits the filtration~\eqref{eq:filtered-g} and that a
  filtered Lie algebra is isomorphic (as a vector space) to its
  associated graded Lie algebra, giving a vector space isomorphism
  \begin{equation}
    \label{eq:g-vs-iso}
    \begin{split}
      \g & \cong (W \otimes \ann W) \oplus \frac{\g \cap (W \otimes \V^*)}{W
        \otimes \ann W} \oplus \frac{\g \cap (\V \otimes \ann W)}{W
        \otimes \ann W}\\
      & \cong \Hom(\V/W,W) \oplus \so(W,\gamma)\oplus \so(\V/W,\overline\eta),
    \end{split}
  \end{equation}
  where we have used Proposition~\ref{prop:so-isos}.

  We show that $W \otimes \ann W$ is an ideal of $\g$.  Indeed, let $X \in \g$ and $Y
  \in W \otimes \ann W$.  Then
  \begin{equation*}
    [X,Y] = X \circ Y - Y \circ X.
  \end{equation*}
  Since $\im Y \subset W$ and $W \subset \V$ is a submodule, we see
  that $\im [X,Y] \subset W$.  If $w \in W$,
  \begin{equation*}
    [X,Y]w = X (Y w) - Y(X w).
  \end{equation*}
  The first term is zero because $W \subset \ker Y$ and the second
  term is zero because, in addition, $X w \in W$.  Therefore $W
  \subset \ker[X,Y]$ and hence $[X,Y] \in W \otimes \ann W$.
  
  The subalgebras $\g \cap (\V \otimes \ann W)$ and
  $\g \cap (W \otimes \V^*)$ do not commute: if
  $X \in \g\cap(\V \otimes \ann W)$ and
  $Y \in \g\cap(W \otimes \V^*)$, then since $\im Y \subset \ker X$,
  we have $X \circ Y = 0$ and hence their commutator is $[X,Y] = -Y
  \circ X$, which belongs to the ideal $W \otimes \ann W$. Therefore
  they do commute when we quotient by the ideal.

  Using the shorthand $A = \g \cap (\V \otimes \ann W)$ and $B= \g
  \cap (W \otimes \V^*)$ and observing that $W\otimes\ann W = A \cap
  B$, we have the following exact sequence
  \begin{equation*}
    \begin{tikzcd}
      0 \arrow[r] & A\cap B \arrow[r] & A + B \arrow[r] & (A+B)/B \oplus (A+B)/A \arrow[r] & 0,
    \end{tikzcd}
  \end{equation*}
  where the middle map is made out of the two canonical surjections
  $A+B \to  (A+B)/B$ and $A+B \to (A+B)/A$.  Applying twice the Second
  Isomorphism Theorem to conclude that $(A+B)/B \cong A/(A\cap B)$ and
  $(A+B)/A \cong B/(A\cap B)$, we obtain
  \begin{equation*}
    \g/(W \otimes \ann W) = (A+B)/(A \cap B) \cong (A/A\cap B) \oplus (B/A\cap B),
  \end{equation*}
  so that by Proposition~\ref{prop:so-isos} the quotient is the direct sum
 $\so(W,\gamma)\oplus \so(\V/W,\overline\eta)$.

  Finally, that the action of $\so(W,\gamma)\oplus
  \so(\V/W,\overline\eta)$ on $W \otimes \ann W$ is as stated, follows
  from the proof of Lemma~\ref{lem:common-ideal}.
\end{proof}

The ideal $\Hom(\V/W,W)$ is abelian, whereas when
$\dim W,\dim \V/W \geq 3$, the Lie algebra
$\so(W,\gamma) \oplus \so(\V/W,\overline\eta)$ is semisimple.
Therefore, the extension \eqref{eq:Lie-algebra-g-iso} splits: indeed,
$\Hom(\V/W,W)$ is the radical of $\g$ and it follows from the
Levi--Malcev Theorem that $\g$ splits over its radical and the Levi
factor (i.e., a Lie subalgebra of $\g$ isomorphic to
$\so(W,\gamma) \oplus \so(\V/W,\overline\eta)$) is only defined up to
conjugation. Even if either $\dim W =2$ or $\dim \V/W =2$, so that
$\so(W,\gamma) \oplus \so(\V/W,\overline\eta)$ fails to be semisimple,
the sequence again splits, but the splitting in any case is never
canonical. Every such splitting implies a choice of complementary
subspace to $W$ in $\V$.  Without making this choice,
Proposition~\ref{prop:g-as-lie-algebra} is the best description
possible of the Lie algebra of $G \subset \GL(\V)$.

\subsection{\texorpdfstring{$G$-structures and adapted connections}{G-structures and adapted connections}}
\label{sec:g-structures-adapted-connections}

We will briefly review the basic definitions of $G$-structure and
adapted connections and use them to state the problem we wish to
solve.  We shall make a simplifying assumption on our manifolds in
that they admit a $G$-structure which reduces to the connected
component of the identity of the group we have been discussing.  This
would be automatic if, for example, $M$ were simply-connected.

Hence from now on we will let $G$ be the \emph{connected} subgroup of
$\GL(\V)$ generated by $\g$; that is, the identity component of what
we used to call $G$ in the previous section. We now let $M$ be a
$D$-dimensional manifold admitting a $G$-structure $\pi: P \to M$. In
other words, $P$ is a (right) $G$-principal subbundle of the bundle of
frames of $M$. As reviewed in \cite{Figueroa-OFarrill:2020gpr}, we
will go back and forth between representations $\EE$ of $G$ and the
corresponding associated vector bundles $P\times_G \EE \to M$, and
between $G$-equivariant linear maps $\phi: \EE \to \FF$ and the
corresponding bundle maps $\Phi: P\times_G \EE \to P\times_G \FF$. An
Ehresmann connection $\omega\in \Omega^1(P;\g)$ on $P$ gives rise to a
Koszul connection on any associated vector bundle $P\times_G \EE$. In
particular, if we take $\EE= \V$, we get a Koszul connection on the
fake tangent bundle $P\times_G \V$. The soldering form
$\theta \in \Omega^1(P;\V)$ defines an isomorphism between
$P\times_G \V$ and the tangent bundle $TM$ and the Koszul connection
on $P\times_G \V$ induces an affine connection $\nabla$ on $TM$ which
is said to be \emph{adapted} to the $G$-structure. Adapted connections
parallelise the characteristic tensor fields of the $G$-structure.

Let $\d : \Hom(V, \g) \to \Hom(\wedge^2\V, \V)$ denote the \emph{Spencer
differential}, defined as the composition
\begin{equation}
  \label{eq:spencer-differential}
  \begin{tikzcd}
    \Hom(\V,\g) \arrow[r,"\cong"] & \g \otimes \V^* \arrow[r,"\rho
    \otimes \id_{\V^*}"] & \V \otimes \V^* \otimes \V^*
    \arrow[r,"\id_\V \otimes \wedge^2"] & \V \otimes \wedge^2 \V^*
    \arrow[r,"\cong"] & \Hom(\wedge^2 \V, \V),
  \end{tikzcd}
\end{equation}
where $\rho : \g \to \V \otimes \V^*$ is the representation map $\g
\to \End \V$ (which is just the inclusion) composed with the
isomorphism $\End \V \cong \V \otimes \V^*$.

\begin{lemma}
  The linear map $\d : \Hom(\V,\g) \to \Hom(\wedge^2\V, \V)$ defined in
  Equation~\eqref{eq:spencer-differential} is $G$-equivariant.
\end{lemma}

\begin{proof}
  We argue that all the maps in
  Equation~\eqref{eq:spencer-differential} are $G$-equivariant.  This
  is clear for the isomorphisms $\Hom(\V,\g) \cong \g \otimes \V^*$
  and $\V \otimes \wedge^2\V^* \cong \Hom(\wedge^2\V ,\V)$ and
  similarly for the map $\id_{\V} \otimes \wedge^2$, since the
  identity and the alternation $\wedge^2$ are $\GL(\V)$-equivariant,
  so in particular equivariant under any subgroup, such as $G$.
  It remains to show that $\rho : \g \to \V \otimes \V^*$ is
  $G$-equivariant.  This is equivalent to showing that the
  inclusion $\g \to \End\V$, is $G$-equivariant, but this is almost a
  tautology.  Indeed, let $X \in \g$ and $g \in G$.  Then $\Ad_g X = g
  X g^{-1}$, and the LHS can be interpreted as the inclusion of $g
  \cdot X$ in $\End\V$, whereas the RHS is the composition of the
  inclusion of $X$ with the action of $g$ on $\End \V$.  The equality
  between the LHS and the RHS is precisely $G$-equivariance.
\end{proof}

The explicit form of the Spencer differential on $\kappa : \V \to \g$
is given by $\d\kappa : \wedge^2 \V \to \V$, where for all $v,v' \in
\V$,
\begin{equation}
  \label{eq:spencer-explicit}
  (\d\kappa)(v \wedge v') = \kappa(v)v' - \kappa(v')v,
\end{equation}
and we extend linearly to all of $\wedge^2\V$.  Being a linear map and
$G$-equivariant, the Spencer differential belongs to an exact sequence
of $G$-modules:
\begin{equation}
  \label{eq:intrinsic-torsion-exact-sequence}
  \begin{tikzcd}
    0 \arrow[r] & \ker\d \arrow[r] & \Hom(\V,\g) \arrow[r,"\d"] &
    \Hom(\wedge^2\V, \V) \arrow[r, "\pr"] & \coker \d \arrow[r] & 0,
  \end{tikzcd}
\end{equation}
resulting in a corresponding exact sequence of associated vector
bundles.  The bundle corresponding to $\Hom(\V,\g)$ is the bundle of
one-forms with values in the adjoint bundle $\ad P$: this is the
bundle of differences between two adapted connections.  The bundle
associated to $\Hom(\wedge^2\V, \V)$ is the bundle of $2$-forms with
values in $TM$, where the torsion tensor af an affine connection
lives.  The image under (the bundle version of) the Spencer
differential of $\kappa\in\Omega^1(M;\ad P)$ is the effect on the
torsion of an adapted connection of modifying the connection by adding
$\kappa$.  Symbolically, $T^{\nabla + \kappa} = T^\nabla + \d \kappa$.
Therefore the bundle corresponding to $\coker\d$ is the bundle of
intrinsic torsions: the part of the torsion which is independent of
the connection and hence intrinsic to the $G$-structure.

We now proceed to determine the $G$-module structure of $\coker\d$
more precisely, in order to arrive at the $G$-submodules of $\coker\d$
and hence at the intrinsic torsion classes.

\subsection{\texorpdfstring{Some filtered $G$-modules}{Some filtered G-modules}}
\label{sec:some-filtered-g}

We now start the study of the $G$-modules appearing in the intrinsic
torsion exact sequence \eqref{eq:intrinsic-torsion-exact-sequence}.
We will first exhibit them as filtered $G$-modules.

We have already seen three filtered $G$-modules: $\V$, $\V^*$  and
$\g$, where the $G$-invariant filtrations are given as follows:
\begin{equation}
  \label{eq:filtered-G-modules}
    0 \subset W \subset \V,
  \qquad
  0 \subset \ann W \subset \V^*
  \qquad\text{and}\qquad
    0 \subset W \otimes \ann W \subset \g.
\end{equation}
These filtrations give rise to the following $G$-submodules of
$\Hom(\V,\g) \cong \g \otimes \V^*$:
\begin{equation}\label{eq:hom-V-g-filtered}
  \begin{tikzcd}
      & & \g \otimes \ann W \arrow[rd] & \\
    0 \arrow[r] & W \otimes \ann W \otimes \ann W \arrow[ru] \arrow[rd] & & \g \otimes \V^*\\
      & & W \otimes \ann W \otimes \V^* \arrow[ru] & \\
  \end{tikzcd}
\end{equation}
and of $\Hom(\wedge^2\V,\V) \cong \V \otimes \wedge^2\V^*$:
\begin{equation}
  \begin{tikzcd}\label{eq:torsions-filtered}
      & & & W \otimes \wedge^2\V^* \arrow[rd] & \\
      & & W \otimes \ann W \wedge \V^* \arrow[ru] \arrow[rd] & & \V \otimes \wedge^2\V^*\\
    0 \arrow[r] & W \otimes \wedge^2\ann W \arrow[ru] \arrow[rd] & & \V
    \otimes \ann W \wedge \V^* \arrow[ru]\\
      & & \V \otimes \wedge^2\ann W \arrow[ru] & & \\
  \end{tikzcd}.
\end{equation}
We have written these submodules as submodules of $\g \otimes \V^*$
and $\V \otimes \wedge^2\V^*$, respectively, but it is convenient to
understand them in terms of linear maps in $\Hom(\V,\g)$
and $\Hom(\wedge^2\V,\V)$.  This is easy to do by inspection,
obtaining the following characterisations of the $G$-submodules of
$\Hom(\V,\g)$:
\begin{equation} 
  \label{eq:hom-v-g-submodules}
  \begin{split}
    \g \otimes \ann W &= \left\{\hspace{0.25em}\kappa \in \Hom(\V,\g)  ~\middle |~ \kappa(w) = 0,\quad \forall~ w \in W\right\}\\
    W \otimes \ann W \otimes \V^* &=  \left\{\hspace{0.25em}\kappa \in \Hom(\V,\g)  ~\middle |~ \kappa(v)w = 0\quad\text{and}\quad \kappa(v)v' \in W\quad\forall~ v,v' \in \V,~w \in W\right\}\\
    W \otimes \ann W \otimes \ann W &= \left\{\kappa \in \Hom(\V,\g)  ~\middle | \hspace{-0.25em} \begin{array}{l} \kappa(w)=0,\\ \kappa(v)w = 0\quad\text{and} \quad \kappa(v)v' \in W\end{array}~ \hspace{0.2em} \forall~ v,v' \in \V,~w \in W\right\},
  \end{split}
\end{equation}
and those of $\Hom(\wedge^2\V,\V)$:
\begin{equation}
  \label{eq:torsions-submodules}
  \begin{split}
    W \otimes \wedge^2\V^* &= \left\{\hspace{0.18em} T\in\Hom(\wedge^2\V,\V) ~\middle |~ T(v \wedge v') \in W, \quad\forall~v,v' \in \V\right\}\\
    \V \otimes \ann W \wedge \V^* &= \left\{\hspace{0.18em} T\in\Hom(\wedge^2\V,\V) ~\middle |~ T(w \wedge w') =0, \quad\forall~w,w' \in W\right\}\\
    W \otimes \ann W \wedge \V^* &= \left\{ T\in\Hom(\wedge^2\V,\V) ~\middle |~ \hspace{-0.45em} \begin{array}{l} T(w \wedge w') =0\quad\text{and}\\ T(v\wedge v') \in W,\end{array}~ \quad\forall~w,w' \in W,~v,v' \in \V\right\}\\
    \V \otimes \wedge^2\ann W &= \left\{\hspace{0.18em} T\in\Hom(\wedge^2\V,\V) ~\middle |~ T(v \wedge w)=0, \quad\forall~v \in \V,~w \in W\right\}\\
    W \otimes \wedge^2 \ann W &= \left\{ T\in\Hom(\wedge^2\V,\V) ~\middle |~ \hspace{-0.45em} \begin{array}{l} T(v \wedge w) =0\quad\text{and}\\ T(v\wedge v') \in W,\end{array}~ \quad\forall~w \in W,~v,v' \in \V\right\}.
  \end{split}
\end{equation}

It follows from the description of the $G$-submodules of $\Hom(\V,\g)$
in Equation~\eqref{eq:hom-v-g-submodules} that
\begin{equation}
  \left( \g \otimes \ann W \right) \cap \left( W \otimes \ann W \otimes \V^* \right) = W \otimes \ann W \otimes \ann W
\end{equation}
and hence we may derive the following filtration of $\g \otimes \V^*$ from \eqref{eq:hom-V-g-filtered}:
\begin{equation}
  \label{eq:hom-v-g-filtration}
  0 \subset W \otimes \ann W \otimes \ann W \subset \left( \g
    \otimes \ann W  + W \otimes \ann W \otimes \V^* \right) \subset \g \otimes \V^*.
\end{equation}

Similarly, it follows from the description of the $G$-submodules of
$\Hom(\wedge^2\V,\V)$ given in \eqref{eq:torsions-submodules} that
\begin{equation}
  W \otimes \ann W \wedge \V^* = \left(W \otimes \wedge^2\V^*\right) \cap \left(\V  \otimes \ann W \wedge \V^*\right)
\end{equation}
and that
\begin{equation}
  W \otimes \wedge^2\ann W = \left( W \otimes \ann W \wedge \V^* \right) \cap \left( \V \otimes \wedge^2\ann W \right).
\end{equation}
This results in the following filtration of $\V \otimes \wedge^2\V^*$:
\begin{multline}
  \label{eq:torsions-filtration}
    0 \subset W \otimes \wedge^2\ann W \subset  \left( W \otimes \ann
      W \wedge \V^* + \V \otimes \wedge^2\ann W \right)\\
    \subset \left( W \otimes \wedge^2\V^* + \V \otimes \ann W \wedge
      \V^* \right) \subset  \V \otimes \wedge^2\V^*.
\end{multline}
  
As discussed in Section~\ref{sec:filtrations}, every filtered
$G$-module has an associated graded module which is the direct sum of
the consecutive quotients in the filtration.  Of course a filtered
$G$-module and its associated graded $G$-module are not necessarily
isomorphic as $G$-modules; although they are of course isomorphic as
vector spaces.  Therefore their utility in this paper lies mostly in
counting dimensions.

It may help the discussion to assign explicit filtration degrees.
Recall that we have filtered $\V$ as $\V = \V^0 \supset \V^{-1} = W
\supset \V^{-2} = 0$ and $\V^*$ as $\V^* = (\V^*)^1 \supset (\V^*)^0 =
\ann W \supset (\V^*)^{-1} = 0$.  The associated graded modules are
$\gr \V = \V_0 \oplus \V_{-1}$ with $\V_0 = V$ and $\V_{-1} = W$,
where we have introduced the shorthand $V = \V/W$, and $\gr \V^* =
\V^*_1 \oplus \V^*_0$, with $\V^*_1 = \V/\ann W \cong W^*$ and $\V_0^*
= \ann W \cong V^*$.  These assignments then force $\eta$ (and hence
$\eta^\flat$) to have degree $0$ and $h$ (and hence
$h^\sharp$) to have degree $-2$.  The filtration of $\g$ is
induced from that of $\V \otimes \V^*$ by demanding that the
representation map $\g \to \V\otimes \V^*$ should preserve the
filtration.  Its associated graded Lie algebra is $\gr \g = \g_0 \oplus
\g_{-1}$, with $\g_{-1} = W \otimes V^*$ and $\g_0 = \so(W) \oplus
\so(V)$, where we have used the isomorphism $\ann W \cong V^*$ and
abbreviated $\so(W,\gamma)$ to $\so(W)$ and $\so(V,\overline\eta)$ to
$\so(V)$.

We may work out $\gr (\g \otimes \V^*)$ simply by declaring it to
inherit the grading from $\g$ and $\V^*$, resulting in
\begin{equation}
  \begin{split}
    \gr (\g \otimes \V^*) &= (\g_0 \oplus \g_{-1}) \otimes (\V^*_1 \oplus \V^*_0)\\
    &= (\g_0 \otimes \V_1^*) \oplus (\g_0 \otimes \V_0^* \oplus \g_{-1} \otimes \V_1^*) \oplus (\g_{-1} \otimes \V_0^*),
  \end{split}
\end{equation}
from where we read off (now reverting to writing $\g \otimes \V^*$ as $\Hom(\V,\g)$)
\begin{equation}\label{eq:explicit-grading-hom-v-g}
  \gr \Hom(\V,\g) = \bigoplus_{i=-1}^1 \Hom(\V,\g)_i,
\end{equation}
where
\begin{equation}\label{eq:hom-v-g-graded-components}
  \begin{split}
    \Hom(\V,\g)_{-1} &= \Hom(V,\Hom(V,W))\\
    \Hom(\V,\g)_{0} &= \Hom(W,\Hom(V,W)) \oplus \Hom(V,\so(V) \oplus \so(W))\\
    \Hom(\V,\g)_{1} &= \Hom(W, \so(V) \oplus \so(W)).
  \end{split}
\end{equation}
Similarly, we grade $\wedge^2\V^*$ as
\begin{equation}
  \gr \wedge^2\V^* = \wedge^2(\V_1^* \oplus \V_0^*) = \wedge^2 \V_1^*
  \oplus (\V_1^* \wedge \V_0^*) \oplus \wedge^2 \V_0^* = \wedge^2 W^* \oplus (V^* \wedge W^*) \oplus \wedge^2 V^*
\end{equation}
and in turn this allows us to grade $\V \otimes \wedge^2\V^*$ as
\begin{equation}
  \begin{split}
    \gr (\V \otimes \wedge^2\V^*) &= (\V_0 \oplus \V_{-1}) \otimes  ((\wedge^2\V^*)_2\oplus (\wedge^2\V^*)_1 \oplus (\wedge^2 \V^*)_0)\\
    &= \left(\V_0 \otimes (\wedge^2\V^*)_2\right) \oplus \left( \V_{-1}\otimes (\wedge^2\V^*)_2 \oplus \V_0 \otimes (\wedge^2\V^*)_1\right)\\
    & \qquad {} \oplus \left( \V_0\otimes (\wedge^2\V^*)_0 \oplus \V_{-1} \otimes (\wedge^2\V^*)_1 \right) \oplus (\V_{-1} \otimes (\wedge^2\V^*)_0),\\
  \end{split}
\end{equation}
from where we read off
\begin{equation}\label{eq:explicit-grading-torsions}
  \gr \Hom(\wedge^2\V,\V) = \bigoplus_{i=-1}^2 \Hom(\wedge^2\V,\V)_i,
\end{equation}
where
\begin{equation}\label{eq:torsions-graded-components}
  \begin{split}
    \Hom(\wedge^2\V,\V)_{-1} &= \Hom(\wedge^2V,W)\\
    \Hom(\wedge^2\V,\V)_0 &= \Hom(V \wedge W, W) \oplus \Hom(\wedge^2V, V)\\
    \Hom(\wedge^2\V,\V)_1 &= \Hom(\wedge^2 W, W) \oplus \Hom(V \wedge W, V)\\
    \Hom(\wedge^2\V,\V)_2 &= \Hom(\wedge^2W, V).
  \end{split}
\end{equation}

This now allows us to determine the kernel and cokernel of the Spencer
differential as $G$-modules.  We will first deal with the kernel.

\subsection{The kernel of the Spencer differential}
\label{sec:kern-spenc-diff}

Since the Spencer differential $\d : \Hom(\V,\g) \to
\Hom(\wedge^2\V,\V)$ is $G$-equivariant, we may restrict it to the
$G$-submodules of $\Hom(\V,\g)$ listed in
Equation~\eqref{eq:hom-v-g-submodules}.

\begin{lemma}\label{lem:ker-d1-2-3}
  The Spencer differential induces the following $G$-equivariant linear
  maps:
  \begin{equation}
    \label{eq:d1-2-3}
    \begin{tikzcd}
      W \otimes \ann W \otimes \ann W \arrow[r,"\d_1"] & W \otimes \wedge^2 \ann W\\
      W \otimes \ann W \otimes \V^* \arrow[r,"\d_2"] & W \otimes \ann W \wedge \V^*\\
      \g \otimes \ann W \arrow[r,"\d_3"] & \V \otimes \ann W \wedge \V^*,
    \end{tikzcd}
  \end{equation}
  whose kernels are given by
  \begin{equation}
    \label{eq:kernel-submodules}
    \begin{split}
      \ker \d_1 &= \left\{\kappa \in \Hom(\V,\g) ~\middle |~ \kappa(w)=0\quad\text{and}\quad \kappa(v)v'=\kappa(v')v \in W,~\forall~w\in W,~v,v' \in \V\right\}\\
      \ker \d_2 &=\left\{\kappa \in \Hom(\V,\g) ~\middle |~ \kappa(w)w'=0\quad\text{and}\quad \kappa(v)v'=\kappa(v')v \in W,~\forall~w,w'\in W,~v,v' \in \V\right\}\\
    \ker\d_3 &= \left\{\kappa \in \Hom(\V,\g) ~\middle |~ \kappa(w)=0\quad\text{and}\quad \kappa(v)v' = \kappa(v')v,~\forall w\in W,~v,v'\in\V\right\},
    \end{split}
  \end{equation}
  which are $G$-submodules of $\ker\d$.  Furthermore, $\d_1$ is
  surjective and $\ker \d_1 = \ker\d_2 \cap \ker \d_3$.
\end{lemma}

\begin{proof}
  We must work out the codomains of the linear maps $\d_1,\d_2,\d_3$ in
  Equation~\eqref{eq:d1-2-3} and their kernels.  We do this one map at
  a time.

  \begin{enumerate}[label=($\d_{\arabic*}$)]
  \item Suppose that $\kappa \in W \otimes \ann W \otimes \ann W$, so that
    $\kappa(w)=0$, $\kappa(v)w = 0$ and $\kappa(v)v' \in W$ for all
    $v,v' \in \V$ and $w \in W$.  Let $T = \d\kappa$.  Then
    \begin{align*}
      T(v \wedge w) &= \kappa(v)w - \kappa(w)v = 0 &\tag{since $\kappa(w) = 0$ and $\kappa(v)w=0$ $\forall v\in \V, w\in W$}\\
      T(v\wedge v') &= \kappa(v)v' - \kappa(v')v \in W, &\tag{since $\kappa(v)v' \in W$ for all $v,v' \in \V$}
    \end{align*}
    so that $T \in W \otimes \wedge^2\ann W$.  Therefore
    $\d_1 : W \otimes \ann W \otimes \ann W \to W \otimes \wedge^2\ann
    W$ is simply the skew-symmetrisation of the last two tensorands,
    which is clearly surjective with kernel the symmetrisation $W \to
    \odot^2\ann W$.

  \item Now suppose that $\kappa \in W \otimes \ann W \otimes \V^*$, so that
    $\kappa(v)w = 0$ and $\kappa(v)v' \in W$ for all $v,v'\in \V$ and $w
    \in W$.  Let $T = \d\kappa$.  Then
    \begin{align*}
      T(w \wedge w') &= \kappa(w)w' - \kappa(w')w = 0 &\tag{since $\kappa(v)w=0$ for all $v \in \V$ and $w\in W$}\\
      T(v \wedge v') &= \kappa(v)v' - \kappa(v')v \in W, & \tag{since $\kappa(v)v' \in W$ for all $v,v' \in \V$}
    \end{align*}
    so that $T \in W \otimes \ann W \wedge \V^*$.  Therefore
    $\d_2: W \otimes \ann W \otimes \V^* \to W \otimes \ann W \wedge
    \V^*$.
    \item Finally, let $\kappa \in \g \otimes \ann W$, so that $\kappa(w) =
    0$.  Let $T = \d\kappa$.  Then
    \begin{align*}
      T(w \wedge w') &= \kappa(w)w' - \kappa(w')w = 0. &\tag{since $\kappa(w)=0$ for all $w \in W$}
    \end{align*}
    Therefore $T \in \V \otimes \ann W \wedge \V^*$ and hence $\d_3: \g
    \otimes \ann W \to \V \otimes \ann W \wedge \V^*$.
  \end{enumerate}
  Finally, the
  intersection $\ker\d_2 \cap \ker\d_3 = \ker\d_1$ is evident from the
  descriptions of the kernels.
\end{proof}

The following result allows to better understand $\ker \d$ as a
$G$-module.

\begin{lemma}\label{lem:ker-d-props}
  Let $\kappa \in \ker\d$.  Then
  \begin{enumerate}[label=(\alph*)]
  \item For all $v,v' \in \V$, $\kappa(v)v' \in W$;
  \item for all $w,w' \in W$, $\kappa(w)w' = 0$.
  \end{enumerate}
\end{lemma}

\begin{proof}
  Both results follow from the well-known fact (see, e.g.,
  \cite[Lemma~2]{Figueroa-OFarrill:2020gpr}) that for any vector space
  $U$, the two subspaces $U \otimes \wedge^2 U$ and $\odot^2 U\otimes
  U$ of $U^{\otimes 3}$ have zero intersection.  Dually, any
  third-rank tensor $S \in (U^*)^{\otimes 3}$ satisfying $S(v_1,v_2,v_3) = S(v_2,v_1,v_3) = -
  S(v_1,v_3,v_2)$ is identically zero.
  \begin{enumerate}[label=(\alph*)]
  \item Let $S \in (\V^*)^{\otimes 3}$ be defined by
    \begin{equation*}
      S(v_1,v_2,v_3) := \eta(\kappa(v_1)v_2,v_3) \qquad\text{for all
        $v_1,v_2,v_3 \in \V$},
    \end{equation*}
    which obeys
    \begin{align*}
      S(v_1,v_2,v_3) &= S(v_2,v_1,v_3) &\tag{because $\d\kappa = 0$}\\
                     &= - S(v_1,v_3,v_2). &\tag{because $\kappa(v_1)\in\g$}
    \end{align*}
    Therefore $S=0$ and hence for all $v_1,v_2 \in \V$,
    $\kappa(v_1)v_2 \in \ker\eta^\flat = W$.
  \item Define $S \in (W^*)^{\otimes 3}$ by
    \begin{equation*}
      S(w_1,w_2,w_3) := \gamma(\kappa(w_1)w_2,w_3) \qquad\text{for all
      $w_1,w_2,w_3\in W$},
    \end{equation*}
    which obeys
    \begin{align*}
      S(w_1,w_2,w_3) &= S(w_2,w_1,w_3) &\tag{because $\d\kappa = 0$}\\
                     &= - S(w_1,w_3,w_2). &\tag{because $\kappa(w_1)\in\g$}
    \end{align*}
    This says that $S=0$ and hence, since $\gamma$ is an inner product
    on $W$, this shows that $\kappa(w_1)w_2=0$ for all $w_1,w_2\in W$.
  \end{enumerate}
\end{proof}

It follows from Lemmas~\ref{lem:ker-d1-2-3} and \ref{lem:ker-d-props}
that $\ker \d \subset \ker \d_2$, but since $\d_2$ is the restriction of
$\d$ to a subspace, we also have that $\ker \d_2 \subset \ker \d$,
hence we conclude that they are the same.

\begin{proposition}\label{prop:kernel-d}
  The kernel of the Spencer differential $\d: \Hom(\V,\g) \to
  \Hom(\wedge^2\V, \V)$ is given by
  \begin{equation}
    \label{eq:kernel-d}
    \ker \d = \left\{\kappa \in \Hom(\V,\g) ~\middle |~ \kappa(w)w'=0\quad\text{and}\quad \kappa(v)v'=\kappa(v')v \in W,~\forall~w,w'\in W,~v,v' \in \V\right\}.
  \end{equation}
\end{proposition}

In particular, this shows that $\ker \d_3 = \ker\d_1$ and we arrive at
the following $G$-invariant filtration of $\ker \d$:
\begin{equation}
  \label{eq:ker-d-filtration}
    0 \subset \ker \d_1 \subset \ker \d.
\end{equation}
The map $\d_1 : W \otimes \ann W \otimes \ann W \to W \otimes \wedge^2
\ann W$ is simply skew-symmetrising the last two tensorands; that is,
$\d_1 = \id_W \otimes \wedge$, whose kernel is the symmetric part $W
\otimes \odot^2 \ann W$.  Since $\eta \in \odot^2\ann W$ is $G$-invariant,
$\ker \d_1$ has a submodule\footnote{If $p=0$, so that $\dim \ann W =
  1$, then this is all of the kernel of $\d_1$. That was the case
  treated in \cite{Figueroa-OFarrill:2020gpr}, which results in the
  filtration $0 \subset W \otimes \RR\eta \subset \ker\d$, which as we
  will see below induces a similar filtration of $\coker\d$ resulting in
  the well-known classification of Galilean structures into
  torsionless, twistless torsional and torsional.} $W \otimes \RR
\eta$, where $\RR\eta$ is the one-dimensional subpace of $\ann W
\otimes \ann W$ spanned by $\eta$; that is,
\begin{equation}
  \label{eq:trace-submodule-ker}
  W \otimes \RR\eta = \left\{\kappa \in \Hom(\V,\g) ~ \middle | ~
    \kappa(v)v' = w\, \eta(v,v')\quad\text{for some $w\in W$ and all
      $v,v' \in \V$}\right\}.
\end{equation}
Alternatively, we see that $W \otimes \RR\eta$ is the image of the
$G$-equivariant linear map $W \to \Hom(\V,\g)$ given by $w \mapsto w
\otimes \eta^\flat \in W \otimes \ann W \otimes \V^* \subset \g
\otimes \V^*$.

The $G$-submodule $W \otimes \RR \eta$ admits a $G$-invariant
complement, described as follows.  Recall that we are using the
shorthand $V = \V/W$.  There is a canonical isomorphism
$\ann W \cong V^*$ under which any $\tau \in \odot^2 \ann W$ is
sent to $\overline\tau \in \odot^2V^*$. This applies in
particular to $\eta$, resulting in the inner product $\overline\eta$
on $V$ which was introduced at the start of this section. Being an
inner product, it defines a musical isomorphism $\overline\eta^\sharp:
V^* \to V$ and this allows us to define a linear map
\begin{equation}
  \label{eq:eta-trace}
  \tr_\eta: \odot^2\ann W \to \RR
\end{equation}
as the following composition of linear maps:
\begin{equation}
  \begin{tikzcd}
    \odot^2\ann W \arrow[r,hook] & \ann W \otimes \ann W
    \arrow[r,"\cong"] & V^* \otimes V^*
    \ar[r,"\overline\eta^\sharp\otimes \id"] & V \otimes V^* \arrow[r,"\cong"]
    & \End V \arrow[r, "\tr"] & \RR,
  \end{tikzcd}
\end{equation}
where the last map is the usual trace of endomorphisms.  Notice that
we may apply this to $\eta$ itself, resulting in the trace of the
identity endomorphism, so that $\tr_\eta \eta = \dim V = \dim \ann W =
p+1$.  This allows us to decompose every $\tau \in \odot^2\ann W$
uniquely into
\begin{equation}
  \tau = \tau_0 + \frac{\tr_\eta \tau}{p+1} \eta,
\end{equation}
where $\tau_0 = \tau - \frac{\tr_\eta \tau}{p+1}\eta \in \odot_0^2\ann
W$, with $\odot_0^2\ann W$ the kernel of the trace map $\tr_\eta :
\odot^2\ann W \to \RR$ defined in Equation~\eqref{eq:eta-trace}.  We
may summarise this discussion as follows.

\begin{lemma}
  The $G$-submodule $\odot^2\ann W \subset \odot^2 \V^*$ is reducible:
  \begin{equation}
    \odot^2\ann W = \odot^2_0 \ann W \oplus \RR \eta.
  \end{equation}
\end{lemma}

In summary, letting $\cK = \ker\d_1 = \cK_0 \oplus \cK_{\mathrm{tr}}$
denote the above decomposition with $\cK_0 = W \otimes \odot_0^2\ann
W$ and $\cK_{\mathrm{tr}} = W \otimes \RR \eta$, we have the following
$G$-submodules of $\ker \d$:
\begin{equation}\label{eq:ker-d-as-g-module}
  \begin{tikzcd}
    & \cK_{\mathrm{tr}} \arrow[rd] & & \\
    0 \arrow[ru] \arrow[rd] & & \cK \arrow[r] & \ker \d.\\
    & \cK_0 \arrow[ru] & & 
  \end{tikzcd}
\end{equation}

Lemma~\ref{lem:ker-d1-2-3} defines $\cK = \ker\d_1$ and $\ker\d =
\ker\d_2$.  The two $G$-submodules of $\cK$ can be characterised as
follows:
\begin{equation}
  \label{eq:g-submods-ker}
  \begin{split}
    \cK_{\mathrm{tr}} &= \left\{\kappa \in \Hom(\V,\g)~\middle |~\kappa(v)v' = \eta(v,v') \tilde{w} ~\forall~v,v' \in \V\quad\text{and some}~ \tilde{w} \in W\right\}\\
    \cK_0 &= \left\{\kappa \in \Hom(\V,\g)~\middle |~\kappa(w)=0~\forall~w \in W\quad\text{and}\quad \tr_\eta \kappa = 0\right\},
  \end{split}
\end{equation}
where in the second line $\tr_\eta : \Hom(\odot^2V,W) \to W$ is the
$\eta$-trace of the linear map $\odot^2V \to W$ induced by the
linear map $\odot^2\V \to W$ given by sending $(v,v') \mapsto
\kappa(v)v'$, which is symmetric for $\kappa \in \ker\d$.

Up to this point we have been agnostic as to the signatures of
$\gamma$ and $\eta$, only assuming that they are nondegenerate.
However, their signatures do play a rôle when discussing the
(ir)reducibility of $\cK_{\mathrm{tr}}$ and $\cK_0$, to which we now
turn.

It is clear that if $\dim W \geq 3$, then $\cK_{\mathrm{tr}}$, which is
isomorphic to $W$ as a $G$-module, is irreducible, since $G$ acts on
$W$ via $\SO(W,\gamma)$ and the vector representation of the special
orthogonal group is irreducible in dimension $\geq 3$.  Of course, if
$\dim W = 1$, then $W$ is, by definition, irreducible.  It remains to
discuss $\dim W =2$.  If $\gamma$ is positive-definite, then $W$
(being real) is also irreducible, but if $\gamma$ is indefinite then
it breaks up into a direct sum of two one-dimensional irreducible
submodules. Indeed, let $e_\pm$ be a Witt (a.k.a. lightcone) basis for
$W$ relative to which $\gamma(e_+,e_-)=1$ and $\gamma(e_\pm,e_\pm)=0$.
The Lorentz group consists only of boosts.  Infinitesimally, the
boost generator $B$ is such that $B e_\pm = \pm e_\pm$ and hence $W =
W_+ \oplus W_-$, where $W_\pm = \RR e_\pm$.  In summary,
$\cK_{\mathrm{tr}}$ is irreducible unless $\dim W = 2$ and $\gamma$ is
indefinite.

Similarly, $\cK_0 \cong W \otimes \odot^2_0 \ann W$ is irreducible
unless $\dim W = 2$ and $\gamma$ is indefinite (as for
$\cK_{\mathrm{tr}}$) or $\dim\ann W = 2$ and $\eta$ is indefinite,
since in that case $\dim \odot_0^2\ann W = 2$ and again it breaks up
into a direct sum of two one-dimensional submodules. Indeed, in terms
of a Witt basis $\theta^\pm$ for $\ann W$, a basis for $\odot_0^2 \ann
W$ is given by their symmetric squares $(\theta^\pm)^2$ and under the
action of the infinitesimal boost generator $B$, $B(\theta^\pm)^2 = \mp 2
(\theta^\pm)^2$.

We have not made a choice of signature for $\gamma$ and $\eta$, but
there are two perspectives we can take.  On the one hand, we can think
in terms of $p$-brane Galilean structures, where $\eta$ has signature
$(1,p)$ and $\gamma$ is positive-definite of size $D-p-1$.  This means
that $\cK_{\mathrm{tr}}$ is always irreducible, but $\cK_0$ is only
irreducible if $p\neq 1$.  The reducible case corresponds to stringy
Galilean structures.

On the other hand, if we were to think in terms of ($D-p-2$)-brane
Carrollian structures, then $\eta$ would be positive-definite of size
$p+1$ and $\gamma$ would be indefinite of signature $(1,D-p-2)$ and
hence both $\cK_{\mathrm{tr}}$ and $\cK_0$ would fail to be
irreducible when $p = D-3$, which is the case of stringy Carrollian
structures.

Hence our results in this section are incomplete for the stringy
Galilean and Carrollian structures.  In terms of our original
data, we will therefore not claim completeness for $p=1$ ($\eta$
indefinite), and $p=D-3$ ($\gamma$ indefinite). We reiterate that most
of our treatment goes through in these cases as well, but there will
be more intrinsic torsion classes than in the generic case due to some
of the representations occurring no longer being irreducible.

\subsection{The associated graded Spencer differential}
\label{sec:assoc-grad-spenc}

The Spencer differential $\d : \Hom(\V,\g) \to \Hom(\wedge^2\V,\V)$ is
a $G$-equivariant map and it preserves the filtrations on both the
domain and codomain.  In Section~\ref{sec:some-filtered-g} we
discussed the associated graded modules $\gr\Hom(\V,\g)$ and
$\gr\Hom(\wedge^2\V,\V)$, which are given respectively by
Equations~\eqref{eq:explicit-grading-hom-v-g} and
\eqref{eq:hom-v-g-graded-components} for $\gr \Hom(\V,\g)$ and by
Equations~\eqref{eq:explicit-grading-torsions} and
\eqref{eq:torsions-graded-components} for $\gr\Hom(\wedge^2\V,\V)$.
These associated graded $G$-modules are isomorphic (as vector spaces,
but not necessarily as $G$-modules) to the filtered $G$-modules.  Let
$j_1 : \Hom(\V,\g) \to \gr \Hom(\V,\g)$ and
$j_2 : \Hom(\wedge^2\V,\V) \to \gr\Hom(\wedge^2\V,V)$ be the
corresponding vector space isomorphisms.  Then the Spencer
differential induces a map
$\Delta:= j_2 \circ \d \circ j_1^{-1}: \gr \Hom(\V,\g) \to
\gr \Hom(\wedge^2\V,\V)$ by demanding the commutativity of the
following square
\begin{equation}
  \begin{tikzcd}
    \Hom(\g,\V) \arrow[d,"j_1"] \arrow[r,"\d"] & \Hom(\wedge^2\V,\V)\arrow[d,"j_2"] \\
    \gr \Hom(\g,\V)  \arrow[r, "\Delta"] & \gr \Hom(\wedge^2\V,\V).
  \end{tikzcd}
\end{equation}
This linear map $\Delta$ between two graded vector spaces may have
components of different (non-positive, since $\d$ preserves the
filtration) degrees: $\Delta = \d_0 + \d_{-1} + \d_{-2} + \cdots$.
The degree-$0$ component $\d_0$ is the associated graded map $\gr\d$,
but there could in principle be components of negative degrees.  To
see that this does not actually happen, we study the map $\Delta$.
From the explicit form of the Spencer differential in Equation
\eqref{eq:spencer-differential}, we see that
\begin{equation}
  \begin{tikzcd}
      \so(V) \otimes W^* \oplus \so(W) \otimes W^*  \arrow[r,"\Delta"] &
      V \otimes V^* \wedge W^* \oplus W \otimes \wedge^2 W^*\\
      \so(V) \otimes V^* \oplus \so(W)\otimes V^* \oplus W \otimes V^*
      \otimes W^* \arrow[r,"\Delta"] & V \otimes \wedge^2 V^* \oplus W
      \otimes V^* \wedge W^*\\
      W \otimes V^* \otimes V^* \arrow[r, "\Delta"] & W \otimes \wedge^2 V^*\\
  \end{tikzcd}
\end{equation}
so that $\Delta$ maps $\Hom(\V,\g)_i \to \Hom(\wedge^2\V,\V )_i$ for all $i=
-1,0,1$. Hence $\Delta = \d_0 = \gr \d$ and hence by
Lemma~\ref{lem:gr-ker=ker-gr}, $\ker\gr\d = \gr\ker\d$.  It follows
therefore that $\dim \ker\d = \dim \gr\ker\d = \dim \ker\gr\d$ and
from the Rank Theorem also that $\rank \d = \rank \gr\d$.

We now analyse $\gr\d$ in more detail to determine its kernel and
count its dimension.

\begin{proposition}
  The linear map $\gr \d : \gr \Hom(\V,\g) \to \gr
  \Hom(\wedge^2\V,\V)$ induced by the Spencer differential, decomposes
  into the following maps:
  \begin{enumerate}
  \item a surjective map $\Hom(V,\Hom(V,W)) \onto \Hom(\wedge^2V,W)$,
    sending $\kappa$ to $\d\kappa(v\wedge v') = \kappa(v)v' -
    \kappa(v')v$, whose kernel is given by
    \begin{equation}
      \label{eq:kernel-gr-d-1}
      K_{-1} = \left\{\kappa \in \Hom(V,\Hom(V,W))~\middle |~ \kappa(v)v' =
        \kappa(v')v \in W,\quad\forall~v,v' \in V\right\};
    \end{equation}
  \item a surjective map
    \begin{equation}
      \Hom(V,\so(V)) \oplus \Hom(V,\so(W)) \oplus \Hom(W,\Hom(V,W))
      \onto \Hom(\wedge^2V,V) \oplus \Hom(V \wedge W, W),
    \end{equation}
    which itself breaks up into two maps:
    \begin{enumerate}
    \item an isomorphism $\Hom(V,\so(V)) \isom \Hom(\wedge^2V,V)$, and
    \item a surjective map $\Hom(V,\so(W)) \oplus \Hom(W,\Hom(V,W)) \onto
      \Hom(V \wedge W,W)$, whose kernel given by
      \begin{align}\label{eq:kernel-gr-d-2}
        K_0 = \left\{\kappa + \lambda \in \Hom(V,\so(W)) \oplus
          \Hom(W,\Hom(V,W))
         ~\middle |~ \begin{array}{l} \kappa(v)w = \lambda(w)v \in
          W\\\quad\forall~v\in V,~w \in W\end{array}\right\};
      \end{align}
    \end{enumerate}
  \item an injective map $\Hom(W,\so(W))\oplus \Hom(W,\so(V)) \into
    \Hom(\wedge^2W,W) \oplus \Hom(V\wedge W, V)$, which itself breaks up
    into two maps:
    \begin{enumerate}
    \item an isomorphism $\Hom(W,\so(W)) \isom \Hom(\wedge^2 W,W)$, and
    \item an injective map $\Hom(W,\so(V)) \into \Hom(V \wedge W, V)$,
      which sends $\kappa$ to $\d\kappa(v \wedge w) = - \kappa(w)v$.
    \end{enumerate}
  \end{enumerate}
\end{proposition}

\begin{proof}
  The fact that the domains and codomains of the maps are as described
  follows by inspection as was stated above.

  That the maps in (2)(a) and (3)(a) are isomorphisms can be proved as
  follows.  Let $\kappa \in \Hom(V,\so(V))$ be in the kernel of
  $\gr\d$.  Then $\kappa(v_1)v_2= \kappa(v_2)v_1$ for all $v_1,v_2\in V$.
  This shows that $\overline\eta(\kappa(v_1)v_2, v_3)$ is symmetric in
  $v_1 \leftrightarrow v_2$ and, since $\kappa(v_1) \in \so(V)$, it is
  also skew-symmetric in $v_2 \leftrightarrow v_3$.  Hence as in the
  proof of Lemma~\ref{lem:ker-d-props}, we conclude that
  $\overline\eta(\kappa(v_1)v_2, v_3)=0$ for all $v_1,v_2,v_3\in V$,
  and since $\overline\eta$ is an inner product, this implies that
  $\kappa = 0$.  Therefore the map in (2)(a) is injective, but since
  it maps between two spaces of equal dimension, the Rank Theorem says
  that it is an isomorphism.  A similar argument shows that the map in
  (3)(a) is an isomorphism.

  It is clear from the explicit expression that the map in (3)(b) is
  injective, since if $\kappa(w)v = 0$ for all $v\in V$ and $w \in W$,
  then $\kappa(w) = 0$ for all $w\in W$ and thus $\kappa = 0$.  The
  kernels $K_{-1}$ in (1) and and $K_0$ in (2)(b) follow from the
  explicit expression of the Spencer differential.  The fact that the
  map in (1) is surjective can be seen as follows: the map in question
  is the composition $\Hom(V,\Hom(V,W) \cong \Hom(V\otimes V, W) \to
  \Hom(\wedge^2 V, W)$, where the first map is the inverse of the
  currying isomorphism and the second is skew-symmetrisation, which is
  certainly surjective.  Finally, the map in (2)(b) is surjective
  already when restricted to $\Hom(W,\Hom(V,W))$ since it is up to a
  sign the inverse of the currying isomorphism $\Hom(W,\Hom(V,W))
  \cong \Hom(V\otimes W, W)$ composed with the isomorphism $V \wedge W
  \cong V \otimes W$.
\end{proof}

It follows from the previous discussion and the fact that the rank and nullity of
$\d$ and $\gr\d$ agree, that
\begin{equation}
  \label{eq:rank-spencer}
  \begin{split}
    \rank\d = \rank \gr\d &= \dim \Hom(\wedge^2V,W) + \dim\Hom(\wedge^2V,V) + \dim \Hom(V \wedge W,W)\\
     & \qquad {} + \dim\Hom(\wedge^2W,W) + \dim\Hom(W,\so(V))\\
    &= \tfrac12 D^2((D-p-1)^2+Dp)
  \end{split}
\end{equation}
and from the Rank Theorem
\begin{equation}
  \label{eq:dim-ker-spencer}
  \dim\ker\d = \dim\ker\gr\d = \dim \Hom(\V,\g) - \rank \d 
  =\tfrac12 D(D-p-1)(p+1).
\end{equation}
Notice that $\gr\ker\d = K_{-1} \oplus K_0$, so that $\dim\ker\d =
\dim K_{-1} + \dim K_0$.  Together with the fact that $K_{-1} \cong
\Hom(\odot^2V,W)$ has dimension $\dim K_{-1} 
=\tfrac12 (D-p-1)(p+1)(p+2)$, we
see that $\dim K_0  
=\tfrac12 (D-p-1)(p+1)(D-p-2)$.

\subsection{Intrinsic torsion classes}
\label{sec:intr-tors-class}

Notice that $\dim \g = \dim \wedge^2\V$, so that $\dim \Hom(\V,\g) =
\dim \Hom(\wedge^2\V,\V)$ and applying the Euler--Poincaré principle
to the exact sequence~\eqref{eq:intrinsic-torsion-exact-sequence}, we
conclude that $\dim \ker\d = \dim \coker \d$.  More is true, however,
and they are isomorphic as $G$-modules, as we now show.

Define $\varphi: \Hom(\wedge^2\V,\V) \to \Hom(\V,\End\V)$ by $T
\mapsto \varphi_T$ where
\begin{equation}
  \label{eq:phi-map}
  \left<\alpha, \varphi_T(v)v' \right> := \eta(T(v \wedge
  h^\sharp\alpha), v') + \eta(T(v' \wedge h^\sharp\alpha), v)
\end{equation}
for all $\alpha \in \V^*$ and $v,v' \in \V$.

\begin{lemma}
  \label{lem:phi-v-in-hom-v-w}
  $\varphi_T(v)v' = \varphi_T(v')v \in W$ for all $v,v'\in\V$.
\end{lemma}

\begin{proof}
  The symmetry in $v \leftrightarrow v'$ is clear from the
  definition.  To show that the result lies in $W$, let $v''\in \V$ be
  arbitrary and calculate
  \begin{align*}
    \eta(v'', \varphi_T(v)v') &= \left<\eta^\flat v'', \varphi_T(v)v'\right> &\tag{by definition of $\eta$}\\
                              &= \eta(T(v \wedge h^\sharp\eta^\flat v''), v') + \eta(T(v' \wedge h^\sharp\eta^\flat v''), v)\\
                              &= 0, &\tag{since $h^\sharp \circ \eta^\flat = 0$}
  \end{align*}
  so that $\varphi_T(v)v' \in W$.
\end{proof}

\begin{lemma}\label{lem:im-phi-in-ker-d}
  The image of $\varphi$ lies in $\ker\d \subset \Hom(\V,\g)$.
\end{lemma}

\begin{proof}
  It is clear from the definition that $\varphi_T(v)v' =
  \varphi_T(v')v$, so that provided its image lies in $\Hom(\V,\g)$ then it
  lies in $\ker\d$.  From Lemma~\ref{lem:phi-v-in-hom-v-w}, we see
  that $\varphi_T : \V \to \Hom(\V,W)$ and hence $\varphi_T(v)$
  trivially preserves $\eta$: not just is $\eta(\varphi_T(v)v',v') =
  0$, but in fact $\eta(\varphi_T(v)v',v'') =0$.
  
  To see that $\varphi_T(v)$ also preserves $h$, let
  $\alpha_1,\alpha_2\in \V^*$, and calculate
  \begin{align*}
    h(\varphi_T(v)^t\alpha_1, \alpha_2)&= \left<\varphi_T(v)^t\alpha_1, h^\sharp(\alpha_2)\right>\\
                                              &= \left<\alpha_1, \varphi_T(v) h^\sharp(\alpha_2)\right>\\
                                              &= \eta(T(v\wedge h^\sharp\alpha_1), h^\sharp\alpha_2) + \eta(T(h^\sharp \alpha_2 \wedge h^\sharp \alpha_1), v)\\
    &= -\eta(T(h^\sharp \alpha_1 \wedge h^\sharp \alpha_2), v), &\tag{since $h^\sharp \alpha_2 \in W$}
  \end{align*}
  which is clearly skew-symmetric in $\alpha_1$ and $\alpha_2$, so that
  \begin{equation*}
    h(\varphi_T(v)^t\alpha_1, \alpha_2) + h(\varphi_T(v)^t\alpha_2, \alpha_1) = 0.
  \end{equation*}
\end{proof}

As a result of the lemma, we will henceforth think of $\varphi$ as a
linear map $\varphi: \Hom(\wedge^2\V,\V) \to \Hom(\V,\g)$, which is
clearly $G$-equivariant, since it is constructed out of the
$G$-invariant tensors $\eta$ and $h$.

Just like the Spencer differential, the map $\varphi$ too induces a
linear map $\Phi: \gr \Hom(\wedge^2\V,\V) \to \gr\Hom(\V,\g)$ between
graded vector spaces, which can be seen to be of degree $-2$: the dual
pairing and $\eta$ have degree $0$, whereas $h^\sharp$ has degree
$-2$.  Therefore $\Phi : \Hom(\wedge^2\V,\V)_i \to \Hom(\V,\g)_{i-2}$,
where $i=-1,0,1,2$.

\begin{lemma}\label{lem:im-phi-onto-ker-d}
  The linear map $\Phi: \gr \Hom(\wedge^2\V,\V) \to
  \gr\Hom(\V,\g)$ induced by $\varphi$ has two nonzero components:
  \begin{enumerate}
  \item a map $\Hom(V\wedge W,V) \to \Hom(V,\Hom(V,W))$ whose image is
    $K_{-1}$ in Equation~\eqref{eq:kernel-gr-d-1}, and
  \item a map $\Hom(\wedge^2W, V) \to \Hom(V,\so(W)) \oplus
    \Hom(W,\Hom(V,W))$ whose image is $K_0$ in
    Equation~\eqref{eq:kernel-gr-d-2}.
  \end{enumerate}
\end{lemma}

\begin{proof}
  As stated above, $\Phi$ has degree $-2$ and hence it has only
  two nonzero components: $\Hom(\wedge^2\V,\V)_1 \to \Hom(\V,\g)_{-1}$
  and $\Hom(\wedge^2\V,\V)_2 \to \Hom(\V,\g)_0$.  In the former case,
  we have
  \begin{equation*}
    \gr \varphi : \Hom(\wedge^2W,W) \oplus \Hom(V\wedge W, V) \to \Hom(V,\Hom(V,W))
  \end{equation*}
  but the map is identically zero on $\Hom(\wedge^2W,W)$ since any $T
  \in \gr\Hom(\wedge^2\V,\V)$ with values in $W$ is clearly in the
  kernel of $\varphi$.  In the latter case, we have
  \begin{equation*}
    \Phi: \Hom(\wedge^2W,V) \to \Hom(W,\Hom(V,W)) \oplus
    \Hom(V,\so(V)\oplus \so(W)),
  \end{equation*}
  but it follows from Lemma~\ref{lem:phi-v-in-hom-v-w} that
  $\varphi_T(v)v' \in W$, so that for no $T \in \Hom(\wedge^2W,V)$
  does $\varphi_T$ have a component along $\Hom(V,\so(V))$.
  In summary, there are only two nonzero components of $\Phi$:
  \begin{equation*}
    \Hom(V \wedge W, V) \to \Hom(V,\Hom(V,W))
  \end{equation*}
  and
  \begin{equation*}
    \Hom(\wedge^2W,V) \to \Hom(V,\so(W)) \oplus \Hom(W,\Hom(V,W)),
  \end{equation*}
  which we now analyse in turn.
  \begin{enumerate}
  \item Firstly, we observe that Lemma~\ref{lem:phi-v-in-hom-v-w} says
    the image of this map lies in $K_{-1}$.  Secondly, we observe that there is an isomorphism
    \begin{equation}\label{eq:iso-hom-v-w-v}
      \begin{tikzcd}
        \Hom(V \wedge W, V)  \arrow[r,"\cong"] & \Hom(V,\Hom(V,W)),
      \end{tikzcd}
    \end{equation}
    sending $T \mapsto\kappa$, where for all $\alpha \in \V^*$,
    $\left<\alpha, \kappa(v)v'\right> = 2 \overline\eta(T(v\wedge h^\sharp \alpha), v')$ for all
    $v,v' \in V$, where the factor of $2$ is for convenience.  Now suppose that
    $\kappa \in K_{-1} \subset \Hom(V,\Hom(V,W))$ and let $T$ be the
    corresponding vector in $\Hom(V\wedge W, V)$ which maps to
    $\kappa$ under the above isomorphism.  Then
    \begin{align*}
      \left<\alpha, \Phi_T(v)v'\right> &= \overline\eta(T(v\wedge h^\sharp\alpha), v') +  \overline\eta(T(v'\wedge h^\sharp\alpha), v)\\
                                          &= \tfrac12 \left<\alpha, \kappa(v)v'+\kappa(v')v\right> &\tag{by \eqref{eq:iso-hom-v-w-v}}\\
                                          &= \left<\alpha, \kappa(v)v'\right>, &\tag{since $\kappa \in K_{-1}$}
    \end{align*}
    so that $\kappa = \Phi_T$ and hence $\varphi : \Hom(V\wedge W,V) \to \Hom(V,\Hom(V,W))$ has image $K_{-1}$.
  \item Now we depart from the isomorphism
    \begin{equation*}
      \begin{tikzcd}
        \Hom(\wedge^2 W, V) \arrow[r,"\cong"] & \Hom(V,\so(W))
      \end{tikzcd}
    \end{equation*}
    sending $T \mapsto \kappa$, where $\left<\alpha, \kappa(v)w\right>
    = \overline\eta(T(w \wedge h^\sharp \alpha),v)$ for all $v \in V$,
    $\alpha \in \V^*$ and $w\in W$, and the
    injective map
    \begin{equation*}
      \begin{tikzcd}
        \Hom(\wedge^2 W, V) \arrow[r, hook] & \Hom(W,\Hom(V,W))
      \end{tikzcd}
    \end{equation*}
    sending $T \mapsto \lambda$, where $\left<\alpha, \lambda(w)v\right> =
    \overline\eta(T(w\wedge h^\sharp\alpha),v)$ for all $v\in V$, $\alpha \in \V^*$ and $w \in W$.  Putting
    them together we obtain the second nonzero component of $\Phi$ as an injective map $\Hom(\wedge^2 W,V) \into
    \Hom(V,\so(W)) \oplus \Hom(W,\Hom(V,W))$ sending $T \mapsto \kappa
    + \lambda$, where
    \begin{equation*}
      \left<\alpha, \lambda(w)v\right> = \overline\eta(T(w\wedge h^\sharp\alpha),v) = \left<\alpha, \kappa(v)w \right>.
    \end{equation*}
    It follows that $\kappa + \lambda \in K_0$.  Since the map is
    injective, its rank is 
    $\tfrac12 (p+1)(D-p-1)(D-p-2)= \dim K_0$ and hence
    its image is all of $K_0$.
  \end{enumerate}
\end{proof}

\begin{proposition}\label{prop:exact-pair-d-phi}
  We have an exact pair of $G$-modules
  \begin{equation}\label{eq:exact-pair-d-phi}
  \begin{tikzcd}
    \Hom(\V,\g) \arrow[r, shift left, "\d"] & \Hom(\wedge^2\V,\V)
    \arrow[l,shift left, "\varphi"];
  \end{tikzcd}
\end{equation}
that is, $\im\d = \ker\varphi$ and $\im \varphi = \ker \d$.
\end{proposition}

\begin{proof}
  From Lemma~\ref{lem:im-phi-in-ker-d} it follows that $\im \varphi
  \subseteq \ker \d$, but then Lemma~\ref{lem:im-phi-onto-ker-d} shows
  that
  \begin{equation*}
    \rank \varphi = \dim K_{-1} + \dim K_0 = \dim \ker\d,
  \end{equation*}
  and hence $\im \varphi = \ker \d$.

  To show that $\ker\varphi = \im\d$ we first show that
  $\im\d \subseteq \ker\varphi$ and then we count dimensions to show
  the equality.  To prove that $\im\d \subset \ker\varphi$, let $\kappa \in \Hom(\V,\g)$, then for all $\alpha \in \V^*$ and $v \in \V$,
  \begin{align*}
    \eta(\d\kappa(h^\sharp \alpha \wedge v), v) &= \eta(\kappa(h^\sharp\alpha)v - \kappa(v)h^\sharp\alpha, v)\\
                                              &= \eta(\kappa(h^\sharp\alpha)v,v) - \eta(\kappa(v)h^\sharp\alpha, v),
  \end{align*}
  but the first term vanishes because the image of $\kappa$ lies in $\g$
  and the second term vanishes because $h^\sharp\alpha$ and hence
  $\kappa(v)h^\sharp \alpha $ lie in $W$.  By polarisation, we
  see that $\d\kappa \in \ker \varphi$.

  Finally, by the Rank Theorem, we have that
  \begin{equation*}
    \dim \ker\varphi + \rank \varphi = \dim \Hom (\wedge^2\V,\V) =
    \dim \Hom(\V,\g) = \dim\ker \d + \rank \d.
  \end{equation*}
  Since we have already shown that $\rank \varphi = \dim \ker\d$,
  it follows that $\dim \ker \varphi = \rank \d$ and hence $\im\d =
  \ker\varphi$.
\end{proof}

It follows that from the First Isomorphism Theorem that $\varphi$
induces a $G$-module isomorphism
$\overline\varphi: \coker\d \to \ker \d$, whose inverse allows us to
determine the $G$-module structure of $\coker\d$ by transporting
the diagram~\eqref{eq:ker-d-as-g-module} to $\coker\d$.  It will be
convenient to actually consider the preimages in $\Hom(\wedge^2\V,\V)$
of the $G$-modules appearing in the diagram~\eqref{eq:ker-d-as-g-module}
under the surjection $\varphi : \Hom(\wedge^2\V,\V) \to \ker\d$.  Let
us define $\cT := \varphi^{-1}(\cK)$, with
$\cK = \cK_0 \oplus \cK_{\mathrm{tr}}$ the $G$-submodule of $\ker\d$
in the diagram~\eqref{eq:ker-d-as-g-module}.  Since $\varphi$ is
$G$-equivariant, we have that $\cT = \cT_0 \oplus \cT_{\mathrm{tr}}$,
where $\cT_0 = \varphi^{-1}(\cK_0)$ and $\cT_{\mathrm{tr}} =
\varphi^{-1}(\cK_{\mathrm{tr}})$.  Using that $\varphi^{-1}(0) =
\ker\varphi = \im\d$, this results in the following $G$-submodules of
$\varphi^{-1}(\ker\d) = \Hom(\wedge^2\V,\V)$:
\begin{equation}\label{eq:torsions-as-g-module}
  \begin{tikzcd}
    & \cT_{\mathrm{tr}} \arrow[rd] & & \\
    \im\d \arrow[ru] \arrow[rd] & & \cT \arrow[r] & \Hom(\wedge^2\V,\V),\\
    & \cT_0 \arrow[ru] & & 
  \end{tikzcd}
\end{equation}
which may be quotiented by $\im\d = \ker\varphi$ to give the desired
description of $\coker\d$:
\begin{equation}\label{eq:coker-d-as-g-module}
  \begin{tikzcd}
    & \overline\cT_{\mathrm{tr}} \arrow[rd] & & \\
    0 \arrow[ru] \arrow[rd] & & \overline\cT \arrow[r] & \coker \d,\\
    & \overline\cT_0 \arrow[ru] & & 
  \end{tikzcd}
\end{equation}
where $\overline\cT = \overline\varphi^{-1}(\cK) =
\varphi^{-1}(\cK)/\im\d$, et cetera.  This results in five
$G$-submodules of $\coker \d$ and hence five intrinsic torsion classes
of $G$-structures.  In summary, we have the following:

\begin{theorem}
  \label{thm:intrinsic-torsion-classes}
  There are five $G$-submodules of $\coker\d$ and hence five
  intrinsic torsion classes of $G$-structures: $0$ (vanishing
  intrinsic torsion), $\overline\cT_{\mathrm{tr}}$, $\overline\cT_0$,
  $\overline\cT = \overline \cT_0 \oplus \overline \cT_{\mathrm{tr}}$
  and $\coker\d$ (generic intrinsic torsion).
\end{theorem}

We recapitulate that this theorem gives a complete list of the
intrinsic torsion classes (i.e., neither more nor fewer) except in the
following cases: $p=0$ and $p=1$ with $\eta$ indefinite and $p=D-2$
and $p=D-3$ with $\gamma$ indefinite. The cases $p=0$ ($\eta$
indefinite) and $p=D-2$ ($\gamma$ indefinite) can be found in
\cite{Figueroa-OFarrill:2020gpr} and in any case will be shown later
to follow from this. The cases $p=1$ ($\eta$ indefinite) and $p=D-3$
($\gamma$ indefinite), corresponding to stringy Galilean and stringy
Carrollian structures, respectively, will be treated elsewhere.

\subsection{Geometric interpretation}
\label{sec:geom-interpr}

A natural question is now to characterise the intrinsic torsion
classes in Theorem~\ref{thm:intrinsic-torsion-classes} geometrically.
The intrinsic torsion is a section of the associated vector bundle
$P \times_G \coker\d$ to the $G$-structure $P\to M$. This bundle is a
quotient bundle of
$P \times_G \Hom(\wedge^2 \V,\V) \cong TM \otimes \wedge^2 T^*M$,
which is the bundle of which the actual torsion $T^\nabla$ of an
affine connection is a section. Equivalently, we may view the torsion
as a $G$-equivariant function $T^\nabla : P \to \Hom(\wedge^2 \V, \V)$,
whereas the intrinsic torsion can be described as a $G$-equivariant
function $[T^\nabla]: P \to \coker \d$. By the use of local frames in
$P$, we may view the torsion and the intrinsic torsion as locally
defined functions on $M$ with values in $\Hom(\wedge^2 \V, \V)$ and
$\coker\d$, respectively.

Now suppose that $[T^\nabla]$ takes values in a $G$-submodule
$\cT \subset \coker\d$.  What does this say about the actual torsion
$T^\nabla$?

Locally on $M$ and relative to a frame in $P$, it is given by a
function with values in the $G$-submodule
$\widetilde \cT \subset \Hom(\wedge^2 \V, \V)$ which projects down to
$\cT \subset \coker\d$.  In other words, the submodule
$\widetilde \cT$ is the preimage in $\Hom(\wedge^2 \V, \V)$ under the
projection $\pi: \Hom(\wedge^2 \V, \V) \to \coker \d$ which is part of
the short exact sequence
\begin{equation}
  \begin{tikzcd}
    0 \arrow[r] & \im \d \arrow[r] & \Hom(\wedge^2\V,\V)\arrow[r,"\pi"] &  \coker \d \arrow[r] & 0,
  \end{tikzcd}
\end{equation}
which typically will not split. In other words, we cannot view $\cT$
as a $G$-submodule of $\Hom(\wedge^2 \V, \V)$. In practice we may
split the sequence \emph{as vector spaces} by choosing a vector
subspace $\cT'$ of $\Hom(\wedge^2 \V, \V)$ that projects
isomorphically to $\cT$. But since $\cT'$ is not a $G$-submodule, this
has the consequence that whereas we are able to modify the connection
so that, relative to a given local frame in $P$, $T^\nabla$ takes
values in $\cT'$, this will not necessarily be the case relative to
other local frames. If we change frames, we have to modify the
connection again so that its torsion takes values in $\cT'$. This is
why it is important to derive geometric interpretations for the
different intrinsic torsion classes which make no mention of the
actual connection and why we work with the full preimage
$\widetilde\cT$ of $\cT$. This is what we do in this section. First,
we will characterise the $G$-submodules in question algebraically and
then geometrically.

\subsubsection{Algebraic characterisation}
\label{sec:algebr-char}

A necessary first step is then to characterise the $G$-submodules of
$\Hom(\wedge^2 \V,\V)$ which were determined
in~\eqref{eq:torsions-as-g-module}: $\im\d$, $\cT$,
$\cT_{\mathrm{tr}}$ and $\cT_0$. These are vector subspaces of
$\Hom(\wedge^2\V, \V)$ and hence they are determined by linear
equations. The case of generic intrinsic torsion corresponds to
$\Hom(\wedge^2 \V,\V)$ itself, hence to no conditions at all on the
torsion. Any other proper $G$-submodule of $\coker\d$ will lift to a
proper $G$-submodule of $\Hom(\wedge^2 \V,\V)$ determined by some
non-trivial linear equations. We now determine those linear equations
for all but the generic intrinsic torsion classes.

From Proposition~\ref{prop:exact-pair-d-phi}, it follows that $\im\d =
\ker\varphi$, so that
\begin{equation}
  \label{eq:im-d}
  \im\d = \left\{T \in\Hom(\wedge^2\V,\V)~\middle |~\eta(T(v\wedge w),v') + \eta(T(v'\wedge w),v) = 0\right\},
\end{equation}
where the condition holds for all $v,v' \in \V$ and all $w \in W$.

\begin{proposition}
  \label{prop:alg-char-T}
  The $G$-submodule $\cT \subset \Hom(\wedge^2\V,\V)$ is given by
  \begin{equation}
    \label{eq:torsion-mod-T}
    \cT = \left\{T \in\Hom(\wedge^2\V,\V)~\middle |~ \eta(T(w\wedge w'),v) = 0\right\},
  \end{equation}
  where the condition holds for all $w,w' \in W$ and all $v \in \V$.
\end{proposition}

\begin{proof}
  By definition, $T \in \cT$ if and only if $\varphi_T \in \cK$, where
  \begin{equation*}
    \cK = \left\{\kappa \in \Hom(\V,\g)~\middle |~ \kappa(w) =0
      \quad\text{and}\quad \kappa(v)v' = \kappa(v')v \in
      W,\quad\forall~w\in W,~v,v' \in \V \right\}.
  \end{equation*}
  From Lemma~\ref{lem:phi-v-in-hom-v-w}, it follows that
  $\varphi_T(v)v' = \varphi_T(v')v \in W$, so the only condition we
  need to consider is $\varphi_T(w) = 0$.  Let $\alpha\in\V^*$ and
  $v\in\V$ be arbitrary and calculate
  \begin{align*}
    \left<\alpha, \varphi_T(w)v\right> &= \eta(T(w \wedge h^\sharp\alpha), v) +  \eta(T(v \wedge h^\sharp\alpha), w)\\
                                       &= \eta(T(w \wedge h^\sharp\alpha), v) &\tag{since $\eta(-,w)=0$ for $w \in W$.}
  \end{align*}
  Therefore
  \begin{align*}
    \varphi_T(w)=0 &\iff \eta(T(w\wedge h^\sharp\alpha),v) = 0\quad\text{for all $v \in \V$ and $\alpha \in \V^*$}\\
                   &\iff T(w\wedge h^\sharp\alpha) \in W \quad\text{for all $\alpha \in \V^*$}\\
                   &\iff T(w \wedge w') \in W \quad \text{for all $w' \in W$}. &\tag{since $h^\sharp : \V^* \onto W$}
  \end{align*}
\end{proof}

\begin{proposition}
  \label{prop:alg-char-T-tr}
  The $G$-submodule $\cT_{\mathrm{tr}} \subset \Hom(\wedge^2\V,\V)$ is given by
  \begin{equation}
    \label{eq:torsion-mod-T-tr}
    \cT_{\mathrm{tr}} = \left\{T \in\Hom(\wedge^2\V,\V)~\middle |~
      \eta(T(v \wedge w),v') + \eta(T(v'\wedge w),v) = \gamma(\tilde{w},w) \eta(v,v')~\text{for some}~\tilde{w}\in W\right\}
  \end{equation}
  where the condition holds for all $v,v'\in\V$ and all $w\in W$.
\end{proposition}

\begin{proof}
  By definition $T \in \cT_{\mathrm{tr}}$ if and only if $\varphi_T \in
  \cK_{\mathrm{tr}}$.  From Equation~\eqref{eq:g-submods-ker}, this condition becomes $\varphi_T(v)v' = \tilde{w} \eta(v,v')$ for some $\tilde{w}
  \in W$.  Let $\alpha \in \V^*$ be arbitrary, then $T \in
  \cT_{\mathrm{tr}}$ if and only if 
  \begin{equation*}
    \eta(T(v \wedge h^\sharp\alpha),v') + \eta(T(v' \wedge
    h^\sharp\alpha),v) = \eta(v,v') \left<\alpha,\tilde{w}\right>
  \end{equation*}
  for all $\alpha \in \V^*$ and all $v,v' \in \V$.  Using that
  $h^\sharp: \V^* \onto W$, this is equivalent to
  \begin{equation*}
    \eta(T(v \wedge w),v') + \eta(T(v' \wedge w),v) = \eta(v,v') \gamma(w,\tilde{w}),
  \end{equation*}
  for all $w \in W$ and $v,v' \in \V$.  
\end{proof}

To describe $\cT_0$ requires introducing a notion of trace. Every
$T \in \Hom(\wedge^2\V,\V)$ defines a linear map $W \to \odot^2\V^*$
by sending $w \mapsto T_w$, where
$T_w(v,v') = \eta(T(v \wedge w), v') + \eta(T(v'\wedge w),v)$. If
$T \in \cT$, then if either $v$ or $v'$ belongs to $W$, we get zero,
so in fact $T_w \in \odot^2\ann W$ and thus every $T \in \cT$ defines
a linear map $W \to \odot^2\ann W$. The $\eta$-trace $\tr_\eta T$ is
the component of this map along
$\RR\eta \subset \odot^2\ann W = \RR\eta \oplus \odot_0^2\ann W$. We
may describe this alternatively as follows. We observe that $T$
defines a linear map $W \to \End(\V/W)$ sending $w$ to
$\overline T_w$, where, denoting by $v \mapsto \overline v$ the linear
map $\V \to \V/W$, \begin{equation}
  \overline T_w(\overline v) = \overline{T(w \wedge v)},
\end{equation}
which is well-defined because $T(w \wedge w') \in W$, by
Proposition~\ref{prop:alg-char-T}.  Then the condition $\tr_\eta T =
0$ is simply the condition $\tr \overline T_w = 0$ for all $w \in W$.
This trace map plays a rôle in the algebraic characterisation of
$\cT_0$.

\begin{proposition}\label{prop:alg-char-T-0}
  The $G$-submodule $\cT_0 \subset \Hom(\wedge^2\V,\V)$ is given by
  \begin{equation}
    \label{eq:torsion-mods-T-0}
    \cT_0 = \left\{T \in\Hom(\wedge^2\V,\V)~\middle |~ \eta(T(w\wedge w'),v) = 0\quad\text{and}~ \tr_\eta T = 0\right\},
  \end{equation}
  where the condition holds for all $w,w'\in W$ and all $v \in \V$.
\end{proposition}

\begin{proof}
  By definition, $T \in \cT_0$ if and only if $\varphi_T \in \cK_0$.
  From Equation~\eqref{eq:g-submods-ker}, this condition becomes
  $\varphi_T(w)=0$ and $\tr_\eta \varphi_T = 0$.  As shown in the
  proof of Proposition~\ref{prop:alg-char-T}, the former condition is
  simply $T(w \wedge w') \in W$ for all $w,w' \in W$.  The latter
  condition says that the map $(v,v') \mapsto \varphi_T(v)v' \in W$ has zero
  component along $W \otimes \RR\eta$, or equivalently that for all
  $\alpha \in \V^*$, $\left<\alpha,\varphi_T(v)v'\right>$ has no
  component along $\RR\eta$.  But this simply says that for all
  $\alpha \in \V^*$, the map $(v,v') \mapsto T_{h^\sharp\alpha}(v,v')$
  has no component along $\RR\eta$, which since $h^\sharp$ is
  surjective, is the same as the map $(v,v') \mapsto T_w (v,v')$
  having no component along $\RR\eta$ for every $w \in W$.
\end{proof}

\subsubsection{Geometric characterisation}
\label{sec:geom-char}

We are now finally in a position to characterise the different
intrinsic torsion classes geometrically. Our first observation is that
the $G$-submodule $W \subset \V$ defines subbundles $E \subset TM$ and
$\ann E \subset T^*M$ under the isomorphisms $P\times_G \V \cong TM$
and $P\times_G \V^* \cong T^*M$ defined by the soldering form. The
characteristic tensor fields of the $G$-structure are
$\eta \in \Gamma(\odot^2\ann E)$, $h \in \Gamma(\odot^2 E)$ and, since
we restrict to the identity component of $G$, also an $(p+1)$-form
$\Omega \in \Gamma(\wedge^{p+1} \ann E)$ defined as follows. Since
$\ann W$ is an $(p+1)$-dimensional $G$-submodule of $\V^*$, the
one-dimensional vector space $\wedge^{p+1}\ann W$ is a $G$-submodule
of $\wedge^{p+1}\V^*$. Let $\Omega \in \wedge^{p+1}\ann W$ be any
nonzero vector. It is characterised up to scale by the condition that
$\iota_w \Omega = 0$ for all $w \in W$. This means that $\Omega$
induces $\overline\Omega \in \wedge^{p+1}(\V/W)^*$ under the
$G$-module isomorphism $\ann W \cong (\V/W)^*$. Now the action of $G$
on $\V/W$ preserves the inner product
$\overline\eta \in \odot^2(\V/W)^*$, so the representation map sends
$G \to \Ort(\V/W,\overline\eta)$. Since $G$ is assumed connected, it
actually lands in $\SO(\V/W,\overline\eta)$, and hence
$g \cdot \overline\Omega = (\det g)~\overline\Omega =
\overline\Omega$. This says that $\overline\Omega$ (and hence $\Omega$
itself) is $G$-invariant and hence defines a characteristic tensor
field also denoted $\Omega \in \Gamma(\wedge^{p+1}\ann E)$ on any
manifold with a $G$-structure. This is in fact the main technical
reason why we assumed that we could restrict to the identity component
of $G$.

The characteristic tensor fields are parallel relative to any adapted
connection and the different intrinsic torsion classes of adapted
connections can be characterised by what they imply on the subbundle
$E \subset TM$ and the characteristic tensor fields.  We hope it
causes no confusion that we use the same notation for the
characteristic tensor fields as we do for the $G$-invariant tensors:
namely, $\eta$, $h$ and $\Omega$.

\begin{lemma}
  \label{lem:nabla-preserves-E}
  Let $\nabla$ be an adapted connection.  Then $\nabla$ preserves $E$;
  that is, $\nabla_X Y \in \Gamma(E)$ for all $Y \in \Gamma(E)$ and $X
  \in \eX(M)$.
\end{lemma}

\begin{proof}
  Notice that $\eta \in \Gamma(\odot^2 \ann E)$ is nondegenerate, so
  that $\eta(Z,X) = 0$ for all $X \in \eX(M)$ if and only if $Z \in \Gamma(E)$.  Hence
  let $Z \in \Gamma(E)$ and $X,Y \in \eX(M)$.  Since $\nabla \eta =0$,
  we see that for all $X,Y \in \eX(M)$,
  \begin{align*}
    0 &= (\nabla_X\eta)(Y,Z) \\
      &= X \cdot \eta(Y,Z) - \eta(\nabla_X Y, Z) - \eta(Y,\nabla_X Z)\\
      &= - \eta(Y,\nabla_X Z). & \tag{since $Z \in \Gamma(E)$}
  \end{align*}
  Since this is true for all $Y$, it follows that $\nabla_X Z \in \Gamma(E)$.
\end{proof}

\begin{proposition}
  \label{prop:E-involutive}
  $T^\nabla \in \cT$ if and only if $E \subset TM$ is involutive.
\end{proposition}

\begin{proof}
  Suppose that $T^\nabla \in \cT$.  From
  Proposition~\ref{prop:alg-char-T}, this is equivalent to
  $\eta(T^\nabla(X,Y),Z) =0$ for all $X,Y \in \Gamma(E)$ and $Z \in
  \eX(M)$.  In other words, $T^\nabla(X,Y) \in \Gamma(E)$ for all $X,Y
  \in \Gamma(E)$.  But then, by Lemma~\ref{lem:nabla-preserves-E},
  \begin{equation}
    [X,Y] = \nabla_X Y - \nabla_Y X - T^\nabla(X,Y) \in \Gamma(E).
  \end{equation}
\end{proof}

It follows from the Frobenius integrability theorem that if $E$ is
involutive, $M$ is foliated by integral submanifolds of $E$.  Moreover
$h \in \Gamma(\odot^2 E)$ is non-degenerate on each leaf and
hence can be inverted to define a metric $\gamma$ on each leaf.

Since the remaining intrinsic torsion classes are contained in
$\cT$, we will assume from now on that $E$ is involutive.

\begin{proposition}
  \label{prop:trace-intrinsic-torsion}
  $T^\nabla \in \cT_{\mathrm{tr}}$ if and only if for all $X \in \Gamma(E)$,
  \begin{equation}
    \eL_X \eta = 2 \gamma (X,Z) \eta,
  \end{equation}
  for some $Z \in \Gamma(E)$.
\end{proposition}

\begin{proof}
  Let us calculate the Lie derivative of $\eta$ along $X \in \Gamma(E)$:
  \begin{align*}
    (\eL_X\eta)(Y,W) &= X \cdot \eta(Y,W) - \eta([X,Y],W)- \eta(Y,[X,W])\\
                     &= \eta(\nabla_X Y- [X,Y], W) + \eta(Y, \nabla_X W - [X,W])  & \tag{using that $\nabla \eta = 0$}\\
                     &= \eta(T^\nabla(X,Y) + \nabla_Y X ,W) + \eta(Y, T^\nabla(X,W) + \nabla_W X) & \tag{by definition of $T^\nabla$}\\
                     &= \eta(T^\nabla(X,Y) ,W) + \eta(Y, T^\nabla(X,W)). & \tag{since $\nabla_Y X,\nabla_W X \in \Gamma(E)$}
  \end{align*}
  From Proposition~\ref{prop:alg-char-T-tr}, it follows that $T^\nabla \in
  \cT_{\mathrm{tr}}$ if and only if
  \begin{equation}
    \label{eq:trace-aux}
    \eta(T^\nabla(X,Y),W) + \eta(T^\nabla(X,W),Y) = 2 \gamma(X,Z) \eta(Y,W)
  \end{equation}
  for all $X \in \Gamma(E)$, $Y,W \in \eX(M)$ and for some $Z \in
  \Gamma(E)$.  In other words, $T^\nabla \in \cT_{\mathrm{tr}}$ if
  and only if
  \begin{equation}
    (\eL_X\eta)(Y,W) = 2 \gamma(X,Z) \eta(Y,W)
  \end{equation}
  for some $Z \in \Gamma(E)$.  Since this holds for all $Y,W
  \in \eX(M)$, we may abstract them and conclude that $\eL_X \eta = 2
  \gamma(X,Z) \eta$ for some $Z \in \Gamma(E)$ and all $X \in
  \Gamma(E)$.  Notice that this implies that $E$ is involutive, simply
  by taking $Y \in \Gamma(E)$ in Equation~\eqref{eq:trace-aux} from
  where we see that $T^\nabla(X,Y) \in \Gamma(E)$ and appealing to
  Proposition~\ref{prop:E-involutive}.
\end{proof}

The geometric characterisation of the intrinsic torsion class $\cT_0$
imposes a condition which is easier to describe in terms of $\Omega
\in \Gamma(\wedge^n\ann E)$.  We remark that $X \in \Gamma(E)$ if and
only if $\iota_X \Omega = 0$.

Let us now assume that $T^\nabla \in \cT_0$.

\begin{proposition}
  \label{prop:symmetric-traceless-intrinsic-torsion}
  $T^\nabla \in \cT_0$ if and only if $d\Omega = 0$.
\end{proposition}

\begin{proof}
  From the formula for the differential of $\Omega \in \Omega^{p+1}(M)$,
  we have that
  \begin{multline}
    \label{eq:d-omega}
    d\Omega(X_0,X_1,\dots,X_n) = \sum_{i=0}^{p+1} (-1)^i X_i \Omega(X_0, \dots,\widehat{X_i},\dots,X_{p+1})\\
    + \sum_{0\leq i < j \leq {p+1}} (-1)^{i+j}  \Omega([X_i,X_j], X_0,\dots,\widehat{X_i},\dots,\widehat{X_j},\dots,X_{p+1}),
  \end{multline}
  where the hat denotes omission.  Using that $\nabla \Omega = 0$, we
  may express this purely in terms of the torsion.  (This is not
  unexpected, since the exterior derivative is the skew-symmetrisation
  of the covariant derivative relative to any torsion-free connection,
  so if the connection were torsion-free, $\Omega$ would be closed.)
  From $\nabla\Omega=0$, we may write
  \begin{multline*}
    X_i \Omega(X_0,\dots,\widehat{X_i},\dots,X_{p+1}) = \sum_{0\leq j<i} (-1)^j \Omega(\nabla_{X_i} X_j, X_0,\dots,\widehat{X_j},\dots,\widehat{X_i},\dots,X_{p+1}) \\
               + \sum_{i<j\leq {p+1}} (-1)^{j-1} \Omega(\nabla_{X_i}X_j, X_0,\dots,\widehat{X_i},\dots,\widehat{X_j},\dots,X_{p+1}).
  \end{multline*}
  Relabelling the first sum, we find that
  \begin{multline*}
    \sum_{i=0}^{p+1} (-1)^i X_i \Omega(X_0, \dots,\widehat{X_i},\dots,X_{p+1})\\
    = \sum_{0\leq i < j \leq {p+1}} (-1)^{i+j} \Omega(\nabla_{X_j} X_i- \nabla_{X_i} X_j, X_0,\dots,\widehat{X_i},\dots,\widehat{X_j},\dots,X_{p+1}).
  \end{multline*}
  Inserting this into Equation~\eqref{eq:d-omega} and using the definition of torsion, one finds
  \begin{equation}\label{eq:domega-as-torsion}
    d\Omega(X_0,X_1,\dots,X_{p+1}) =- \sum_{0\leq i < j \leq {p+1}}
    (-1)^{i+j} \Omega(T^\nabla(X_i,X_j), X_0,\dots,\widehat{X_i},\dots,\widehat{X_j},\dots,X_{p+1}).
  \end{equation}
  For $T^\nabla \in \cT_0$, this expression vanishes if any two of the
  $X_i$ belong to $\Gamma(E)$, since either one of the two is outside
  $T^\nabla$, in which case the term vanishes because $\iota_Z\Omega =
  0$ for all $Z \in \Gamma(E)$, or they are both in $T^\nabla$, but
  then $T^\nabla(X,Y) \in \Gamma(E)$ for all $X,Y \in \Gamma(E)$ and
  again it vanishes since $\iota_Z \Omega = 0$ for all $Z \in \Gamma(E)$.

  Let $X_1,\dots,X_D$ be a local frame in the $G$-structure $P$
  adapted to the subbundle $E \subset TM$; that is, such that
  $X_{p+2},\dots,X_D$ define a local frame for $E$ and let
  $\theta^1,\dots,\theta^D$ denote the canonically dual local
  coframe.

  Then $d\Omega(X_{i_0},\dots,X_{i_{p+1}}) =0$ unless precisely one of the
  $i_j \in \{p+2,\dots,D\}$: if this holds for two of the $i_j$ then
  the expression is zero as argued above, and if this holds for none
  of the $i_j$ the expression is zero because at least two of the
  arguments coincide, by the pigeonhole principle.  Without loss of
  generality, let us thus consider $d\Omega(Y,X_1,\dots,X_{p+1})$ with $Y
  \in \Gamma(E)$.  From Equation~\eqref{eq:domega-as-torsion}, 
  \begin{equation}
    \begin{split}
      d\Omega(Y,X_1,\dots,X_{p+1}) &= - \sum_{j=1}^{p+1} (-1)^j  \Omega(T^\nabla(Y,X_j), X_1,\dots,\widehat{X_j},\dots,X_{p+1})\\
      &= \Omega(T^\nabla(Y,X_1),X_2,\dots,X_{p+1}) + \dots + \Omega(X_1,\dots,X_p,T^\nabla(Y,X_{p+1})),
    \end{split}
  \end{equation}
  since any term with $Y$ sitting outside $T^\nabla$ is zero because
  $\iota_Y \Omega = 0$.  Now notice that in a term such as
  \begin{equation}
    \Omega(X_1,\dots,T^\nabla(Y,X_i),\dots,X_{p+1}) 
  \end{equation}
  only the component of $T^\nabla(X,X_i)$ along $X_i$ contributes.  That
  component is $\theta^i(T^\nabla(X,X_i)$ (no sum) and hence
  \begin{equation}
    d\Omega(Y,X_1,\dots,X_{p+1}) = \left( \sum_i \theta^i(T^\nabla(Y,X_i))
    \right) \Omega(X_1,\dots,X_{p+1}),
  \end{equation}
  where the prefactor $\sum_i \theta^i(T^\nabla(Y,X_i))$ is precisely
  the trace of the endomorphism of $TM/E$ defined by $\overline X
  \mapsto \overline{T^\nabla(Y,X)}$, which, by
  Proposition~\ref{prop:alg-char-T-0}, vanishes if and only
  if $T^\nabla$ takes values in $\cT_0$.
\end{proof}

\begin{remark}
  We observe that $d\Omega = 0$ implies that $\eL_X\Omega = 0$ for all
  $X \in \Gamma(E)$.  Indeed, by the Cartan formula, $\eL_X \Omega =
  d\iota_X\Omega + \iota_X d\Omega$, with the first summand vanishing
  since $\iota_X\Omega = 0$ and the second since $d\Omega = 0$.  This
  implies that $E$ is involutive, so that it does not not need to be
  assumed.  This follows from another of the Cartan formulae.  Let
  $X,Y \in \Gamma(E)$, then $\iota_{[X,Y]} \Omega =
  [\eL_X,\iota_Y]\Omega$, but both $\eL_X\Omega = 0$ and $\iota_Y
  \Omega = 0$.  But a vector field $X \in \Gamma(E)$ if
  and only if $\iota_X \Omega = 0$, hence $[X,Y] \in \Gamma(E)$.
\end{remark}

Finally we treat the case of vanishing intrinsic torsion.

\begin{proposition}
  \label{prop:vanishing-intrinsic-torsion}
  An adapted connection $\nabla$ has vanishing intrinsic torsion if
  and only if $\eL_X\eta = 0$ for all $X \in \Gamma(E)$.
\end{proposition}

\begin{proof}
  The calculation in the proof of
  Proposition~\ref{prop:trace-intrinsic-torsion} shows that for all $X
  \in \Gamma(E)$ and all $Y,W \in \eX(M)$,
  \begin{equation}
    (\eL_X\eta)(Y,W) = \eta(T^\nabla(X,Y),W) + \eta(T^\nabla(X,W),Y)
  \end{equation}
  and from Equation~\eqref{eq:im-d}, $T^\nabla \in \im\d$ if and only if
  \begin{equation}
    \eta(T^\nabla(X,Y),W) + \eta(T^\nabla(X,W),Y) = 0,
  \end{equation}
  which implies that for all $X\in \Gamma(E)$ and all $Y,W \in \eX(M)$,
  \begin{equation}
    (\eL_X \eta)(Y,W) = 0.
  \end{equation}
  Abstracting $Y,W$, we arrive at the desired expression.  Notice that
  this also implies that $E$ is involutive, since if both $X,Y \in
  \Gamma(E)$, the above calculation shows that $T^\nabla(X,Y) \in
  \Gamma(E)$ and we can appeal to Proposition~\ref{prop:E-involutive}.
\end{proof}

We may finally summarise our main result as follows.

\begin{theorem}
  \label{thm:summary}
  Let $M$ be a $D$-dimensional manifold with a $G$-structure, with $G$
  (the identity component of) the group defined in
  Section~\ref{sec:group-interest}, and let $E \to M$ denote the
  vector bundle associated to the $G$-submodule $W \subset \RR^D$,
  which is a real vector bundle of rank $D-p-1$, where $1<p<D-3$.  Let
  $\eta \in \Gamma(\odot^2\ann E)$, $h \in \Gamma(\odot^2 E)$ and
  $\Omega \in \Gamma(\wedge^{p+1}\ann E)$ denote the characteristic
  tensor fields of the $G$-structure.  There are five intrinsic torsion
  classes of adapted connections to the $G$-structure in
  nondecreasing order of specialisation:
  \begin{enumerate}[label=($\cT^{\arabic*}$),start=0]
  \item generic intrinsic torsion;
  \item $E$ is involutive, so that $M$ is foliated by integral
    submanifolds of $E$ with $h$ inducing a metric $\gamma$ on each
    leaf of the foliation;
  \item $\eL_X\eta = 2 \gamma(X,Z) \eta$ for all $X \in \Gamma(E)$ and
    some $Z \in \Gamma(E)$;
  \item $d\Omega = 0$; and 
  \item $\eL_X \eta = 0$ for all $X \in \Gamma(E)$ (vanishing intrinsic torsion).
  \end{enumerate}
\end{theorem}

Let us remark that we may read off the results for the particle
Galilean and Carrollian $G$-structures from this theorem.  In
the Galilean case, for which $p=0$, $\eta = \tau^2$ where $\tau \in
\Omega^1(M)$ is the clock one-form, which is actually $\Omega$ in the
above theorem, so that $E = \ker \tau$.  The intrinsic torsion classes
$\cT^1$ and $\cT^2$ coincide: $\ker\tau$ is involutive if and only if
$d\tau = \tau \wedge \alpha$ for some $\alpha \in \Omega^1(M)$, but
then for any $X \in \Gamma(\ker\tau)$
\begin{align*}
  \eL_X \eta &= \eL_X \tau^2\\
             &= 2 \tau \eL_X\tau\\
             &= 2 \tau (d\iota_X + \iota_X d)\tau\\
             &= 2 \tau \iota_X (\tau \wedge \alpha)\\
             &= 2 \alpha(X) \tau^2\\
             &= 2 \alpha(X) \eta.
\end{align*}
Thus $\cT^1 = \cT^2$ agrees with the twistless torsional Galilean
structure.  Similarly, $\cT^3 = \cT^4$ corresponds to the case of
torsionless Galilean structure, since $d\tau = 0$ implies that $\eL_X
\tau = 0$ for all $X \in \Gamma(\ker\tau)$ as we saw above (for
$\Omega$) and hence $\eL_X \eta = 2 \tau \eL_X\tau = 0$.

In the case of particle Carrollian $G$-structure, which corresponds
here to $p=D-2$, $E$ is of rank $1$ and hence it is trivially
involutive, so that $\cT^1 = \cT^0$ is the generic case.  The other
cases typically remain distinct, resulting in the four classes of
Carrollian $G$-structures in \cite{Figueroa-OFarrill:2020gpr}.

In the degenerate case of $D=2$, then as explained in
\cite[Appendix~B.1]{Figueroa-OFarrill:2020gpr}, Galilean and
Carrollian structures coincide.  In the Galilean case, $\cT^1=\cT^0$
since $d\tau \wedge \tau = 0$ by dimension.  So we only have $\cT^0 = \cT^1 =
\cT^2$ as generic intrinsic torsion and $\cT^3=\cT^4$ as vanishing
intrinsic torsion.  The same holds for the Carrollian case.

The only cases which remain are the stringy Galilean ($p=1$) and
Carrollian ($p=D-3$) $G$-structures, for which the intrinsic classes
in Theorem~\ref{thm:summary} must be refined further as discussed
already in Section~\ref{sec:kern-spenc-diff}.  We hope to treat these
cases fully elsewhere.


\section{Dictionary}
\label{sec:dictionary}

In this section we will provide a dictionary to translate the
mathematical treatment in Section~\ref{sec:without} to the more
physical treatment in Section~\ref{sec:TNC2}. From this section
onwards, we will continue to denote the $D$-dimensional spacetime manifold we are
working on by $M$.  Furthermore, where applicable, Einstein summations
over uppercase Roman letters $A,B,C$ etc. run from $0$ to $p$, over
lowercase Roman letters $a,b,c$ etc. run from $p+1$ to $D-1$, and over
hatted capital Roman letters $\hat{A},\hat{B},\hat{C}$ etc. and Greek
letters $\mu,\nu, \rho$ etc. run from $0$ to $D-1$.

\paragraph{}

To write objects such as Vielbeine and (affine) connections in
indices, we need to locally choose a basis of the tangent space at
every point of spacetime which varies smoothly. As a tangent space
$T_pM$ at a specific point $p\in M$ of spacetime does not have a
canonical basis, we will have to choose the basis we want to work with
ourselves. Such a choice of a basis (i.e. linear isomorphism)
$u:\bb{R}^{D} \to T_pM$ at a specific point $p\in M$ of spacetime is
called a \emph{frame} at $p$. Such a choice is not unique: there are
plenty of choices of a basis of a tangent space. As we are considering
spacetime however, we do not allow for all basis transformations any
more, that is to say, there is a group $G\subseteq \GL(D,\bb{R})$ of
all the basis transformations that we still allow for, which we call
the local structure group.\footnote{Note that we are still free to
  choose a basis once, but after we do this, all other bases we could
  use are those related to the one we chose by a transformation in
  $G$.} In the Galilean $p$-brane case, $G$ is a Lie group that is
(abstractly) isomorphic to (the identity component of)
\begin{equation}
  \label{eq:structgroupSNC}
  (\Ort(1,p) \times \Ort(D-p-1)) \ltimes \bb{R}^{(p+1)(D-p-1)}\,.
\end{equation}
The group $G$ can be characterised in terms of the tensors that it
leaves invariant:
\begin{equation}
  \label{eq:structgroupSNCinv}
  G=\{g\in \GL(D,\bb{R}) \ | \ g \cdot h = h, g \cdot \eta = \eta\},
\end{equation}
where we have used $\cdot$ to denote the action of $G$ on the relevant
representations.  In this definition, $W$ is a $(D-p-1)$-dimensional
subspace of $\bb{R}^D$, and $\eta \in \bigodot^2 \ann \vperp$ and
$h\in \bigodot^2 \vperp$ are symmetric tensors. Once we pick a basis
of $\bb{R}^D$, these tensors can be expressed in index notation. If we
let $(e_{\hat{A}})_{\hat{A}=0}^{D-1}$ be a basis of $\bb{R}^D$ such
that $V\coloneqq\text{span}\{e_0,\dots,e_{p}\}$ and $W=\text{span}
\{e_{p+1},\dots,e_{D-1}\}$, and if we let
$(\theta^{\hat{A}})_{\hat{A}=0}^{D-1}$ denote the canonical dual basis, we obtain
\begin{align}
    \eta= \eta_{AB}\theta^A\theta^B \qquad \qquad \text{and} \qquad \qquad h= h^{ab} e_ae_b \,.
\end{align}
With respect to such a basis, the elements of $G$ correspond to $D
\times D$ matrices of the form:
\begin{equation}
  \label{eq:Gpbrane-as-matrix-group}
  \begin{pmatrix}
    \Lambda & \mathbf{0} \\ b & R
  \end{pmatrix}\,, 
\end{equation}
where $\Lambda \in \Ort(1,p)$, $R \in \Ort(D-p-1)$, $b$ is a real
$(D-p-1) \times (p+1)$ matrix and $\mathbf{0}$ denotes the
$(p+1)\times (D-p-1)$ zero matrix. In the physics literature, the
$\Ort(1,p)$ and $\Ort(D-p-1)$ factors are referred to as `longitudinal
Lorentz transformations' and `transversal orthogonal transformations',
respectively, whereas transformations (in the $\bb{R}^{(p+1)(D-p-1)}$
part of $G$) for which $\Lambda$ and $R$ are the identity matrices are
called `$p$-brane Galilean boosts'.

Picking a basis of $\bb{R}^D$, and in particular choosing a complement
$V$ of $W\subseteq \bb{R}^D$, allows us to also find partial inverses
of $\eta$ and $h$, which are
\begin{equation}
  \eta^{-1}=\eta^{AB}e_Ae_B\in {\bigodot}^2 V
\end{equation} 
and
\begin{equation}
    h^{-1}= h_{ab} \theta^a\theta^b\in {\bigodot}^2 \ann V
\end{equation}
defined by $\eta^{AB}\eta_{BC}=\delta_C^A$ and
$h^{ab}h_{bc}=\delta_c^a$, where $\delta$ is the Kronecker delta. In
the rest of this paper, we will frequently raise and lower indices
with these tensors.\footnote{In Section \ref{sec:without}, the
  ``raising and lowering'' of indices is achieved via the musical maps
  $\eta^\flat$ and $h^\sharp$ associated to $\eta$ and $h$.
  Restricting to $W$ or passing to the quotient $\V/W$ they become
  invertible.}  We remark that the subspace $W\subseteq \bb{R}^D$
remains invariant under a basis transformation in $G$, but picking a
complement of $W$ in $\bb{R}^D$, however, is a choice made to
simplify calculations, but one that is not invariant under all the
basis transformations that $G$ allows. If this would have been the
case, the movement of $p$-branes would be confined to a fixed
($p+1$)-dimensional subbundle of $TM$, which would severely hinder the
potential of any possible dynamics.  The desire to work with manifest
$G$-symmetry is the main reason why in Section~\ref{sec:without} we
only use $W$ and $\bb{V}/W$.  This results in our treatment in terms
of filtered representations of $G$; although by passing to their
associated graded representations allows us to make contact with the
choice of complement $V$ here, which in Section~\ref{sec:without}
appears as one of the graded pieces of the filtered representation
$\V$.  The group $G$ respects the filtration, but not the
grading, since the boosts have nonzero degree.  This explains why in
Section~\ref{sec:TNC2}, the symmetry which is manifest is that
associated to the degree-zero part of $G$, namely $\Ort(1,p)\times
\Ort(D-p-1)$, and one has to explicitly check how the boosts act on
the resulting objects.

Furthermore, we do not want to solely study objects at one point of
spacetime, but rather at (local or global) manifolds of spacetime,
which requires to also look at the smooth structure of choosing bases.
First, we note that for every point in our spacetime manifold, we have
a copy of the group $G$ of `allowed' basis transformations. To combine
all of those into a smooth structure, we discover the structure we are
actually studying: a \emph{$G$-structure}, that is, a principal
$G$-subbundle of the frame bundle of $M$. We call the principal bundle
we are studying here $P$. Its projection back to $M$ we denote by
$\pi: P\to M$. Then a choice of basis for all points in a local patch
$U\subseteq M$ of spacetime is given by an \emph{inverse Vielbein}, or a
\emph{moving frame}, that is, a local section of $P$, denoted by
$s:U\to P$ without indices.\footnote{We note that the terminology diverges on the definition of (inverse) Vielbeine, as moving frames are also commonly referred to as Vielbeine, and the term inverse Vielbein is in that case reserved for the moving coframes. We opt for the other convention, as it is in line with \cite{Bergshoeff:2022fzb}. } To denote an inverse Vielbein with
indices, we first write $E=s\in \Gamma(U,P_U)$ with
\begin{equation}
  \label{eq:dictinvvielb}
  E = {E_{\hat{A}}}^\mu\frac{\partial}{\partial x^\mu} \otimes \theta^{\hat{A}},
\end{equation}
where $(x^{\mu})_{\mu=0}^{D-1}: U\to \bb{R}^D$ are local coordinates, and where ${E_{\hat{A}}}^\mu\in C^\infty(U)$ is given by 
\begin{equation}
  {E_{\hat{A}}}^\mu(p) = d{x^\mu}_p (s(p)e_{\hat{A}}).
\end{equation}
We can split ${E_{\hat{A}}}^\mu$ up in ${\tau_A}^\mu,{e_a}^\mu\in C^\infty(U)$, which are given by
\begin{equation}
  {\tau_A}^\mu(p) \coloneqq {E_A}^\mu(p)= d{x^\mu}_p (s(p)e_A)
\end{equation}
and 
\begin{equation}
  {e_a}^\mu(p) \coloneqq {E_a}^\mu(p)=  d{x^\mu}_p (s(p)e_a).
\end{equation}
In the physics literature, $\tau_{A}{}^\mu (p)$ and $e_a{}^\mu (p)$
are referred to as the longitudinal and transversal inverse
Vielbeine, respectively.

Similarly, given an inverse Vielbein $s:U\to P$, we also have a notion
for a choice of a dual basis for all points in a local patch
$U\subseteq M$. Such a dual basis is given by a \emph{Vielbein}, or
a \emph{moving coframe}, that is, a local $\bb{R}^D$-valued one-form
of $U$, denoted by $s^*\theta \in \Omega(U,\bb{R}^D) $ without
indices.  Here, $\theta\in \Omega^1(P,\bb{R}^D)$ is the soldering form
on $P$. It picks out the part of tangent vectors at a point of $P$
that does not consider possible changes of bases above a point,
projects it down to a tangent vector on $M$ and composes it with the
(inverse of the) frame that is the considered point of $P$.
In a formula, it is given by
\begin{equation}
  \theta_u(X_u) = u^{-1}(\pi_* X_u),
\end{equation}
where $u\in P$ is a frame at a certain point of spacetime,
and where $X_u\in T_uP$ is a tangent vector of the $G$-structure.  In
words, it gives the components of the projection $\pi_*X_u$ of $X_u$
to $M$ relative to the frame defined by $u$.  Note that this is
well-defined because the $G$-structure is a subbundle of the frame
bundle.

To denote a Vielbein with indices, we first write $\tilde{E}=s^*\theta\in \Omega^1(U,\bb{R}^{D})$ with 
\begin{align}
\label{eq:dictvielb}
     \tilde{E} = {E_\mu}^{\hat{A}} d x^\mu e_{\hat{A}},
\end{align}
where ${E_\mu}^{\hat{A}} \in C^\infty(U)$ is given by 
\begin{align}
     \begin{split}
         {E_\mu}^{\hat{A}}(p) &= \theta^{\hat{A}}\left( s^*(\theta)(\frac{\partial}{\partial x^\mu}(p))\right) = \theta^{\hat{A}} \Big(s(p)^{-1}\frac{\partial}{\partial x^\mu}(p)\Big).
     \end{split}
\end{align}
Similarly as before, we can split ${E_{\mu}}^{\hat{A}}$ up in  ${\tau_\mu}^A,{e_\mu}^a\in C^\infty(U)$, which are given by
\begin{align}
    {\tau_\mu}^A(p) \coloneqq {E_\mu}^A(p)
     = \theta^A \Big(s(p)^{-1}\frac{\partial}{\partial x^\mu}(p)\Big)
\end{align}
and 
\begin{align}
    {e_\mu}^a(p) \coloneqq {E_\mu}^a(p)
     = \theta^a \Big(s(p)^{-1}\frac{\partial}{\partial x^\mu}(p)\Big).
\end{align}
In the physics literature, $\tau_\mu{}^{A} (p)$ and $e_\mu{}^a (p)$ are usually referred to as the longitudinal, respectively transversal Vielbeine. Equations \eqref{eq:dictinvvielb} and \eqref{eq:dictvielb} illustrate why (inverse) Vielbeine are often thought of as connecting curved indices (that is, vector fields and differential forms on the spacetime coordinates) to flat indices (a (dual) vector space with a preferred basis, such as $\bb{R}^D$ or $(\bb{R}^D)^*$).

Note that the above definitions imply that $E_{\hat{A}}{}^\mu$ and $E_\mu{}^{\hat{A}}$ are each other's inverse as $D \times D$ matrices, since
\begin{align}
    \begin{split}
        {E_\mu}^{\hat{A}}(p)  {E_{\hat{B}}}^\mu(p)
        &= \theta^{\hat{A}} \Big(s(p)^{-1}\frac{\partial}{\partial x^\mu}(p)\Big) d{x^\mu}_p (s(p)e_{\hat{B}})
        =\theta^{\hat{A}} \Big(s(p)^{-1}s(p)e_{\hat{B}}\Big) =\delta_{\hat{B}}^{\hat{A}}
    \end{split}
\end{align}
and that
\begin{align}
    \begin{split}
         {E_\mu}^{\hat{A}}(p)  {E_{\hat{A}}}^\nu(p) &
         = d{x^\nu}_p \Big(s(p)\big(e_{\hat{A}}\theta^{\hat{A}}  (s(p)^{-1}\frac{\partial}{\partial x^\mu}(p)\big)\Big)
         = d{x^\nu}_p (\frac{\partial}{\partial x^\mu}(p)) =\delta_\mu^\nu.
    \end{split}
\end{align}
In terms of $\tau_\mu{}^A$, $e_\mu{}^a$, $\tau_A{}^\mu$, $e_a{}^\mu$, this invertibility is expressed as:
\begin{align}
\begin{aligned}
  \label{eq:invVielbeineSNC}
  & \tau_A{}^\mu \tau_\mu{}^B = \delta_A^B \,  \qquad & & \tau_A{}^\mu e_\mu{}^{a} = 0 \, \qquad  & & e_{a}{}^\mu \tau_\mu{}^A = 0 \,\\
  & e_\mu{}^{a} e_{b}{}^\mu = \delta_{b}^{a} \,  \qquad & & \tau_\mu{}^A \tau_A{}^\nu + e_\mu{}^{a} e_{a}{}^\nu = \delta_\mu^\nu \,. & &
\end{aligned}
\end{align}

As can be seen from the aforementioned definitions, the explicit formulation of an (inverse) Vielbein depends on a chosen local frame field, that is, a section $s:U\to P$. Two different sections $s,s':U\to P$ are related via a transformation $g \in C^\infty(U,G)$, according to:\footnote{Technically, we would like to know how two different sections $s:U\to P$, $s': U'\to P$ can be related on the overlap $U\cap U'$ of their neighbourhoods, but for clarity, we will assume that $U'=U$.}
\begin{align}
    s'(p) = s(p) \circ g(p),
\end{align}
so $g(p) = s(p)^{-1} \circ s'(p)$.
In terms of indices, we thus find that the inverse Vielbeine transform as follows when we choose a different (local) frame field:
\begin{align} \label{eq:indvielbtrafo}
    \begin{split}
        E'_{\hat{A}}{}^\mu (p) :&= d{x^\mu}_p (s'(p) e_{\hat{A}}) 
    = d{x^\mu}_p ((s(p) \circ g(p))e_{\hat{A}})
    = {E_{\hat{B}}}^\mu(p)\theta^{\hat{B}}g(p)e_{\hat{A}}.
    \end{split}
\end{align}
For the Vielbeine, we find
\begin{align} \label{eq:indinvvielbtrafo}
   \begin{split}
        E'_\mu{}^{\hat{A}} (p):&= \theta^{\hat{A}}({s'}^*\theta)_p\Big(\frac{\partial}{\partial x^\mu}(p) \Big)
    = \theta^{\hat{A}}g(p)^{-1} s(p)^{-1}\Big(\frac{\partial}{\partial x^\mu}(p) \Big)
   = \theta^{\hat{A}}g(p)^{-1}e_{\hat{B}} {E_\mu}^{\hat{B}} (p)\,.
   \end{split}
\end{align}
In the physics literature, one usually writes these transformation rules in infinitesimal form. To do this, one writes $g(p) = \exp(-\lambda(p))$, where $\lambda$ is valued in the Lie algebra $\mathfrak{g}$ of $G$, i.e., its matrix elements $\lambda(p)^{\hat{A}}{}_{\hat{B}}\coloneqq \theta^{\hat{A}} \lambda(p) e_{\hat{B}}$ obey (omitting the underlying point of spacetime regularly from now on):
\begin{align}
  \lambda^{AB} \coloneqq {\lambda^A}_C\eta^{BC} = - \lambda^{BA} \, \qquad \quad \lambda^{ab} \coloneqq {\lambda^a}_ch^{bc}= - \lambda^{ba} \, \qquad \quad \lambda(p)^{A}{}_{a} = 0 \,.
\end{align}
Using $\delta E_\mu{}^{\hat{A}}$ to denote the difference $E'_\mu{}^{\hat{A}} - E_\mu{}^{\hat{A}}$ to first order in $\lambda$ (and mutatis mutandis for $E_{\hat{A}}{}^\mu $), Equations \eqref{eq:indvielbtrafo} and \eqref{eq:indinvvielbtrafo} then lead to:
\begin{align}
  \begin{split}
      \delta E_\mu{}^{\hat{A}} &= E_\mu{}^{\hat{B}} \theta^{\hat{A}} \lambda e_{\hat{B}} = \lambda^{\hat{A}}{}_{\hat{B}} E_\mu{}^{\hat{B}} \, \\
  \delta E_{\hat{A}}{}^\mu &= - E_{\hat{B}}{}^{\mu} \theta^{\hat{B}} \lambda e_{\hat{A}} = - E_{\hat{B}}{}^{\mu}  \lambda^{\hat{B}}{}_{\hat{A}} \,,
  \end{split}
\end{align}
or equivalently:
\begin{alignat}{2} 
  \delta \tau_\mu{}^A  &=  \lambda^A{}_B \tau_\mu{}^B \, \qquad &  \delta e_\mu{}^{a}  &= \lambda^{a}{}_{b} e_\mu{}^{b} + \lambda^a{}_A \tau_\mu{}^A \,  \label{eq:localtrafosframeSNC} \\
    \delta \tau_A{}^\mu &=  -\lambda^B{}_A \tau_B{}^\mu - \lambda^a{}_{A} e_{a}{}^\mu \, \qquad & \delta e_{a}{}^\mu &= -\lambda^{b}{}_a e_{b}{}^\mu \,.  \label{eq:localtrafosinvSNC}
\end{alignat}

To be able to relate different choices of basis at different points of
spacetime, we also have to introduce the notion of a connection.
Without indices, a \emph{connection one-form}
$\omega\in \Omega^1(P,\gm)$ is determined by a horizontal
$G$-invariant distribution $\ker \omega_u$, which says, among other
things, that $\ker(\pi_*)_u \oplus \ker \omega_u = T_uP $ for all
$u\in P$ and $\omega(\xi_x) = x$ for $x\in \gm$, where
$\xi_x\in \mathfrak{X}(P)$ is defined by
$\xi_x(u) = \frac{d}{dt}\big(u\circ e^{tx}\big)
\big|_{t=0}$.\footnote{Note that this definition of a connection is
  equivalent to that of an Ehresmann connection, where the horizontal
  $G$-invariant distribution of $P$ is given by
  $\mathcal{H}_u = \ker \omega_u$ for every $u\in P$.} Given a choice
of inverse Vielbeine $s\in (U,P_U)$, we can also deduce the expression
of the spin connection $\Omega \coloneqq s^*\omega \in
\Omega^1(U,\gm)$ with indices by
writing \begin{equation}
  \Omega = \Omega_\mu dx^\mu = {J_{\hat{A}}}^{\hat{B}}{{\omega_\mu}^{\hat{A}}}_{\hat{B}} dx^\mu ,
\end{equation} 
with ${{\omega_\mu}^{\hat{A}}}_{\hat{B}}\in C^\infty(U)$ given by
\begin{equation}
    {{\omega_\mu}^{\hat{A}}}_{\hat{B}}(p) = {\overline{J}^{\hat{A}}}_{\hat{B}}\left( (s^*\omega )_p (\frac{\partial}{\partial x^\mu}(p))\right) 
    = {\overline{J}^{\hat{A}}}_{\hat{B}}\left( \omega_{s(p)}(d s_p \frac{\partial}{\partial x^\mu}(p))\right).
\end{equation}
In these equations, the maps ${J_{\hat{A}}}^{\hat{B}}\in\text{End}(\bb{R}^D)$ are defined by ${J_{\hat{A}}}^{\hat{B}}e_{\hat{C}} = \delta_{\hat{C}}^{\hat{B}} e_{\hat{A}}$ and the ${\overline{J}^{\hat{A}}}_{\hat{B}}\in \text{End}((\bb{R}^D)^*) $ form the canonical dual basis. Note that, since $\omega$ is valued in the Lie algebra $\mathfrak{g}$ of $G$, one has:
\begin{equation}
  \omega_{\mu}{}^{AB}  \coloneqq{{\omega_\mu}^A}_B\eta^{BC} = - \omega_\mu{}^{BA} \, \qquad \quad \omega_\mu{}^{ab}  \coloneqq{{\omega_\mu}^a}_bh^{bc}= - \omega_\mu{}^{ba}  \, \qquad \quad \omega_\mu{}^{A}{}_a  = 0 \,.
\end{equation}
In the physics literature, one refers to $\omega_\mu{}^{AB}$,
$\omega_\mu{}^{ab}$ and $\omega_\mu{}^a{}_A$ as the spin connections
for longitudinal Lorentz transformations, transversal rotations and
$p$-brane Galilean boosts.

As can be found in \cite{Figueroa-OFarrill:2020gpr}, an equivalent way to express the connection 1-forms is the affine connection $\nabla: \mathcal{X}(M) \to \Omega^1(M,TM)$, which is given with indices in local coordinates by 
\begin{equation}
    \nabla = \Gamma_{\mu\nu}^\rho dx^\mu\otimes \frac{\partial}{\partial x^\rho}dx^\nu ,
\end{equation}
with $\Gamma_{\mu\nu}^\rho\in C^\infty(U)$ given by
\begin{equation}
    \Gamma_{\mu\nu}^\rho(p)= d{x^\rho}_p(\nabla_{ \frac{\partial}{\partial x^\mu}} \frac{\partial}{\partial x^\nu}).
\end{equation}
The affine connection $\nabla$ is induced from the connection 1-form $\omega$ (as described in \cite{Figueroa-OFarrill:2020gpr}), i.e., $\nabla$ and $\omega$ are related via the Vielbein postulates 
\begin{equation}
  \label{eq:VielbpostSNC}
  \begin{split}
    & \partial_\mu \tau_\nu{}^A - \omega_\mu{}^A{}_B \tau_\nu{}^B -  \Gamma_{\mu\nu}^\rho  \tau_\rho{}^A  = 0 \,\\
    &  \partial_\mu e_\nu{}^{a} - \omega_\mu{}^{ab} e_{\nu b} - \omega_\mu{}^{a A} \tau_{\nu A} - \Gamma_{\mu\nu}^\rho e_{\rho}{}^{a} = 0 \,.
  \end{split}
\end{equation}
The affine connection $\nabla$ is induced from the spin connection $\omega$ in such a way that it preserves the local smooth versions of the rank-$(p+1)$ `longitudinal' metric
\begin{equation}
\label{eq:localversioneta}
    \eta = dx^\mu \odot dx^\nu \tau_{\mu\nu} \in \Gamma(U,\bigodot{}^2T^*U) \qquad \text{with } \tau_{\mu\nu} = \tau_\mu{}^A\tau_{\nu}{}^B \eta_{AB}
\end{equation}
and the rank-$(D-p-1)$ `transversal' co-metric
\begin{equation}
\label{eq:localversionh}
    h = \frac{\partial}{\partial x^\mu} \odot \frac{\partial}{\partial x^\nu}h^{\mu\nu} \in \Gamma(U,\bigodot{}^2TU) \qquad \text{with  }h^{\mu\nu} = e_a{}^\mu e_b{}^\nu h^{ab}.
\end{equation}

Similarly as in the case with the Vielbeine, we would like to know how the spin connection transforms under a change of local frame field. When we reintroduce our local sections $s,s': U\to P$, we find that the spin connection transforms by 
\begin{align}
    ((s')^*\omega)_p
        = \text{Ad}(g(p)^{-1})\circ (s^*\omega)_{p}+  g(p)^{-1}\circ d g(p) \,.
\end{align}
In infinitesimal form, this transformation rule is given by:
\begin{align}
  \label{eq:localtrafosdepSNC}
    \delta \omega_\mu{}^{AB} &= \partial_\mu \lambda^{AB} + 2 \lambda^{[A|C|} \omega_{\mu C}{}^{B]} \, \qquad \qquad
    \delta \omega_\mu{}^{ab} = \partial_\mu \lambda^{ab} + 2 \lambda^{[a|c|} \omega_{\mu c}{}^{b]} \, \nonumber \\[.1truecm]
  \delta \omega_\mu{}^{a A} &= \partial_\mu \lambda^{a A} + \lambda^A{}_B \omega_\mu{}^{a B} + \lambda^{a}{}_{b} \omega_\mu{}^{b A} -  \lambda^{a B} \omega_\mu{}^A{}_B + \lambda^{b A} \omega_{\mu b}{}^{a} \,.
\end{align}

Using a connection 1-form, we can now discuss the Spencer differential. In the mathematics section, this is given by as a map $\partial: \Hom(\bb{V},\gm)\to \Hom(\wedge^2 \bb{V},\bb{V})$. 
This is a pointwise map that suffices for the mathematical treatment we are considering, because it is a map of $G$-representations. To be more concrete, because $\partial: \Hom(\bb{V},\gm)\to \Hom(\wedge^2 \bb{V},\bb{V})$ is $G$-equivariant, it induces a bundle map 
\begin{align}
    \begin{split}
        P\times_G \Hom(\bb{V},\gm) &\to P \times_G \Hom(\wedge^2 \bb{V},\bb{V})\\
        [(u,\kappa) ] &\mapsto [(u,\partial \kappa)],
    \end{split}
\end{align}
which in turn induces the Spencer map $\partial: \Omega^1(M,P\times_G
\gm) \to \Omega^2(M,P\times_G \bb{V})$, 
from the space of differences of connection 1-forms
$\Omega^1(M,P\times_G \gm)$ to the space of all possible torsions
$\Omega^2(M,P\times_G \bb{V})$, by identifying the associated vector
bundles $P\times_G \wedge^k \bb{V}^*$ to the space of differential
$k$-forms $\Omega^k(M)$ on $M$ using the soldering form.

This procedure of first considering the associated vector bundle and
secondly its isomorphism under the soldering form to (co)tangent
bundles applies also for $G$-submodules such as the subspace
$W\subset \bb{V}$ from the mathematics section, and the intrinsic
torsion classes $\mathcal{T}^i$ ($i=0,1,2,3,4$),\footnote{We note that
  the way mathematicians and physicists denote the classes of
  intrinsic notion is different, as mathematicians like to denote the
  $G$-submodule in which the torsion resides, such as
  $T\in \mathcal{T}_0$, whereas physicists prefer to denote equations
  that the torsion has to satisfy, such as ${{T_a}^A}_A=0$.} but also
the (partial) metrics $\eta \in \bigodot^2 \ann \vperp$ and
$h\in \bigodot^2 \vperp$, which are already given in Equation
\eqref{eq:localversioneta}-\eqref{eq:localversionh}. One of the
consequences hereof is that all connection 1-forms
$\omega\in \Omega^1(P,\gm)$ (and thus also all affine connections) are
compatible with the (partial) metrics. It also allows us to do the
mathematics in Section \ref{sec:without} pointwise, as long as the
objects we study are invariant under $G$-transformations.

Since the submodule in which the intrinsic torsion resides is an
invariant for the (spacetime) manifold $M$, that is, the constraints
the intrinsic torsion satisfies do not change under a transformation
of bases, it encodes geometric information about the manifold. The
specific geometric characterization of those different constraints is
given in Theorem~\ref{thm:summary} and we will elaborate on this
result and its physical implications in Section~\ref{sec:TNC2}.


\section{\texorpdfstring{$p\,$}{p}-brane Galilean Geometries with Indices} \label{sec:TNC2}

The aim of this section is to discuss the results on $p$-brane Galilean geometries of Section \ref{sec:without} using a physics language and index notation, explained in the dictionary \ref{sec:dictionary}, which is less precise but hopefully more easily accessible to a physics audience. Here, we will study $p$-brane Galilean geometries by trying to obtain explicit expressions for their spin connection components from Cartan's first structure equation and we will explain how this analysis reflects the results of that of the Spencer differential of Section \ref{sec:without}.  

We start from Cartan's first structure equation of $p$-brane Galilean geometry, which is obtained by antisymmetrizing the Vielbein postulates \eqref{eq:VielbpostSNC}:
\begin{align}
  T_{\mu\nu}{}^A &= 2 \partial_{[\mu} \tau_{\nu]}{}^A - 2 \omega_{[\mu}{}^A{}_B \tau_{\nu]}{}^B \, \label{eq:torsionVielbSNC1} \\  E_{\mu\nu}{}^{a} &= 2 \partial_{[\mu} e_{\nu]}{}^{a} - 2 \omega_{[\mu}{}^{ab} e_{\nu] b} - 2 \omega_{[\mu}{}^{aA} \tau_{\nu] A} \label{eq:torsionVielbSNC2} \,.
\end{align}
Here, we have split the torsion tensor $2 \Gamma^\rho_{[\mu\nu]}$ into `longitudinal torsion tensor' components $T_{\mu\nu}{}^A$ along $\tau_A{}^\rho$ and `transversal torsion tensor' components $E_{\mu\nu}{}^{a}$ along $e_{a}{}^\rho$ as follows:
\begin{align} \label{eq:torsiondecompSNC}
  & 2 \Gamma^\rho_{[\mu\nu]} = \tau_A{}^\rho T_{\mu\nu}{}^A + e_{a}{}^\rho E_{\mu\nu}{}^{a} \qquad \text{or} \qquad
    T_{\mu\nu}{}^A  = 2 \Gamma^\rho_{[\mu\nu]} \tau_\rho{}^A  \quad \text{and} \quad E_{\mu\nu}{}^{a} = 2 \Gamma^\rho_{[\mu\nu]} e_\rho{}^{a} \,.
\end{align}
Note that under local SO$(1,p)$, SO$(D-p-1)$ and $p$-brane Galilean boosts, $T_{\mu\nu}{}^A$ and $E_{\mu\nu}{}^{a}$  transform as follows:
\begin{align}
  \label{eq:trafoTATAp}
  \delta T_{\mu\nu}{}^A = \lambda^A{}_B T_{\mu\nu}{}^B \, \qquad \qquad \qquad \delta E_{\mu\nu}{}^{a} = \lambda^{a}{}_{b} E_{\mu\nu}{}^{b} + \lambda^{a}{}_A T_{\mu\nu}{}^A \,.
\end{align}
Rewriting \eqref{eq:torsionVielbSNC1} and \eqref{eq:torsionVielbSNC2} as
\begin{align}
  & \hat{T}_{\mu\nu}{}^A = - 2 \omega_{[\mu}{}^A{}_B \tau_{\nu]}{}^B  \, \label{eq:torsionVielbSNC12} \\
  & \hat{E}_{\mu\nu}{}^a = - 2 \omega_{[\mu}{}^{ab} e_{\nu] b} - 2 \omega_{[\mu}{}^{aA} \tau_{\nu] A} \,, \label{eq:torsionVielbSNC22} \\
  & \text{with } \hat{T}_{\mu\nu}{}^A \equiv T_{\mu\nu}{}^A - 2 \partial_{[\mu} \tau_{\nu]}{}^A  \text{ and }  \hat{E}_{\mu\nu}{}^a \equiv E_{\mu\nu}{}^a - 2 \partial_{[\mu} e_{\nu]}{}^a \label{eq:defhatThatE} \,,
\end{align}
one can attempt to solve these equations for the spin connection components. In doing so, one however finds that they do not determine the spin connections uniquely. In particular, careful analysis shows that some (combinations of) spin connection components do not occur in any of the above equations for the $\hat{T}$- and $\hat{E}$-components. These spin connection components can thus not be expressed in terms of $\hat{T}/\hat{E}$-components and instead parametrize a family of adapted connections. Furthermore, some of the equations \eqref{eq:torsionVielbSNC12}, \eqref{eq:torsionVielbSNC22} for the $\hat{T}/\hat{E}$-components do not contain any of the spin connection components so that those $\hat{T}/\hat{E}$-components vanish. The torsion tensor components that appear in them then constitute the intrinsic torsion of $p$-brane Galilean geometry. 

Below, we will verify these statements explicitly. Our findings can however be anticipated from the analysis of the Spencer differential $\partial$ in $p$-brane Galilean geometry of Section \ref{sec:without}. Indeed, the existence of (combinations of) spin connection components that do not occur in any of the $\hat{T}$- and $\hat{E}$-components in Equations \eqref{eq:torsionVielbSNC12}, \eqref{eq:torsionVielbSNC22} reflects that $\ker \partial$ is non-trivial and the number of such (combinations of) spin connection components is given by $\dim\ker \partial$. The fact that some of the $\hat{T}/\hat{E}$-components in \eqref{eq:torsionVielbSNC12}, \eqref{eq:torsionVielbSNC22} do not contain any of the spin connection components is related to the non-triviality of $\coker\partial$. The dimension $\dim \coker\partial$ equals the number of such $\hat{T}/\hat{E}$-components, while $\dim\im \partial $ gives the number of $\hat{T}/\hat{E}$ components that do contain connection components. The information about the system of equations \eqref{eq:torsionVielbSNC12}, \eqref{eq:torsionVielbSNC22} that is encoded in the kernel, cokernel and image of the Spencer differential can thus be summarized as follows:
\begin{eqnarray}
\ker \partial &\leftrightarrow& \textrm{those combinations of spin connection\ components\ that\ do\ not}\nonumber \\ && \textrm{occur in any of the $\hat{T}$- and $\hat{E}$-components;}\\[.1truecm]
\coker \partial &\leftrightarrow& \textrm{those $\hat{T}$- and $\hat{E}$-components that do not contain any of the}\nonumber\\
&&\textrm{spin connection components; and}\\[.1truecm]
\im \partial &\leftrightarrow& \textrm{those $\hat{T}$- and $\hat{E}$-components that do contain a spin connection}\nonumber\\
&&\textrm{component.}
\end{eqnarray}

In order to make the above discussion explicit, we distinguish the different spin connection and $\hat{T}$- and $\hat{E}$-components, by decomposing curved $\mu$-indices into longitudinal $A$- and transversal $a$-indices using the following decomposition rule (illustrated here for a one-form $V_\mu$):
\begin{equation}\label{decomposition}
V_\mu = \tau_\mu{}^A V_A + e_\mu{}^a V_a\hskip .5truecm \textrm{or}\hskip .5truecm V_A = \tau_A{}^\mu V_\mu\ \ \textrm{and}\ \ V_a = e_a{}^\mu V_\mu\,.
\end{equation}
Equations \eqref{eq:torsionVielbSNC12}, \eqref{eq:torsionVielbSNC22} can then be decomposed as follows:\,\footnote{Sometimes, if confusion could arise, we put a comma between the indices to indicate which two indices form (the projection of) an anti-symmetric pair, e.g.
\begin{align}
    \omega_{A,BC} = \tau_A{}^\mu\, \omega_{\mu BC}\,\hskip .5truecm
  \textrm{or}\hskip .5truecm  \hat{E}_{A[b,c]} = \tau_A{}^\mu e_{[b}{}^\nu \hat{E}_{|\mu\nu|c]}\,.
\end{align}\label{notation1}}
\begin{alignat}{2}\label{decompositions}
 \hat{T}_{BC,A} &= -2 \, \omega_{[B,|A|C]}  \, & \qquad \qquad \qquad \qquad \hat{E}_{bc,a} &= -2\, \omega_{[b,|a|c]} \, \nonumber \\
 \hat{T}_{a[A,B]} &= \omega_{a,AB} \, & \qquad \qquad \qquad \qquad \hat{E}_{A[a,b]} &= \omega_{A,ab} + \omega_{[a,b]A} \, \nonumber \\
 \hat{T}_{a(A,B)} &= 0 \, & \qquad \qquad \qquad \qquad \hat{E}_{A(a,b)} &= \omega_{(a,b)A} \, \nonumber \\
 \hat{T}_{ab,A} &= 0 \, & \qquad \qquad \qquad \qquad \hat{E}_{AB,a} &= -2\, \omega_{[A,|a|B]} \,.
\end{alignat}

The two equations on the first line can be seen to be equivalent to:
\begin{align}
  \hat{T}_{C[A,B]} - \frac12 \hat{T}_{AB,C} = \omega_{C,AB} \, \qquad \qquad \hat{E}_{c[a,b]} - \frac12 \hat{E}_{ab,c} = \omega_{c,ab} \,.
\end{align}
From this and the remaining equations in \eqref{decompositions}, one sees that the following (combinations of) spin connection components can be expressed in terms of $\hat{T}$- and $\hat{E}$-components:
\begin{equation}\label{dependent}
\hskip -.25truecm  \textrm{dependent spin connection components}\ :\ \omega_{[A}{}^a{}_{B]}\,,\,\, \omega^{(ab)C}\,,\,\,  \omega_{c,ab} \,,\,\, \omega_\mu{}^{AB} \,\,\, \textrm{and} \,\,\, \omega_C{}^{ab} + \omega^{[ab]}{}_C\,.
\end{equation}
The remaining spin connection components do not occur in any of the $\hat{T}$- and $\hat{E}$-components and remain as independent components that parametrize a family of adapted connections in $p$-brane Galilean geometry:
\begin{equation}\label{independent}
  \textrm{independent spin connection components}\, (\leftrightarrow \ker \partial): \omega_{\{A}{}^a{}_{B\}}\,,\omega^{Aa}{}_A
\ \ \textrm{and} \ \ \omega_C{}^{ab} - \omega^{[ab]}{}_C\,.
\end{equation}
We use here a notation where $\{AB\}$ indicates the symmetric traceless part of $AB$. The cases  $p = 0$ ($p =D-2$) are special with only one longitudinal  (transverse) direction. In those cases, the first (last) spin connection component vanishes. Note that the number of independent spin connection components is equal to $D(p+1)(D-p-1)/2$, which indeed agrees with the dimension of the kernel of the Spencer differential $\partial$ in $p$-brane Galilean geometry:
\begin{align} \label{eq:indepspinconndim}
  \#(\textrm{independent spin connection components}) = \dim \ker \partial = \frac12 D (p+1) (D-p-1) \,.
\end{align}

From the equations \eqref{decompositions} we also see that there are certain components of $\hat{T}_{\mu\nu}{}^A$ that do not contain any spin connection component. These are a manifestation of the non-triviality of $\coker \partial$ and their corresponding torsion tensor components $T_{\mu\nu}{}^A = \hat{T}_{\mu\nu}{}^A + 2 \partial_{[\mu} \tau_{\nu]}{}^A$ are the \textbf{intrinsic torsion} components of Section \ref{sec:without}. Setting (some of the) intrinsic torsion components to zero leads to geometric constraints and it is these different geometric constraints that lead to a classification of the possible $p$-brane Galilean geometries. Using \eqref{decompositions}, we find the following intrinsic torsion components:
\begin{align}
  \label{eq:intrinsictorsionSNC}
\hskip -1.5truecm \textrm{intrinsic torsion components}\, (\leftrightarrow \coker \partial):\hskip 1truecm T_a{}^{\{AB\}}\,,\,\, T_a{}^A{}_A \hskip .3truecm \textrm{and} \hskip .3truecm  T_{ab}{}^A \,.
\end{align}
Note that there are no $\hat{E}$-components that give rise to intrinsic torsion.\,\footnote{This is related to the fact that according to \eqref{eq:torsionVielbSNC22} these $\hat{E}$-components contain {\it two} different spin connection components which, as it turns out,  cannot be projected away both at the same time. Note that it can also be seen directly from Equation \eqref{eq:im-d} that every component of torsion that lies in $\Hom(\wedge^2 \V,W)$ (including all the $\hat{E}$-components) is in the image of $\partial$.} The number of intrinsic torsion components is given by the dimension of $\coker \partial$:
\begin{align} 
  \#(\textrm{intrinsic torsion components}) = \dim\coker \partial = \frac12 D (p+1) (D-p-1) \,.
\end{align}
Note that this equals the number \eqref{eq:indepspinconndim} of independent spin connection components. Furthermore, the intrinsic torsion components \eqref{eq:intrinsictorsionSNC} fall in the same representations of the structure group as the independent spin connection components \eqref{independent}. This reflects the fact that $\ker \partial$ and $\coker \partial$ are isomorphic, as was shown in Section \ref{sec:without}.  

The remaining $\hat{T}$- and $\hat{E}$-components all contain a spin connection component. Their number is given by the dimension of $\im \partial$. In the (super-)gravity literature, the equations of \eqref{eq:torsionVielbSNC12} and \eqref{eq:torsionVielbSNC22} that correspond to these $\hat{T}$- and $\hat{E}$-components are often called \textbf{conventional constraints}. They do not lead to constraints on the geometry but can instead be used to solve some of the spin connection components in terms of the Vielbein fields and their derivatives.
In the case at hand we find the following conventional $\hat{T}$- and $\hat{E}$-components:
\begin{align}
\label{conventional}
\hskip -3truecm \textrm{ conventional $\hat{T}$- and $\hat{E}$-components} \ (\leftrightarrow \im  \partial):\hskip .5truecm \hat{T}_a{}^{[AB]}\,,\,\,  \hat{T}_{AB}{}^C \,,\,\,   \hat{E}_{\mu\nu}{}^a\,.
\end{align}
It should be emphasized that the (combinations of) spin connection components given in Equation~\eqref{dependent} are dependent only if we impose the maximum set of conventional constraints (corresponding to the components given in \eqref{conventional}). Later, in Section \ref{sec:galgrav}, we will encounter an example where we only apply those conventional constraints that follow as equations of motion of a first-order action. In that case we will not impose the maximum set of conventional constraints and consequently we cannot solve for all the spin connection components given in Equation~\eqref{dependent}. Assuming that the maximum number of conventional constraints is imposed, we find the following solutions for the spin connection components: 
\begin{align}
    \omega_{C,AB}(\tau,T) &= \hat{T}_{C[A,B]} - \frac12 \hat{T}_{AB,C}\,,
    &\omega_{a,AB}(\tau, e,T) &= \hat{T}_{a[A,B]}\,,\label{omega1}\\[.1truecm]
    \omega_{c,ab}(e,E) &= \hat{E}_{c[a,b]} - \frac{1}{2} \hat{E}_{ab,c}\,,
    &\omega_{[A|,a|B]}(\tau, e, E) &= -\frac{1}{2} \hat{E}_{AB, a}\,,\label{omega2}\\[.1truecm]
    (\omega_{A,ab} + \omega_{[a,b]A})(\tau, e, E) &= \hat{E}_{A[a,b]}\,,
    &\omega_{(a,b)A}(\tau, e, E) &= \hat{E}_{A(a,b)}\,,\label{omega3}
\end{align}
where $\hat{T}_{\mu\nu}{}^A$ and $\hat{E}_{\mu\nu}{}^a$ should be replaced by their expressions given in \eqref{eq:defhatThatE}.

We now continue with the classification of the Galilean $p$-brane geometries. We will only classify the representations of $\coker \partial$,
 i.e.~the intrinsic torsion tensors, since these can give rise to geometric constraints.\,\footnote{A classification of the other torsion tensor components in the case of string Galilei geometry has been given in \cite{Bergshoeff:2022fzb}.}  The different intrinsic torsion  components given in Equation~\eqref{eq:intrinsictorsionSNC} transform under longitudinal Lorentz transformations, transversal spatial rotations and $p$-brane Galilean boosts. Under  Galilean boosts, some components of the intrinsic torsion tensors transform to other components, and hence, those torsion tensors cannot be set to zero independently from other torsion components. The way that these boost transformations act on the torsion components
are displayed in Figure \ref{fig:galilp2}.\,\footnote{Note that when compared to the Hasse diagram of $G$-submodules presented in Section \ref{sec:without}  (see Figure~\ref{eq:coker-d-as-g-module}), the arrows that indicate the way that boosts act are oriented the opposite way. This is due to the fact that in mathematics, for a vector $v = \partial_\mu v^\mu$ it is conventional to consider the change in the \textit{basis} $\partial_\mu$, while in physics one usually considers the change of the coefficients $v^\mu$.}
Since some of the intrinsic torsion components are identically vanishing for particles ($p=0$) or domain walls ($p=D-2$), we will discuss these two special cases separately, see Figures \ref{fig:galilp4} and \ref{fig:DW} below.

\begin{figure}[ht]
    \centering
    \includegraphics[width=.45\textwidth]{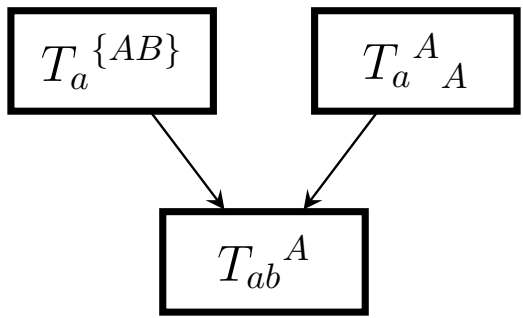}
  \caption{This figure indicates the non-vanishing intrinsic torsion components for $p\ne 0$ and $p \ne D-2$. The arrow indicates the direction in which the $p$-brane Galilean boost transformations act. For instance, the boost transformation of $T_a{}^{\{AB\}}$ gives $T_{ab}{}^A$, but not the other way around.}
  \label{fig:galilp2}
\end{figure}

Figure \ref{fig:galilp2} shows that, besides generic intrinsic torsion and zero intrinsic torsion, one may consider  three boost-invariant sets of constraints whose geometric interpretation are as follows:
\vskip .2truecm

\noindent ${\bf T_{ab}{}^A=0}$.\ \ According to the Frobenius theorem, this constraint implies  that the foliation by transverse submanifolds of dimension $D-p-1$ is integrable, i.e.~given two transverse vectors ${\bf X}$ and ${\bf Y}$ with $\tau_\mu{}^A X^\mu=\tau_\mu{}^A Y^\mu=0$,  the commutator of these two vector fields is also transverse. The proof follows from the fact the longitudinal part of the commutator is given by
\begin{equation}
\tau_\mu{}^A \bigg(X^\lambda(\partial_\lambda Y^\mu) -(\partial_\lambda X^\mu) Y^\lambda \bigg) =
\tau_\mu{}^A\bigg(X^\lambda (\nabla_\lambda Y^\mu) - (\nabla_\lambda X^\mu) Y^\lambda + 2 X^\lambda Y^\rho\Gamma^\mu_{[\lambda\rho]}  \bigg) \,.
\end{equation}
Using the Vielbein postulate \eqref{eq:VielbpostSNC}, one can pull the $\tau_\mu{}^A$ in the first two terms at the right-hand-side  through the covariant derivative after which these two terms vanish upon using the fact that ${\bf X}$ and ${\bf Y}$ are transverse vectors. In the last term we use the fact that transverse vectors only have transverse components, i.e.~$X^\mu = e_a{}^\mu X^a$ and $Y^\mu = e_a{}^\mu{} Y^a$, after which this term becomes proportional to the intrinsic torsion component $T_{ab}{}^A$ which we assume is zero. \footnote{We note that the arguments used to conclude that the different terms are zero coincide with those in Lemma \ref{lem:nabla-preserves-E} and Proposition \ref{prop:E-involutive}. }

\vskip .2truecm

\noindent ${\bf  T_a{}^{\{AB\}}} = {\bf T_{ab}{}^A=0} $.\ \  It is convenient to first consider the constraint $ T_a{}^{(AB)}=0$ without taking the traceless part:
\begin{equation}\label{combined}
T_a{}^{(AB)}= e_a{}^\mu \tau^{(A|\nu|}\big(\partial_\mu\tau_\nu{}^{B)} - \partial_\nu\tau_\mu{}^{B)}\big) = 0\,.
\end{equation}
Using the orthogonality condition $e_a{}^\mu \tau_\mu{}^A=0$ one can show that this equation is equivalent to
\begin{equation}
K_{ABa}= \tau_A{}^\mu \tau_B{}^\nu \, K_{\mu\nu a}=0\,,
\end{equation}
where $K_{\mu\nu a}=K_{\nu\mu a}$ is the Lie derivative of $\tau_{\mu\nu}$ with respect to $e_a{}^\lambda$:
\begin{equation}
K_{\mu\nu a} = e_a{}^\lambda\partial_\lambda \tau_{\mu\nu} + 2 \big(\partial_{(\mu} e_a{}^\lambda\big)\, \tau_{\lambda \nu)}\,.
\end{equation}
Using boost symmetry this leads to the following boost-invariant set of constraints\,\footnote{To show that $K_{bBa}=0$, one needs to use the boost transformation of the constraint $T_a{}^{(AB)}=0$, i.e.~$T_{ab}{}^A=0$.}
\begin{equation}
K_{ABa} = K_{bBa} = K_{bca}=0
\end{equation}
 and hence that the full Lie derivative is zero:
 \begin{equation}
 K_{\mu\nu a}=0\,.
 \end{equation}

We now consider the traceless part of
Equation~\eqref{combined}:
\begin{equation}
T_a{}^{(AB)} - \frac{1}{p+1}\eta^{AB}\eta_{CD}T_a{}^{(CD)}= 0\,.
\end{equation}
This leads to an extra term in the above derivation giving the following modified constraint
\begin{equation}
K_{\mu\nu a} =  \frac{1}{p+1}\big(\tau^{\rho\sigma}K_{\rho\sigma a}\big) \tau_{\mu\nu}\,.
\end{equation}
This shows that $e_a{}^\mu$ are conformal Killing vectors with respect to the longitudinal metric $\tau_{\mu\nu}$.

\vskip .2truecm

\noindent ${\bf T_a{}^A{}_A} = {\bf T_{ab}{}^A=0}$. To clarify the geometric implication of these constraints we first define the worldvolume $(p+1)$-form
\begin{equation}
\Omega = \epsilon_{A_1\cdots A_{p+1}}\tau_{\mu_1}{}^{A_1}\cdots \tau_{\mu_{p+1}}{}^{A_{p+1}}\,.
\end{equation}
We will now show that this worldvolume form is closed, i.e.
\begin{equation}\label{absolute}
d\Omega=0\,.
\end{equation}
Taking the exterior derivative of $\Omega$,  one finds an expression that involves the curl
$\tau_{\rho \mu_1}{}^{A_1}\equiv 2 \partial_{[\rho} \tau_{\mu_1]}{}^{A_1}$. Due to the intrinsic torsion constraint $T_{ab}{}^A=0$ and the fact that the non-intrinsic torsion component $T_{AB}{}^C$ has been set to zero to solve for some of the spin connection components, the only nonzero component is given by
\begin{equation}\label{traceless}
\tau_{\rho \mu_1}{}^{A_1} = e_\rho{}^a \tau_{\mu_1 B} T_a{}^{\{BA_1\}}\,.
\end{equation}
Writing all the flat indices on the longitudinal Vielbeine as an epsilon symbol, one now has two epsilon symbols that combine into the longitudinal Minkowski metric  as follows:
\begin{equation}
\epsilon_{A_1 A_2\cdots A_{p+1}}\epsilon_{B}{}^{A_2\cdots A_{p+1}} \sim \eta_{A_1B}\,.
\end{equation}
The resulting Minkowski metric $\eta_{A_1 B}$ projects out the trace of the traceless intrinsic torsion components in \eqref{traceless}, which is zero and hence we find that $d\Omega=0$.

Equation \eqref{absolute} is the natural generalization of the notion of absolute time for a particle to the notion of an absolute worldvolume for a $p$-brane: independent of how a $p$-brane transgresses from one transverse submanifold to another transverse submanifold the worldvolume swept out by this $p$-brane is the same. More details about this for the case of strings, i.e. $p=1$, leading to the notion of an absolute worldsheet, can be found in \cite{Bergshoeff:2022fzb}.
\vskip .3truecm

The above geometrical constraints then lead, for $p\ne 0$ and $p\ne D-2$,  to the following five distinct  $p$-brane Galilean geometries.

\begin{equation*}
\textrm{\bf Five  $p$-brane  Galilean Geometries}
\end{equation*}

\begin{description}
\item{1.} The intrinsic torsion is unconstrained.

\item{2.} ${\bf T_{ab}{}^A=0}$: the foliation by transverse submanifolds is integrable.

\item{3.} ${\bf T_a{}^{\{AB\}}= T_{ab}{}^A=0}$: the spacetime manifold is foliated by transverse submanifolds  and the vectors $e_a{}^\mu$ are conformal Killing vectors with respect to the longitudinal metric.

\item {4.} ${\bf T_a{}^A{}_A = T_{ab}{}^A=0}$: the foliation is integrable and the worldvolume is absolute.

\item{5.} ${\bf T}_{\mu\nu}{\bf {}^A=0}$: the foliation is integrable, the vectors $e_a{}^\mu$ are conformal Killing vectors with respect to the longitudinal metric and the worldvolume is absolute.
\end{description}

The cases of particles ($p=0$) and domain walls ($p=D-2$) are  special.\,\footnote{Note that for $D=2$ these two special cases coincide.} In the particle case  we find that  there are three distinct geometries (see  Figure \eqref{fig:galilp4}).

\begin{figure}[ht]
    \centering
    \includegraphics[width=.2\textwidth]{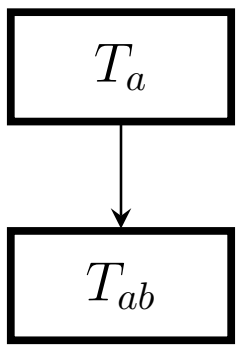}
    \caption{This figure indicates the non-zero intrinsic torsion components for $p=0$ where, with $A=0$, we have written $T_a = T_a{}^0{}_0$ and $T_{ab} = T_{ab}{}^0$.}\label{fig:galilp4}
\end{figure}

\begin{equation*}
\textrm{\bf Three particle Galilean Geometries}
\end{equation*}

\begin{description}
\item{1.} The intrinsic torsion is unconstrained.

\item{2.} ${\bf T_{ab}=0}$: the foliation by transverse submanifolds is integrable. A torsion with this constraint is called \emph{twistless torsional}. Alternatively, the foliation is called \emph{hypersurface orthogonal}. This geometry occurs in Lifshitz holography \cite{Christensen:2013lma}.

\item{3.} ${\bf T}_{\mu\nu}{\bf =0}$: the foliation is integrable and time is absolute.
\end{description}

On the other hand, in the domain wall case there are four distinct geometries (see Figure \eqref{fig:DW}).
\vskip .2truecm

\begin{figure}[ht]
    \centering
    \includegraphics[width=0.45\textwidth]{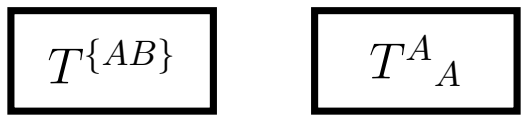}
  \caption{This figure indicates the non-zero intrinsic torsion components for $p=D-2$ where, with $a=z$, we have written $T^{\{AB\}} = T_z{}^{\{AB\}}$ and $T^A{}_A = T_z{}^A{}_A$.}\label{fig:DW}
\end{figure}

\begin{equation*}
\textrm{\bf Four domain wall  Galilean Geometries}
\end{equation*}

\begin{description}
\item{1.} The intrinsic torsion is unconstrained.

\item{2.} ${\bf T{}^{\{AB\}} =0}$: the vector $e_z{}^\mu = e^\mu$ is a  conformal Killing vector with respect to the longitudinal metric.

\item{3.} ${\bf T{}^A{}_A=0}$: the worldvolume is absolute.

\item{4.} ${\bf T{}^{(AB)}=0}$: the vector $e_z{}^\mu = e^\mu$ is a  conformal Killing vector with respect to the longitudinal metric and the worldvolume is absolute.
\end{description}

This finishes our classification of the $p$-brane Galilean  geometries. In the next section we will show how some of these geometries arise when taking special limits of general relativity.


\section{\texorpdfstring{$p$}{p}-brane Galilean Gravity}
\label{sec:galgrav}

In this section we will generalize the `particle' limit of general relativity that we considered in \cite{Bergshoeff:2017btm} to a  so-called `$p$-brane' limit such that we end up with a gravity theory with an underlying $p$-brane Galilean  geometry. This will lead to a  `$p$-brane Galilei gravity' model that for $p=0$ reduces to the Galilei gravity theory constructed in \cite{Bergshoeff:2017btm}. We will show how for general $p$ we will get a gravity realization  of several of the geometries that we found in the previous section.

Before we start we would like to emphasize that the results obtained below can also be applied to construct $p$-brane Carroll gravity theories. This is due to a formal duality between Galilei geometry and Carroll geometry from a brane point of view \cite{Barducci:2018wuj,Bergshoeff:2020xhv}. More explicitly, we will make use of the following duality:
\begin{equation}\label{duality1}
p-\textrm{brane\ Galilean geometry} \hskip .5truecm \longleftrightarrow\hskip .5truecm (D-p-2)-\textrm{brane\ Carroll\ geometry}\,.
\end{equation}
Keeping the convention that the time direction is always one of the longitudinal directions characterized by the indices $A$ and that the transverse directions are labeled by the index $a$, the duality \eqref{duality1} between Galilei and Carroll geometry implies that the corresponding Vielbein fields  are related by the following formal interchange:
\begin{equation}
(\tau_\mu{}^A\,, e_\mu{}^a)\hskip .5truecm \longleftrightarrow \hskip .5truecm (e_\mu{}^b\,, \tau_\mu{}^B)\,,
\end{equation}
where range $A$ = range $b$ and range $a$ = range $B$. A special example of this duality is that domain wall Galilean gravity  is dual to particle Carroll gravity via the formal interchange
\begin{equation}\label{duality3}
{\rm domain\ wall\ Galilean\ gravity\ with}\ (\tau_\mu{}^A, e_\mu) \leftrightarrow {\rm particle\ Carroll\ gravity\ with}\ (e_\mu{}^a\,,\tau_\mu)
\end{equation}
with range $A$ = range $a = D-1$.

A special role in the discussion below is played by branes with one or two transverse directions. To distinguish them from the generic case, we will call these special branes domain walls and defect branes, respectively. We will indicate the generic case with three or more transverse directions, i.e.~with $0 \le p \le D-4$, as $p$-branes.

To move towards $p\,$-brane Galilei gravity, we start with the $D$-dimensional Einstein-Hilbert action
\begin{align}\label{EH}
S_{\rm EH} & = - \frac{1}{16 \pi G_N} \int E E_{\hat{A}}{}^\mu E_{\hat{B}}{}^\nu{} R_{\mu\nu}{}^{\hat{A}\hat{B}}(\Omega)
\end{align}
in a first-order formulation. Here, $G_N$ is Newton's constant, $E_\mu{}^{\hat A}\ (\mu,\hat A = 0,1,\cdots,D-1)$ is the relativistic Vielbein and we have defined the inverse Vielbein $E_{\hat{A}}{}^\mu$ by
\begin{equation}
E_\mu{}^{\hat A} E_{\hat{B}}{}^\mu = \delta^{\hat A}_{\hat B}\,,\hskip 2truecm E_\mu{}^{\hat A} E_{\hat{B}}{}^\nu = \delta_\mu^\nu\,.
\end{equation}
Furthermore, the curvature $R_{\mu\nu}{}^{\hat{A}\hat{B}}(\Omega)$ is defined in terms of the relativistic spin connection field $\Omega_\mu{}^{\hat A\hat B}$ as follows:
\begin{equation}
R_{\mu\nu}{}^{\hat{A}\hat{B}}(\Omega) = 2\partial_{[\mu}\Omega_{\nu]}{}^{\hat A\hat B} -2 \Omega_{[\mu}{}^{\hat B\hat C} \Omega_{\nu]}{}^{\hat A}{}_{\hat C}\,.
\end{equation}

To define the $p$-brane Galilean limit we decompose the index $\hat A$ into a longitudinal index $A\ (A=0,1,\cdots ,p)$ and a transverse index $a\ (a = p+1, \cdots ,D-1)$, i.e. $\hat A = (A,a)$. We then redefine the Vielbeine and spin connection fields with a dimensionless contraction parameter $\omega$ as follows:
\begin{equation}
    \begin{aligned}\label{redef1}
        E_\mu{}^A &= \omega\, \tau_\mu{}^A \, \hspace{.75cm} &&\Omega_\mu{}^{Aa} \;=\; \omega^{-1} \omega_\mu{}^{Aa} \,\hspace{.75cm} &&E_\mu{}^a = e_\mu{}^a \,\\ 
        \Omega_\mu{}^{ab} \;&=\; \omega_\mu{}^{ab} \, &&\Omega_\mu{}^{AB} \;=\; \omega_\mu{}^{AB}\,.
    \end{aligned}
\end{equation}
Performing these redefinitions in the action \eqref{EH} and redefining Newton's constant in such a way that the leading power in $\omega$ is $\omega^0$, we obtain, for general values of $p$, three separate curvature terms $R_{\mu\nu}{}^{AB}(M), R_{\mu\nu}{}^{Aa}(G)$ and $R_{\mu\nu}{}^{ab}(J)$ that scale with relative powers $\omega^{-2}, \omega^{-2}$ and $\omega^0$, respectively. An exception is formed by  domain walls, i.e.~$p=D-2$, in which case the $R_{\mu\nu}{}^{ab}(J)$ curvature vanishes and, upon making an adapted redefinition of Newton's constant, we end up with two curvature terms: $R_{\mu\nu}{}^{AB}(M)$ and $R_{\mu\nu}{}^{Aa}(G)$, which scale like $\omega^0$. We will discuss the generic case of $p$-branes below first and the special cases of defect branes and domain walls next.
\vskip .3truecm

{
\noindent ${\bf p}${\bf-branes\ $(0\le p \le D-4)$.}\
In this case, after taking the limit $\omega \to \infty$,  the action takes the form
\begin{align}\label{NRaction1}
S_{p{\rm -brane}} & = - \frac{1}{16 \pi G_{\rm NL}} \int e\, e_a{}^\mu e_b{}^\nu R_{\mu\nu}{}^{ab}(J)\,,
\end{align}
where, before taking the limit, we have redefined $G_{\rm N} = \omega^{p+1} G_{\rm NL}$ and  $e=\det{(\tau_\mu{}^A\,, e_\mu{}^a)}$. The non-relativistic  curvature $R_{\mu\nu}{}^{ab}(J)$ is given by
\begin{equation}\label{R(J)}
R_{\mu\nu}{}^{ab}(J) = 2\partial_{[\mu}\omega_{\nu]}{}^{ab} - 2\omega_{[\mu}{}^{ac}\omega_{\nu]c}{}^b\,.
\end{equation}
One may verify that the action \eqref{NRaction1} is invariant under the following emerging local anisotropic scale transformations:
\begin{equation}\label{scale}
\delta \tau_\mu{}^A = -\frac{D-p-3}{p+1}\lambda(x)\,\tau_\mu{}^A\,,\hskip 2truecm \delta e_\mu{}^a = \lambda(x)\, e_\mu{}^a\,.
\end{equation}
This implies that, compared with the relativistic case, one field is lacking. Consequently, one should consider the action \eqref{NRaction1} as a pseudo-action that reproduces all non-relativistic equations of motion except for one.

Not all components of the spin connection $\omega_\mu{}^{ab}$ are determined by the equations of motion corresponding to this action. Any component that does not occur in the quadratic spin connection term given in the curvature \eqref{R(J)} becomes a Lagrange multiplier imposing a geometric constraint.  To identify these components, we decompose the curved index of the spin connection as follows:
\begin{equation}
\omega_\mu{}^{ab} = \tau_\mu{}^C \omega_C{}^{ab} + e_\mu{}^c \omega_c{}^{ab}\,.
\end{equation}
Substituting this decomposition into the action \eqref{NRaction1}, it is easy to see that the only surviving term quadratic in $\omega_\mu{}^{ab}$ is a term quadratic in the  $\omega_c{}^{ab}$ components. The other possible terms involving $\omega_A{}^{ab}$ are projected out by the inverse transverse Vielbeine in front of the curvature term in \eqref{NRaction1}. This $\omega_A{}^{ab}$ component occurs only linearly in the action and  has become a Lagrange multiplier imposing the following  geometric constraint:
\begin{equation}
T_{ab}{}^A =0\,.
\end{equation}
This is case 2 in the  classification of $p$-brane Galilean geometries given in Section \ref{sec:TNC2}.
\vskip .3truecm

\noindent {\bf{defect branes} $(p=D-3)$.}\  This case is special, because the group of transverse rotations is abelian. There are therefore no quadratic spin connection terms in the first-order Einstein-Hilbert action \eqref{EH} and the action \eqref{EH} can be written as
\begin{align}\label{defect}
S_{\rm defect \ brane} & = - \frac{1}{16 \pi G_N}\int e\Big( T_{ab}{}^C\omega_C{}^{ab} -2\, T_{aC}{}^C\omega_b{}^{ab}\Big)\,.
\end{align}
We see that all transverse rotation components of the spin connection are independent Lagrange multipliers leading to the following set of geometric constraints:
\begin{equation}
T_{ab}{}^A = T_{aA}{}^A =0\,.
\end{equation}
This corresponds to case 4 in the classification of $p$-brane Galilean geometries given in Section \ref{sec:TNC2}.  Note that in this case we have $D-p-3=0$ and hence only the transverse Vielbeine $e_\mu{}^a$ transform under the anisotropic scale transformations \eqref{scale}.
A special case amongst the defect branes is the 3$D$ particle in which case the above geometric constraints are equivalent to setting all torsion components to zero, i.e.,
\begin{equation}
T_{\mu\nu}=0\,.
\end{equation}
This is case 3 in the  classification of particle Galilean geometries given in Section \ref{sec:TNC2}.  This special case corresponds to the 3$D$ Chern-Simons gravity action discussed in \cite{Bergshoeff:2017btm}. Using the identity $e\epsilon^{ab}e_a{}^\mu e_b{}^\nu = 2\epsilon^{\mu\nu\rho}\tau_\rho$, one can show that the particle defect brane action is proportional to the following Chern-Simons action:
\begin{align}\label{defect2}
S_{\rm Chern-Simons} & = - \frac{1}{16 \pi G_N}\int \epsilon^{\mu\nu\rho} \tau_\mu\partial_\nu\omega_\rho\,,
\end{align}\noeqref{defect2}
where we have written $ \omega_\mu{}^{ab} =\epsilon^{ab}\omega_\mu$. Note that not all gauge fields corresponding to the Galilei algebra occur in this action. This is due to the fact that the Galilei algebra has  a degenerate invariant bilinear form.

\vskip .3truecm

\noindent {\bf{domain walls} $(p=D-2)$.}\ This case is special because there is only one transverse direction $a=z$ and hence $\omega_\mu{}^{ab} =0$. Writing $e_\mu\coloneqq e_\mu{}^z$ and $\omega_\mu{}^A\coloneqq \omega_\mu{}^{zA}$ we obtain,  after taking the limit $\omega \to \infty$,  the following action:
\begin{align}\label{NRaction2}
S_{\rm domain \ wall} & = - \frac{1}{16 \pi G_{\rm NL}} \int e\, \biggl(
\tau_A{}^\mu \tau_B{}^\nu R_{\mu\nu}{}^{AB}(M) + 2  e^\mu \tau_A{}^\nu  R_{\mu\nu}{}^{A}(G)\biggr)\,,
\end{align}
where, before taking the limit, we have redefined $G_{\rm N} = \omega^{p-1} G_{\rm NL}$ and where the non-relativistic  curvatures are given by
\begin{align}
R_{\mu\nu}{}^{AB}(M) &= 2\,\partial_{[\mu}\omega_{\nu]}{}^{AB}  -2\omega_{[\mu}{}^{AC}\omega_{\nu]C}{}^B\,,\\
R_{\mu\nu}{}^{A}(G) &=  2\,\partial_{[\mu}\omega_{\nu]}{}^{A} - 2\omega_{[\mu}{}^{AC}\omega_{\nu]C}\,.
\end{align}
Note that the action \eqref{NRaction2}  does not exhibit a local scale symmetry like it did for the general $p$-brane case. This implies that  this action is a true action yielding all non-relativistic equations of motion.

To see which flat spin connection components become Lagrange multipliers, we first make the following decompositions:\footnote{We use a notation where we only once omit the single spatial transverse index $z$ to indicate a boost spin connection component, i.e.~$\omega_\mu{}^{zA} = \omega_\mu{}^A$. If we further decompose the curved index $\mu$ we keep writing the $z$ index. In this way we have the simple rule that every spin connection component with three indices refers to a transverse rotation component of the spin connection, whereas each connection with only two indices refers to a boost spin connection component. \label{notation2}}
\begin{equation}\label{decomp}
\omega_\mu{}^{AB} = \tau_\mu{}^C \omega_C{}^{AB} + e_\mu \omega_z{}^{AB}\,,\hskip 1,5truecm \omega_\mu{}^A = \tau_\mu{}^C\omega_C{}^A + e_\mu\omega_z{}^A\,.
\end{equation}
Substituting these decompositions into the action \eqref{NRaction2} we see that the component $\omega_C{}^A$ can be the Lagrange multiplier for a geometric constraint provided the contraction
\begin{equation}
\omega_z{}^{AB}\omega_{AB}
\end{equation}
vanishes. Decomposing $\omega^{AB}$ further as
\begin{equation}
\omega^{AB} = \omega^{[AB]} + \omega^{\{AB\}} + \frac{1}{p+1}\eta^{AB}\omega,
\end{equation}
we conclude that the  components $\omega^{\{AB\}}$and $\omega$ are Lagrange multipliers. From the action \eqref{NRaction2}, we derive that they impose the  geometric constraints
\begin{equation}
T{}^{\{AB\}}= T{}^A{}_A=0\,.
\end{equation}
Thus the total symmetric intrinsic torsion $T^{(AB)}$ is zero, corresponding to case 4 in the classification of domain wall Galilean geometries given in Section \ref{sec:TNC2}.

The second-order formulation of the above $p$-brane Galilean gravity theories has been given in \cite{Bergshoeff:2017btm} for particle Galilei, i.e.~for $p=0$, and for particle Carroll, i.e.~upon using the duality \eqref{duality3}, for domain wall Galilei gravity. The particle case $p=0$ can easily be extended to the general $p$-brane case with $0 \le p \le D-4$. We first discuss this general case. Among the equations of motion corresponding to the action \eqref{NRaction1},  there are two conventional constraints:\,\footnote{The first conventional constraint is invariant under the anisotropic scale transformations \eqref{scale} by itself, the second term in the second conventional constraint acts like  the transverse components of a dependent dilatation gauge field.}
\begin{equation}\label{conventional1}
E_{ab}{}^c - \frac{2}{D-p-2} E_{[a|d|}{}^d \delta_{b]}^c = 0\,, \hskip 1truecm
E_{ab}{}^b + \frac{D-p-2}{D-p-3}\, T_{aC}{}^C =0\,,
\end{equation}
where $E_{\mu\nu}{}^a$ is defined by the right-hand-side of Equation~\eqref{eq:torsionVielbSNC2}. These conventional constraints can be used to solve for the transverse rotation part $\omega_\mu{}^{ab}$ of the spin connection, except for the component $\omega_C{}^{ab}$.
Due to the transverse projections of $E_{\mu\nu}{}^a$, this spin connection component does not occur in the above conventional constraints.
The solution for the components $\omega_c{}^{ab}$ is dilatation-covariant and can be written as \cite{Bergshoeff:2017btm}
\begin{align} \label{Dcovariant}
  & \omega_c{}^{ab}(e,T) = \omega_c{}^{ab}(e) - \frac{2}{D - p - 3} \,\delta_c^{[a} T^{b]}{}_C{}^C\,, \nonumber \\
  & \text{with} \ \ \ \omega_c{}^{ab}(e) \equiv e^{a \mu} e^{b \nu} \partial_{[\mu} e_{\nu] c} - 2 e_c{}^{\mu} e^{[a|\nu|} \partial_{[\mu} e_{\nu]}{}^{b]} \,.
\end{align}
We note that, after solving for $\omega_c{}^{ab}$,  the dependent spin connection components $\omega_c{}^{ab}(e,T)$ obtain extra terms in their boost transformation rule, since the conventional constraints \eqref{conventional1} are not invariant under boosts. Using the conventional constraints \eqref{conventional1}, one can rewrite the first-order action \eqref{NRaction1} in the following second-order form:
\begin{align}\label{NRaction1secondorder}
\hskip -.1truecm S^{\rm 2nd-order}_{p{\rm -brane}} &= - \frac{1}{16 \pi G_{\rm NL}} \int e\,\Big( \JR{\omega_a{}^{bc}(e,T)\,\omega_b{}^a{}_c (e,T)- \omega_a{}^{ac}(e,T)\,\omega_b{}^b{}_c (e,T)} +  T_{ab}{}^C\,\omega_C{}^{ab}\Big),
\end{align}
where we have explicitly indicated which are the dependent and which are the independent spin connection fields.

The case of defect branes is special in the sense that all components of the spin connection fields in the action \eqref{defect} occur as independent Lagrange multipliers.  Since there are no spin connection components to be solved for, one cannot consider a second-order formulation in this case.

In the case of domain walls, the conventional constraints that follow from the domain wall Galilean gravity action \eqref{NRaction2} can be found from the particle Carroll gravity case discussed in \cite{Bergshoeff:2017btm} by using the duality \eqref{duality3} between these two gravity theories. In this way one  finds the following conventional constraints:
\begin{equation}\label{conventional2}
T_z{}^{[AB]} = T_{AB}{}^C = E_{\mu\nu}=0\,.
\end{equation}
These equations can be used to solve for the longitudinal rotation components $\omega_\mu{}^{AB}$ of the spin connection and for
the boost spin connection components $\omega_\mu{}^A$, except for the component $\omega^{(AB)}$. The explicit solutions are given by 
 \begin{equation}
    \begin{aligned}\label{dependent1}
    \omega_A{}^{BC}(\tau) &= -\tau_A{}^{[BC]} + \tfrac{1}{2}\, \tau^{BC,A}  \,\\[.1truecm]
    \omega_z{}^{AB}(\tau,e) &= -e^\mu\tau_\mu{}^{[AB]}\,\\[.1truecm]
    \omega^{[AB]}(\tau,e) &= \tau^{A \mu} \tau^{B \nu} \partial_{[\mu}\, e_{\nu]}\,\\[.1truecm]
    \omega_z{}^A (\tau,e)&= 2 e^\mu \tau^{A \nu}\partial_{[\mu}\, e_{\nu]}\,,
    \end{aligned}
\end{equation}
with $\tau_{\mu\nu}{}^A \equiv 2 \partial_{[\mu} \tau_{\nu]}{}^A$. Note that the above dependent spin connection components acquire extra terms under boost transformations since the conventional constraints \eqref{conventional2} are not invariant under boosts. An exception are the (unprojected) longitudinal rotation components of the spin connection $\omega_\mu{}^{AB}(\tau)$, since they are independent of $e_\mu$ and therefore are invariant under boosts. Using the conventional constraints \eqref{conventional2}, one can rewrite the first-order action \eqref{NRaction2} in the following second-order form:
\begin{equation}
    \begin{aligned}\label{NRaction2ndorder}
        S^{\rm 2nd-order}_{{\rm domain \ wall}} & = - \frac{1}{16 \pi G_{\rm NL}} \int e\, \Big(
        \omega_A{}^B{}_C(\tau)\,\omega_B{}^{AC}(\tau) -\omega_A{}^{AC}(\tau)\,\omega_B{}^B{}_C(\tau)  \\[.1truecm]
        &\hskip 2.5truecm +2\omega_{z,AB}(\tau,e)\,\omega^{[AB]}(\tau,e) \JR{- 2 \omega_A{}^{AC}(\tau)\,\omega_{zC}(\tau,e)} \\[.1truecm]
         &\hskip 2.5truecm -2T_C{}^C\omega_A{}^A +2 T^{(AB)}\omega_{(AB)}\Big)\,,
\end{aligned}
\end{equation}
where we have explicitly indicated which spin connections components are dependent and which are independent.

It is instructive to compare the above second-order actions \eqref{NRaction1secondorder} and \eqref{NRaction2ndorder}  with the actions that one obtains by directly taking the limit of  general relativity in a second-order formulation.
This should answer the following puzzle:
\vskip .2truecm}
{\it  How do we obtain independent spin connection components acting as Lagrange multipliers if we take the non-relativistic  limit of general relativity in
a second-order formulation where all spin connection components have already been solved for?}
\vskip .2truecm

The answer to this puzzle lies in the fact that when taking a limit we need to do a different calculation than before, since in a second-order formulation the spin connection components are not redefined according to Equation~\eqref{redef1}, but according to their dependent expressions (see below). At leading order, both in the general $p$-brane case and in the special defect brane and domain wall cases, a different invariant is found that is given by  the square of a boost-invariant intrinsic torsion tensor component in each case.
At this point one can do one of three things:
\vskip .2truecm

\noindent 1.\ One can take these new invariants as a limit of general relativity in the second-order formulation by redefining Newton's constant such that the leading term scales as $\omega^0$. Following the convention in the literature, we will call these invariants `electric' gravity theories. A noteworthy feature of these  electric invariants is that they are independent of the spin connection and correspond to geometries without any geometric constraints. Due to the absence of a spin connection, these invariants do not have a first-order formulation.
\vskip .2truecm

\noindent 2. The second option is to  first tame the leading divergence by performing a Hubbard-Stratonovich transformation. One then finds an invariant at  sub-leading order.
The Hubbard-Stratonovich transformation is based upon the fact that any quadratic divergence of the form $\omega^2 X^2$ for some $X$
can be tamed by  introducing an auxiliary field $\lambda$ and rewriting the quadratic divergence in the equivalent form
\begin{equation}\label{trick}
-\frac{1}{\omega^2}\lambda^2 -2\lambda X\,.
\end{equation}
Solving for $\lambda$ and substituting this solution back, one finds the original quadratic divergence $\omega^2 X^2$. The transformation rule of this auxiliary field, before taking the limit, follows from its solution $\lambda = -\omega^2X$.
After taking the limit $\omega \to \infty$,  $\lambda$ becomes a Lagrange multiplier imposing the constraint $X=0$. The transformation  also applies if $X$ and $\lambda$ carry  longitudinal and/or transverse indices.
We will show below that in this way one finds a gravity theory in a second-order formulation that
\begin{enumerate}[label=(\roman*)]
    \item for $p$-branes is not quite the same as the second-order formulation of $p$-brane Galilean gravity that we found above,
    \item for defect branes gives an action different from the defect brane action found above, which is not truly second-order in the sense that the dependent spin connection components can be redefined away into a Lagrange multiplier, and
    \item for domain walls gives precisely the second-order formulation of the domain wall Galilean gravity theory found above.
\end{enumerate}
\noindent 3.\ A third option that we will not discuss here is to cancel the leading divergence by adding a $(p+1)$-form field to the Einstein-Hilbert term and to redefine it such that the leading divergence is cancelled. This should lead to a Newton-Cartan version of $p$-brane Galilei gravity. 
\vskip .2truecm

We continue with options 1 and 2. Before taking the limit in a second-order formulation, it is convenient to write the Einstein-Hilbert action \eqref{EH} in the following second-order form

\begin{align}\label{EH2}
S_{\rm EH} & = - \frac{1}{16 \pi G_N} \int  E\,\Big (\Omega_{\hat A}{}^{\hat B\hat C} \Omega_{\hat B}{}^{\hat A}{}_{\hat C}   - \Omega_{\hat A}{}^{\hat A\hat C}\Omega_{\hat B}{}^{\hat B}{}_{\hat C} \Big )\,,
\end{align}
where the dependent part of the spin connection is defined by the equation
\begin{equation}
\partial_{[\mu}E_{\nu]}{}^{\hat A} - \Omega_{[\mu}{}^{\hat A}{}_{\hat B} E_{\nu]}{}^{\hat B} = 0\,,
\end{equation}
whose solution is given by
\begin{equation}
\Omega_{\hat A\hat B\hat C} = \frac{1}{2} \bar E_{\hat B\hat C,\hat A}-\bar E_{\hat A[\hat B,\hat C]} \,.
\end{equation}
Here we have used the notation where
\begin{equation}
\bar E_{\mu\nu}{}^{\hat A} = 2\partial_{[\mu}E_{\nu]}{}^{\hat A},
\end{equation}
and where the comma in $\bar E_{\hat B\hat C,\hat A}$ indicates that the first two indices are anti-symmetric.

Writing $\hat A = (A,a)$ and redefining the Vielbein fields as
\begin{equation}
E_\mu{}^A = \omega\tau_\mu{}^A\,\hskip 2truecm E_\mu{}^a = e_\mu{}^a
\end{equation}
induces  expansions  of the different spin connection components. It turns out that for $p$-branes and defect branes  only the expansion of the transverse rotation components of the spin connection contributes to the action. This expansion is given by
\begin{align}
\Omega_{Cab} &= \frac{\omega}{\IR{2}}\, T_{ab,C} +  \IR{\frac{1}{\omega}(\omega_{C,ab} + \omega_{[a,b]C})(\tau,e)} \equiv \frac{\omega}{\IR{2}}\, T_{ab,C} - \frac{2}{\omega}\tau_C{}^{\mu} e_{[a|}{}^{\nu} \partial_{[\mu} e_{\nu]|b]} \, \\[.1truecm]
\Omega_{c,ab} &=  \omega_{c,ab}(e) \equiv e_a{}^{\mu} e_b{}^{\nu} \partial_{[\mu} e_{\nu] c} - 2 e_c{}^{\mu} e_{[a|}{}^{\nu} \partial_{[\mu} e_{\nu]|b]} \,\\
\Omega_{A,Ba} &=\IR{ -T_{a(A,B)} - \frac{1}{\omega^2} \omega_{[A,|a|B]}(\tau, e)} \equiv -T_{a(A,B)} - \frac{1}{\omega^2} \tau_A{}^\mu \tau_B{}^\nu \partial_{[\mu} e_{\nu] a} \,, \\
\Omega_{a,Ab} &=\IR{ \frac{\omega}{2}T_{ab,A} - \frac{1}{\omega} \omega_{(a,b)A}}(\tau, e) \equiv \frac{\omega}{2} T_{ab,A} + \frac{2}{\omega} \tau_A{}^\mu e_{(a|}{}^\nu \partial_{[\mu} e_{\nu]|b)} \,, \\
\Omega_{C,AB} &= \IR{\frac{1}{\omega} \omega_{C,AB}(\tau)} \equiv \frac{1}{\omega} \left(\tau_A{}^\mu \tau_B{}^\nu \partial_{[\mu} \tau_{\nu] C} - 2 \tau_C{}^\mu \tau_{[A|}{}^\nu \partial_{[\mu} \tau_{\nu]|B]} \right)\\
\Omega_{a,AB} &=\IR{\omega_{a,AB}(\tau, e) + \frac{1}{\omega^2} \omega_{[A,|a| B]}(\tau, e)} \equiv - 2 e_a{}^\mu \tau_{[A|}{}^\nu \partial_{[\mu} \tau_{\nu]|B]} + \frac{1}{\omega^2} \tau_A{}^\mu \tau_B{}^\nu \partial_{[\mu} e_{\nu]a} \,,
\end{align}
where $T_{ab}{}^A$ are the intrinsic torsion components of Figure \ref{fig:galilp2}.

In the case of domain walls, we only need the expansion of the longitudinal rotation and boost components of the spin connection. For this special case, we find the following expansions:
\begin{align}
\Omega_{C,AB} &= \IR{\frac{1}{\omega} \omega_{C,AB}(\tau)} \equiv \frac{1}{\omega} \left(\tau_A{}^\mu \tau_B{}^\nu \partial_{[\mu} \tau_{\nu] C} - 2 \tau_C{}^\mu \tau_{[A|}{}^\nu \partial_{[\mu} \tau_{\nu]|B]} \right) \,\label{redef2}\\[.1truecm]
\Omega_{z,AB} &= \IR{\omega_{z,AB}(\tau,e) + \frac{1}{\omega^2} \omega_{[AB]}(\tau, e)} \equiv - 2 e^\mu \tau_{[A|}{}^\nu \partial_{[\mu} \tau_{\nu]|B]} + \frac{1}{\omega^2} \tau_A{}^\mu \tau_B{}^\nu \partial_{[\mu} e_{\nu]}\,\label{redef3}\\[.1truecm]
\Omega_{AB} &=  T_{(AB)} \IR{+} \frac{1}{\omega^2} \omega_{[AB]}(\tau,e)\label{redef4} \equiv T_{(AB)} + \frac{1}{\omega^2} \tau_A{}^\mu \tau_B{}^\nu \partial_{[\mu} e_{\nu]} \,\\[.1truecm]
\Omega_{zA} &=  \frac{1}{\omega} \omega_{zA}(\tau, e) \equiv \frac{2}{\omega} e^\mu \tau_A{}^\nu \partial_{[\mu} e_{\nu]} \label{redef5} \,,
\end{align}\noeqref{redef3, redef4}
where  $T^{(AB)}$ are the intrinsic torsion components of Figure \ref{fig:galilp4}. Note that the spin connection components $\omega^{(AB)}$ do not occur in the above expressions, since, according to the previous section, these particular spin connection components are independent and therefore cannot  arise in the limit of a second-order formulation of general relativity where all spin connections have been solved for.

We now continue to compare the second-order actions for  the cases of general $p$-branes, defect branes and domain walls  with the second-order actions that we obtained before when taking the limit of general relativity in a first-order formulation.

\vskip .3truecm

\noindent {\bf $p$-branes.}\ \
Substituting the redefinitions \eqref{redef2}-\eqref{redef5} into the Einstein-Hilbert action \eqref{EH2} and following option 1 above, we redefine $G_{\rm N} = \omega^{p+3}\,G_{\rm NL}$ and obtain the following electric $p$-brane Galilei gravity action:
\begin{equation}\label{S_2}
S_{{\rm electric}\ p{\rm -brane}} =  - \frac{1}{16 \pi G_{\rm NL}} \int  \IR{\frac{e}{4}}\, T_{ab}{}^A T^{ab}{}_A\,.
\end{equation}
Since this action is quadratic in the intrinsic torsion tensor $T_{ab}{}^A$, a general class of solutions to the equations of motion corresponding to this action is given by all manifolds with $T_{ab}{}^A=0$, i.e.~all manifolds that have an integrable foliation by the transverse submanifolds.

The second option is that we redefine $G_{\rm N} = \omega^{p+1}\,G_{\rm NL}$ and tame the leading divergence by performing the
Hubbard-Stratonovich transformation \eqref{trick}, introducing a Lagrange multiplier $\lambda_C{}^{ab}$ imposing the constraint $T_{ab}{}^A=0$. In that case one finds in subleading order, which has now become the leading order, the following magnetic $p$-brane Galilei gravity action:
\begin{equation}
    \begin{aligned}\label{magneticGalilei}
        S_{{\rm magnetic}\ p{\rm -brane}} &= \IR{ - \frac{1}{16 \pi G_{\rm NL}} \int e\,\Biggl[ \omega^{a,bc} (e) \omega_{b,ac}(e) - \omega_b{}^{ba}(e) \omega_c{}^c{}_a(e) + 2T^a{}_A{}^A \omega_b{}^b{}_a}\\
        & \hspace{2.5cm} \IR{- T^a{}_A{}^A T_{aB}{}^B  + T^{a(A,B)}T_{a(A,B)}}\\
        & \hspace{2.5cm} \IR{+ \Big( \lambda^{A,ab} - (\omega^{A,[ba]} + \omega^{[b,a]A})(\tau, e) \Big) T_{ab,A} \Biggr]}.
    \end{aligned}
\end{equation}
Comparing with the previously found second-order $p\,$-brane action, we find, in fact,
\begin{equation}
    \begin{aligned}
        S_{\text{magnetic $p\,$-brane}} &= S_{\text{$p\,$-brane}}^{\text{2nd-order}}\\
        &\hspace{.25cm} \IR{- \frac{1}{16\pi G_{NL}}\int e \Biggl[ T_a{}^{\{BC\}} T^a{}_{\{BC\}} + \frac{(D-2)}{(p+1)(D-p-3)} T^a{}_A{}^A T_{aB}{}^B\Biggr]}\,,
    \end{aligned}
\end{equation}
where it is understood that in the action $S^{\rm 2nd-order}_{p\rm -brane}$, see Equation~\eqref{NRaction1secondorder}, the term  $T_{ab}{}^C\omega_C{}^{ab}$ has been replaced by
\begin{equation}
T_{ab}{}^C\Big(\lambda_C{}^{ab} + \omega_C{}^{ab}(e)\Big)\,.
\end{equation}
This answers the puzzle that we posed above: the role of the independent spin connection $\omega_C{}^{ab}$ that arises when taking the limit in a first-order formulation is taken over by (a redefinition of) the Lagrange multiplier field $\lambda_C{}^{ab}$ that arises when taking the limit in a second-order formulation after performing  a Hubbard-Stratonovich transformation.
 From the above, we see that
\begin{equation}
S_{{\rm magnetic}\ p{\rm -brane}} \ne S^{\rm 2nd-order}_{p{\rm -brane} }  \,.
\end{equation}
The additional $T_a{}^{(BC)} T^a{}_{(BC)}$ term plays a crucial role when we consider the limit in the case of domain walls, see below. Another consequence of this  term is that it breaks the anisotropic dilatation symmetry \eqref{scale}. Due to this the magnetic $p$-brane Galilei action \eqref{magneticGalilei} is, unlike the second-order $p$-brane Galilei action \eqref{NRaction1secondorder},  a true action without a missing equation of motion. One lesson we draw from this is that taking the limit in a first-order formulation and then going to a second-order formulation is not always the same as first going to a second-order formulation and then taking the limit.

\vskip .3truecm

\noindent {\bf {defect branes}\ $(p=D-3)$.} In the case of defect branes there are only independent spin connections and intrinsic tensors. The action for the electric defect brane is given by the same action \eqref{S_2} as for $p$-branes, \IR{while the action for the magnetic defect brane is given by the action \eqref{magneticGalilei}, with the only difference being that all quadratic spin connection terms vanish identically.} Because there are only independent spin connection components, there is no comparison that can be made here.


\vskip .3truecm

\noindent {\bf{domain walls}\ $(p=D-2)$.}
This case is dual to Carroll gravity, which has been discussed recently as part of a so-called Carroll expansion of general relativity \cite{Hansen:2021fxi}.
The discussion that follows here has some overlap, although in a different language, with that of \cite{Hansen:2021fxi}. A special feature of domain walls is that the intrinsic tensor $T_{ab}{}^A$ is identically zero. Similarly, the whole magnetic $p$-brane Galilean action vanishes except for the $T$-squared terms.
We can again follow options 1 and 2 for this leading term. Following option 1 and
redefining  $G_{\rm N} = \omega^{p+1}G_{\rm NL}$ one finds the following electric invariant:\footnote{Our notation for the domain wall intrinsic tensor components is explained in the caption of Figure \ref{fig:DW}.}
\begin{align}
S_{{\rm electric\ domain\ wall}}^{(1)} =
 - \frac{1}{16 \pi G_{\rm NL}} \int e\, \Biggl[ T^{(AB)} T_{(AB)} \IR{- T_A{}^A T_B{}^B} \Biggr] \,.
\end{align}
 The Carroll gravity dual to this is known in the literature as `electric Carroll gravity' and was first found in \cite{Henneaux:1979vn} in a Hamiltonian formulation, and has recently been re-investigated in \cite{Hansen:2021fxi,Henneaux:2021yzg,Campoleoni:2022ebj,Figueroa-OFarrill:2022pus}. A similar action occurs in \cite{Hartong:2015xda}.
 Again, it is independent of the spin connection and therefore it did not arise in the first-order formulation. Solutions to the equations of motion have been discussed in \cite{Hansen:2021fxi}. Using the geometric interpretation of the constraint $T^{(AB)}=0$, we find that a general solution to the equations of motion corresponding to this invariant is given by all manifolds where the vector $e_z{}^\mu= e^\mu$ is a conformal Killing vector with respect to the longitudinal metric and where the  worldvolume is absolute.

\IR{We should note that this electric limit is not unique, since we can apply the Hubbard-Stratonovich transformation to either $T$-squared term separately, i.e. 
\begin{align}
    -T_A{}^A T_B{}^B &\to -2\omega^{-2}\lambda T_A{}^A + \omega^{-4}\lambda^2,
\end{align}
or
\begin{align}
    T^{(AB)}T_{(AB)} &\to 2\omega^{-2}\lambda^{(AB)}T_{(AB)} - \omega^{-4}\lambda^{(AB)}\lambda_{(AB)}.
\end{align}
These two procedures result in two additional electric domain wall actions,
\begin{align}
    S^{(2)}_{\text{electric domain wall}} &= -\frac{1}{16\pi G_{NL}}\int eT^{(AB)}T_{(AB)}
\end{align}
and
\begin{align}
    S^{(3)}_{\text{electric domain wall}} &= -\frac{1}{16\pi G_{NL}}\int -e T_A{}^A T_B{}^B
\end{align}
respectively.}

Following option 2 we redefine $G_{\rm N} = \omega^{p-1}G_{\rm NL}$ and tame the leading divergence by performing the Hubbard-Stratonovich transformation
\eqref{trick} \IR{on both terms simultaneously, introducing a Lagrange multiplier $\lambda_{(AB)}$, that imposes the constraint $T^{(AB)}=0$}. Substituting the redefinitions \eqref{redef2}-\eqref{redef5} for $p=D-2$ into the Einstein-Hilbert action \eqref{EH2},
one ends up  with the action
\begin{align}\label{NRaction2secondorder}
S_{\rm magnetic\ domain\ wall} & = S^{\rm 2nd-order}_{{\rm domain\ wall}}\,,
\end{align}
where it is understood that $S^{\rm 2nd-order}_{{\rm domain\ wall}}$ is the action given in Equation~\eqref{NRaction2ndorder}, but with the terms $-T_C{}^C \omega_A{}^A + T^{(AB)}\omega_{(AB)}$ replaced by
\begin{equation}
T^{(AB)}\, \lambda_{(AB)}\,.
\end{equation}
The Carroll gravity dual to the domain wall Galilei gravity with a magnetic invariant \eqref{NRaction2secondorder} is known in the literature as `magnetic Carroll gravity'. It was obtained in a Hamiltonian formulation in \cite{Henneaux:2021yzg} and has recently been discussed in \cite{Campoleoni:2022ebj,Hansen:2021fxi,deBoer:2023fnj}.


\section{Conclusions} \label{sec:conclusions}

In this paper we have classified $p$-brane Galilean and Carrollian
geometries via the intrinsic torsion of the corresponding
$G$-structures. We presented our results both in the Cartan-theoretic
language of contemporary mathematics and in a formulation much more
familiar to theoretical physicists. We hope that we contributed in
this way to an improved communication between these two communities.

Our classification is incomplete for the cases of stringy
Galilean and Carrollian geometries, since some of the representations
appearing in the space of intrinsic torsions are further reducible: in
the stringy case (and independent on dimension), one can use a lightcone
frame for the two longitudinal directions \cite{Bergshoeff:2021tfn},
and these transform independently under boosts. We hope to extend our
classification to these cases as well elsewhere.

To extend the $p$-brane Galilean geometry studied in this work to a
$p$-brane Newton-Cartan geometry underlying non-relativistic
$p$-branes the following three ingredients need to be added:
\begin{enumerate}
\item The geometry underlying non-relativistic $p$-branes contains an
  additional $(p+1)$-form $b_{\mu_1 \cdots \mu_{p+1}}$, whose
  mathematical description seems to require, according to some of the
  literature, the introduction of the notion of gerbes. It transforms
  under $p$-brane boost transformations and plays an important role in
  describing the geometry.\,\footnote{In the case of membranes in D=11
    dimensions, one needs to add a 3-form $c_{\mu\nu\rho}$ with a
    4-form curvature which allows a self-duality condition with
    respect to the transverse rotation group SO(8) of the type
    discussed above.}
\item The frame fields of $p$-brane  Newton-Cartan geometry transform
  under an emergent anisotropic local scale symmetry
  \cite{Bergshoeff:2021bmc}, like the one we found for $p$-brane
  Galilean gravity, see Equation~\eqref{scale}. This extra gauge symmetry
  requires an additional dilatation gauge field $b_\mu$ beyond the
  spin connection. It has the effect that one should define new
  dilatation-covariant tensor components
  \begin{equation}
    \tilde T{}_a{}^A{}_A = T_a{}^A{}_A  - b_a\,.
  \end{equation}
This has the effect that the new torsion tensor components $\tilde
T{}_a{}^A{}_A$ are conventional instead of intrinsic. Setting these
components to zero leads to an equation that  can be used to solve for
the transverse components $b_a$ of the dilatation gauge field in terms
of $T_a{}^A{}_A$.
\item In the presence of supersymmetry, the $p$-brane Newton-Cartan
  geometry  needs to be embedded into a so-called supergeometry which
  also contains fermionic intrinsic torsion tensor components. To
  enable this embedding, we need to impose constraints on some of the
  intrinsic torsion tensor components.
\end{enumerate}

Writing Carroll geometry as a special domain-wall case of $p$-brane
Galilean geometry allows one to export techniques from Galilean to
Carroll geometry. An  example of this is the construction of a Carroll
fermion as the limit of a relativistic fermion. Carroll fermions have
been discussed in \cite{Bagchi:2022owq,Bagchi:2022eui}. In the
Galilean case, to obtain a domain-wall Galilean fermion from a
relativistic fermion $\psi$, one first redefines \cite{Gomis:2004pw}
\begin{equation}\label{redefpsi}
  \psi = \sqrt{\omega}\, \psi_+ + \frac{1}{\sqrt{\omega}}\, \psi_-\,,
\end{equation}
where the projected spinors are defined by
\begin{equation}
  \psi_{\pm} = \frac{1}{2}\big(1 \pm \Gamma_{01\cdots D-2}\big)\psi,
\end{equation}
and then one takes the limit $\omega\to\infty$.  Using the
duality between Galilean and Carroll geometry given in
Equation~\eqref{duality1}, this suggests that one can obtain a Carroll
fermion from a relativistic fermion $\psi$ by redefining $\psi$ as in
\eqref{redefpsi}, with the projected spinors given by
\begin{equation}
  \psi_{\pm} = \frac{1}{2}\big(1 \pm \Gamma_{12\cdots D-1}\big)\psi\,,
\end{equation}
where $12\cdots D-1$ refer to the $D-1$ directions transverse to the
Carroll particle. It would be interesting to compare these Carroll
fermions with the ones introduced in
\cite{Bagchi:2022owq,Bagchi:2022eui}.

It would be interesting to consider a Carroll limit of general
relativity in $D$ spacetime dimensions \`a la Newton--Cartan using the
third option discussed in the last section by writing it as a Galilean
domain wall. This would require the addition of a $(D-1)$-form field
that is dual to a cosmological constant. Such a cosmological constant
arises in the 10$D$ massive Romans supergravity theory
\cite{Romans:1985tz} that has been formulated in terms of a 9-form
potential \cite{Bergshoeff:1996ui}. This would suggest a limit defined
by the redefinition
\begin{equation}
  C_{\mu_1\cdots \mu_9} = \epsilon_{A_1 \cdots A_9} \tau_{\mu_1}{}^{A_1}\cdots \tau_{\mu_9}{}^{A_9} + c_{\mu_1\cdots \mu_9}\,.
\end{equation}
It would be interesting to see whether the Carroll theory dual to this could
lead to a 10$D$ Carrollian supergravity theory.

Finally, we have only considered the (intrinsic) torsion of the
$G$-structures and, as has been amply demonstrated in
(pseudo-)Riemannian geometry, the classification of holonomy groups of adapted
connections becomes a very convenient further organising principle for
the geometries.  Not much is known beyond the strictly Riemannian and
Lorentzian geometries when it comes to the possible holonomy
representations on a pseudo-Riemannian manifold.  The problem in the
case of geometries with nontrivial intrinsic torsion seems to be
harder still, but one worthwhile of investigation.  Is there a
Berger-like classification for adapted connections to a
Galilean or Carrollian structure?


\section*{Acknowledgments}
\label{sec:acknowledgments}

This work was started during a visit of JMF to Groningen in December
2022 and he would like to thank EAB for the invitation and, as well as
the other Groningen-based members of this collaboration, for the
hospitality.

KVH is funded by the Fundamentals of the Universe program at the
University of Groningen.


%

\printbibliography[heading=bibintoc]

\end{document}